%% file: tbs-mm.tex
\pdfoutput=1 
\documentclass[sigconf]{acmart-edbt2018}

\setcopyright{rightsretained}

\usepackage{xspace}
\usepackage{graphicx}
\usepackage{subfigure}
\usepackage{balance}  % for  \balance command ON LAST PAGE  (only there!)
\usepackage{mathtools}
\usepackage{color}

\usepackage{amsmath,amssymb}
\usepackage[linesnumbered,vlined, boxed, ruled]{algorithm2e}
\usepackage{caption}
\DeclareCaptionType{copyrightbox}
\usepackage{epsfig}
\usepackage{epstopdf,url}

\usepackage[T1]{fontenc}
\usepackage[utf8]{inputenc}
\usepackage[english]{babel}
\usepackage[font=small,labelfont=bf]{caption}
\usepackage{capt-of}% or \usepackage{caption}
\usepackage{booktabs}
\usepackage{varwidth}
\newsavebox\tmpbox

\newtheorem{remark}{Remark}%[theorem]

\DeclarePairedDelimiter\ceil{\lceil}{\rceil}
\DeclarePairedDelimiter\floor{\lfloor}{\rfloor}

\SetKwInOut{Input}{input}\SetKwInOut{Output}{output}
\SetKwComment{Comment}{//}{}
\SetEndCharOfAlgoLine{}
\SetKw{True}{true}
\SetKw{False}{false}
\SetKw{Not}{not}

\DeclareMathOperator{\var}{Var}

\DeclareMathOperator{\mean}{E}
\DeclareMathOperator{\prb}{Pr}
\newcommand{\prob}[1]{\prb[#1]}

\newcommand{\xB}{\mathcal{B}}
\let\eps\epsilon
\DeclareMathOperator{\frc}{frac}

\newcommand{\BigCrunch}{\vspace*{-1.5em}}

\newcommand{\SmallCrunch}{\vspace*{-1ex}}

\newcommand{\eat}[1]{}
\newcommand{\longPaper}{\xspace} %this is for long paper

\begin{document}

\title[Temporally-Biased Sampling]{Temporally-Biased Sampling for Online Model Management}

%\author{
%	{\fontsize{11.5}{11.5}\selectfont
%		Brian Hentschel{\small $^\star$}, Peter J. Haas{\small $^\S$ $^1$}, Yuanyuan Tian{\small $^{\dagger}$}}\\
%	\fontsize{9.5}{9.5}\selectfont\rmfamily\itshape
%	{\small $^\star$}Harvard University\ \ \fontsize{7}{7}\selectfont\ttfamily\upshape
%	bhentschel@g.harvard.edu
%	\fontsize{9.5}{9.5}\selectfont\rmfamily\itshape
%	{\small $^\S$}University of Massachusetts \ \
%	\fontsize{7}{7}\selectfont\ttfamily\upshape
%	phaas@cs.umass.edu
%	\fontsize{9.5}{9.5}\selectfont\rmfamily\itshape
%	{\small $^{\dagger}$}IBM Research\ \
%	\fontsize{7}{7}\selectfont\ttfamily\upshape
%	ytian@us.ibm.com
%}

\author{Brian Hentschel}
\authornote{Work performed at IBM Research -- Almaden}
\affiliation{%
  \institution{Harvard University}
%  \city{Cambridge} 
%  \state{Massachusetts} 
%  \country{USA}
}
\email{bhentschel@g.harvard.edu}

\author{Peter J. Haas\raisebox{0pt}{\small{*}}}%\small{*}}}
\affiliation{%
  \institution{University of Massachusetts}
%  \city{Amherst} 
%  \state{Massachusetts} 
%  \country{USA}
}
\email{phaas@cs.umass.edu}

\author{Yuanyuan Tian}
\affiliation{%
  \institution{IBM Research -- Almaden}
%  \city{San Jose} 
%  \state{California} 
%  \country{USA}
}
\email{ytian@us.ibm.com}

% The following code works for vldb.cls but not for acmart.cls
%\numberofauthors{2}
%\author{
%\alignauthor Brian Hentschel\titlenote{Work performed at IBM Research -- Almaden.}\\
%       \affaddr{Harvard University\\Cambridge, Massachusetts, USA}
%       \email{bhentschel@g.harvard.edu}
%\alignauthor Peter J.~Haas \hspace{0.5cm} Yuanyuan Tian\\
%       \affaddr{IBM Research -- Almaden\\San Jose, California, USA}
%       \email{\{phaas,ytian\}@us.ibm.com, }
%}

\input{abstract} % abstract must now precede \maketitle in new amcart template

\maketitle

\input{intro}
\SmallCrunch
\input{background}
\input{ttbs}

\input{rtbs3}
\input{implmt_ver3}

%\SmallCrunch
%\input{expmt}
\SmallCrunch
\input{expmt_copy}

\input{relwork}
\SmallCrunch
\input{concl}

%\section{Acknowledgments}
%\small
%\bibliographystyle{ACM-Reference-Format}
%\bibliography{tbs-mm}  

\bibliographystyle{ACM-Reference-Format}
\bibliography{tbs-mm}

%\newpage
\begin{appendix}
\input{th1prf}

\input{spark-implementation}

\input{appendix-expmts}
\end{appendix}

%\bibliographystylelatex{ACM-Reference-Format}
%\bibliographylatex{tbs-mm}

\end{document}

%% file: abstract.tex
% !TEX root = tbs-mm.tex
% above command is for TeXShop

\begin{abstract}
To maintain the accuracy of supervised learning models in the presence of evolving data streams, we provide temporally-biased sampling schemes that weight recent data most heavily, with inclusion probabilities for a given data item decaying exponentially over time. We then periodically retrain the models on the current sample. This approach speeds up the training process relative to training on all of the data. Moreover, time-biasing lets the models adapt to recent changes in the data while---unlike in a sliding-window approach---still keeping some old data to ensure  robustness in the face of temporary fluctuations and periodicities in the data values. In addition, the sampling-based approach allows existing analytic algorithms for static data to be applied to dynamic streaming data essentially without change. We provide and analyze both a simple sampling scheme (T-TBS) that probabilistically maintains a target sample size and a novel reservoir-based scheme (R-TBS) that is the first to provide both complete control over the decay rate and a guaranteed upper bound on the sample size, while maximizing both expected sample size and sample-size stability. The latter scheme rests on the notion of a ``fractional sample'' and, unlike T-TBS, allows for data arrival rates that are unknown and time varying. R-TBS and T-TBS are of independent interest, extending the known set of unequal-probability sampling schemes.
%We show how to implement efficient distributed versions of T-TBS and R-TBS on Spark, leveraging a recent ``in-place updating'' technique for RDDs due to Xie, et al.. 
We discuss distributed implementation strategies; experiments in Spark illuminate the performance and scalability of the algorithms, and show that our approach can increase machine learning robustness in the face of evolving data. 

\end{abstract}

%% file: intro.tex
% !TEX root = tbs-mm.tex
% above command is for TeXShop

\section{Introduction}\label{sec:intro}

A key challenge for machine learning (ML) is to keep ML models from becoming stale in the presence of evolving data. In the context of the emerging Internet of Things (IoT), for example, the data comprises dynamically changing sensor streams~\cite{WhitmoreAX15}, and a failure to adapt to changing data can lead to a loss of predictive power.

One way to deal with this problem is to re-eng\-i\-neer existing static supervised learning algorithms to become adaptive. Some parametric algorithms such as SVM can indeed be re-engineered so that the parameters are time-varying, but for non-parametric algorithms such as kNN-based classification, it is not at all clear how re-engineering can be accomplished. We therefore consider alternative approaches in which we periodically retrain ML models, allowing static ML algorithms to be used in dynamic settings essentially as-is. There are several possible retraining approaches.

\textbf{Retraining on cumulative data:} Periodically retraining a model on all of the data that has arrived so far is clearly infeasible because of the huge volume of data involved. Moreover, recent data is swamped by the massive amount of past data, so the retrained model is not sufficiently adaptive.

\textbf{Sliding windows:} A simple sliding-window approach would be to, e.g., periodically retrain on the data from the last two hours. If the data arrival rate is high and there is no bound on memory, then one must deal with long retraining times caused by large amounts of data in the window. The simplest way to bound the window size is to retain the last $n$ items. Alternatively, one could try to subsample within the time-based window~\cite{GemullaL08}. The fundamental problem with all of these bounding approaches is that old data is completely forgotten; the problem is especially severe when the data arrival rate is high. This can undermine the robustness of an ML model in situations where old patterns can reassert themselves. For example, a singular event such as a holiday, stock market drop, or terrorist attack can temporarily disrupt normal data patterns, which will reestablish themselves once the effect of the event dies down. Periodic data patterns can lead to the same phenomenon. Another example, from \cite{XieTSBH15}, concerns influencers on Twitter: a prolific tweeter might temporarily stop tweeting due to travel, illness, or some other reason, and hence be completely forgotten in a sliding-window approach. Indeed, in real-world Twitter data, almost a quarter of top influencers were of this type, and were missed by a sliding window approach.

\eat{A simple sliding-window approach---e.g., periodically retrain on the data from the last two hours---has two drawbacks. First, if the data-arrival rate is high and memory is limited, then one must choose a very short window length, making the retrained model myopic; if one does not bound the amount of data in a window, then one must deal with long retraining times caused by large data volumes. If the arrival rate varies in an unpredictable manner, then choosing a window length is very challenging. Perhaps more importantly, old data is completely forgotten under a sliding window approach. This can undermine the robustness of an ML model in situations where old patterns can reassert themselves. For example, a singular event such as a holiday, stock market drop, or terrorist attack can temporarily disrupt normal data patterns, which will reestablish themselves once the effect of the event dies down. Periodic data patterns can lead to the same phenomenon. Another example, from \cite{XieTSBH15}, concerns influencers on Twitter: a prolific tweeter might temporarily stop tweeting because of travel, illness, or some other reason, and hence be completely forgotten in a sliding-window approach. Indeed, in real-world Twitter data, almost a quarter of top influencers were of this type, and were missed by a sliding window approach.}

\textbf{Temporally biased sampling:} An appealing alternative is a temporally biased sampling-based approach, i.e., maintaining a sample that heavily emphasizes recent data but also contains a small amount of older data, and periodically retraining a model on the sample. By using a time-biased sample, the retraining costs can be held to an acceptable level while not sacrificing robustness in the presence of recurrent patterns. This approach was proposed in \cite{XieTSBH15} in the setting of graph analysis algorithms, and has recently been adopted in the MacroBase system~\cite{BailisGMNRS17}. The orthogonal problem of choosing when to retrain a model is also an important question, and is related to, e.g.,  the literature on ``concept drift''~\cite{GamaZBPB14}; in this paper we focus on the problem of how to efficiently maintain a time-biased sample.

In more detail, our time-biased sampling algorithms ensure that the ``appearance probability'' for a given data item---i.e., the probability that the item appears in the current sample---decays over time at a controlled exponential rate. Specifically, we assume that items arrive in batches (see the next section for more details), and our goal is to ensure that (i) our sample is representative in that all items in a given batch are equally likely to be in the sample, and (ii) if items $i$ and $j$ belong to batches that have arrived at (wall clock) times $t'$ and $t''$ with $t'\le t''$, then for any time $t\ge t''$ our sample $S_t$ is such that
\begin{equation}\label{eq:expratio}
\prob{i\in S_t}/\prob{j\in S_t}=e^{-\lambda(t''-t')}.
\end{equation}
Thus items with a given timestamp are sampled uniformly, and items with different timestamps are handled in a carefully controlled manner. The criterion in \eqref{eq:expratio} is natural and appealing in applications and, importantly, is interpretable and understandable to users. As discussed in \cite{XieTSBH15}, the value of the decay rate $\lambda$ can be chosen to meet application-specific criteria. For example, by setting $\lambda=0.058$, around 10\% of the data items from 40 batches ago are included in the current analysis. As another example, suppose that, $k=150$ batches ago, an entity such as a person or city was represented by $n=1000$ data items  and we want to ensure that, with probability $q=0.01$,  at least one of these data items remains in the current sample. Then we would set $\lambda=-k^{-1}\ln\bigl(1-(1-q)^{1/n}\bigr)\approx 0.077$. If training data is available, $\lambda$ can also be chosen to maximize accuracy via cross validation.

The exponential form of the decay function has been adopted by the majority of time-biased-sampling applications in practice because otherwise one would typically need to track the arrival time of every data item---both in and outside of the sample---and decay each item individually at an update, which would make the sampling operation intolerably slow. (A ``forward decay" approach that avoids this difficulty, but with its own costs, has been proposed in \cite{CormodeSSX09}; we plan to investigate forward decay in future work.) Exponential decay functions make update operations fast and simple.

For the case in which the item-arrival rate is high, the main issue is to keep the sample size from becoming too large. On the other hand, when the incoming batches become very small or widely spaced, the sample sizes for all of the time-biased algorithms that we discuss (as well as for sliding-window schemes based on wall-clock time) can become small. This is a natural consequence of treating recent items as more important, and is characteristic of any sampling scheme that satisfies \eqref{eq:expratio}. We emphasize that---as shown in our experiments---a smaller, but carefully time-biased sample typically yields greater prediction accuracy than a sample that is larger due to overloading with too much recent data or too much old data. I.e., more sample data is not always better. Indeed, with respect to model management, this decay property can be viewed as a feature in that, if the data stream dries up and the sample decays to a very small size, then this is a signal that there is not enough new data to reliably retrain the model, and that the current version should be kept for now.

It is surprisingly hard to both enforce \eqref{eq:expratio} and to bound the sample size. As discussed in detail in Section~\ref{sec:relwork}, prior algorithms that bound the sample size either cannot consistently enforce \eqref{eq:expratio} or cannot handle wall-clock time. Examples of the former include algorithms based on the A-Res scheme of Efraimidis and Spirakis~\cite{efraimidisS06}, and Chao's algorithm~\cite{Chao82}. A-Res enforces conditions on the \emph{acceptance} probabilities of items; this leads to appearance probabilities which, unlike \eqref{eq:expratio}, are both hard to compute and not intuitive. A similar example is provided by Chao's algorithm~\cite{Chao82}. In Appendix~\ref{sec:chao}\longPaper we demonstrate how the algorithm can be specialized to the case of exponential decay and modified to handle batch arrivals. We then show that the resulting algorithm fails to enforce \eqref{eq:expratio} either when initially filling up an empty sample or in the presence of data that arrives slowly relative to the decay rate, and hence fails if the data rate fluctuates too much. The second type of algorithm, due to Aggarwal~\cite{Aggarwal06} can only control appearance probabilities based on the indices of the data items. For example, after $n$ items arrive, one could require that, with 95\% probability, the $(n-k)$th item should still be in the sample for some specified $k<n$. If the data arrival rate is constant, then this might correspond to a constraint of the form ``with 95\% probability a data item that arrived 10 hours ago is still in the sample'', which is often more natural in applications. For varying arrival rates, however, it is impossible to enforce the latter type of constraint, and a large batch of arriving data can prematurely flush out older data. Thus our new sampling schemes are interesting in their own right, significantly expanding the set of unequal-probability sampling techniques.

\eat{The literature on time-biased sampling, which is relatively sparse, is reviewed in Section~\ref{sec:relwork}.  In brief, prior algorithms that provide precise control over the decay rate either cannot handle wall-clock time or impose stringent assumptions on the data arrival process. The first type of algorithm, due to Aggarwal~\cite{Aggarwal06} can only control inclusion probabilities based on the indices of the data items. For example, after $n$ items arrive, one could require that, with 95\% probability, the $(n-k)$th item should still be in the sample for some specified $k<n$. If the data arrival rate is constant, then this might correspond to a constraint of the form ``with 95\% probability a data item that arrived 10 hours ago is still in the sample'', which is often more natural in applications. For varying arrival rates, however, it is impossible to enforce the latter type of constraint, and a large batch of arriving data can prematurely flush out older data. The other type of algorithm, due to Chao~\cite{Chao82}, cannot precisely enforce exponentially decaying sample inclusion probabilities either when initially filling up an empty sample or in the presence of data that arrives slowly relative to the decay rate. Thus our new sampling schemes are interesting in their own right, significantly expanding the set of unequal-probability sampling techniques.}

\textbf{T-TBS:} We first provide and analyze Target\-ed-Size Time-Biased Sampling (T-TBS), a simple algorithm that generalizes the sampling scheme in \cite{XieTSBH15}. T-TBS allows complete control over the decay rate (expressed in wall-clock time) and probabilistically maintains a target sample size. That is, the expected and average sample sizes converge to the target and the probability of large deviations from the target decreases exponentially or faster in both the target size and the deviation size. T-TBS is simple and highly scalable when applicable, but only works under the strong restriction that the mean data arrival rate is known and constant. There are scenarios where T-TBS might be a good choice (see Section~\ref{sec:ttbs}), but many applications have non-constant, unknown mean arrival rates or cannot tolerate sample overflows.

\textbf{R-TBS:} We then provide a novel algorithm, Reservoir-Based Time-Biased Sampling (R-TBS), that is the first to simultaneously enforce \eqref{eq:expratio} at all times, provide a guaranteed upper bound on the sample size, and allow unknown, varying data arrival rates. Guaranteed bounds are desirable because they avoid memory management issues associated with sample overflows, especially when large numbers of samples are being maintained---so that the probability of \emph{some} sample overflowing is high---or when sampling is being performed in a limited memory setting such as at the ``edge'' of the IoT. Also, bounded samples reduce variability in retraining times and do not impose upper limits on the incoming data flow.

The idea behind R-TBS is to adapt the classic reservoir sampling algorithm, which bounds the sample size but does not allow time biasing. Our approach rests on the notion of a ``fractional'' sample whose nonnegative size is real-valued in an appropriate sense. We show that, over all sampling algorithms having exponential decay, R-TBS maximizes the expected sample size whenever the data arrival rate is low and also minimizes the sample-size variability. 

\textbf{Distributed implementation:} Both T-TBS and R-TBS can be parallelized. Whereas T-TBS is relatively straightforward to implement, an efficient distributed implementation of R-TBS is nontrivial. We exploit various implementation strategies to reduce I/O relative to other approaches, avoid unnecessary concurrency control, and make decentralized decisions about which items to insert into, or delete from, the reservoir.

%We leverage a recent in-place updating technique for Spark RDDs due to Xie et al.~\cite{XieTSBH15}; in our setting, this technique allows us to exploit native Spark operations, reduce I/O relative to other approaches, avoid unnecessary concurrency control, and make decentralized decisions about which items to insert into, or delete from, the reservoir.

\textbf{Organization:} The rest of the paper is organized as follows. In Section~\ref{sec:background} we formally describe our batch-arrival problem setting and discuss two prior simple sampling schemes: a simple Bernoulli scheme as in~\cite{XieTSBH15} and the classical reservoir sampling scheme, modified for batch arrivals.  These methods either bound the sample size but do not control the decay rate, or control the decay rate but not the sample size. We next present and analyze the T-TBS and R-TBS algorithms in Section~\ref{sec:ttbs} and Section~\ref{sec:tbsamp}. We describe the distributed implementation in Section~\ref{sec:imp}, and Section~\ref{sec:exp} contains experimental results. We review the related literature in Section~\ref{sec:relwork} and conclude in Section~\ref{sec:concl}.

%% file: background.tex
% !TEX root = tbs-mm.tex
% above command is for TeXShop

\section{Setting and Prior Schemes}\label{sec:background}

After introducing our problem setting, we discuss two prior sampling schemes that provide context for our current work: simple Bernoulli time-biased sampling (B-TBS) with no sample-size control and the classical reservoir sampling algorithm (with no time biasing), modified for batch arrivals (B-RS).  

\textbf{Setting:} %Assume that 
Items arrive in \emph{batches} $\xB_1,\xB_2,\ldots$, at time points $t=1,2,\ldots$, where each batch contains 0 or more items.  This simple integer batch sequence often arises from the discretization of time~\cite{QianHSWZZZYZ13,ZahariaDLSS13}. Specifically, the continuous time domain is partitioned into intervals of length $\Delta$, and the items are observed only at times $\{k\Delta:k=0,1,2,\ldots\}$. All items that arrive in an interval $\bigl[k\Delta,(k+1)\Delta\bigr)$ are treated as if they arrived at time $k\Delta$, i.e., at the start of the interval, so that all items in batch $\xB_i$ have time stamp~$i\Delta$, or simply time stamp~$i$ if time is measured in units of length $\Delta$. As discussed below, our results can straightforwardly be extended to arbitrary real-valued batch-arrival times.

Our goal is to generate a sequence $\{S_t\}_{t\ge 0}$, where $S_t$ is a sample of the items that have arrived at or prior to time~$t$, i.e., a sample of the items in $U_t=S_0\cup\bigl(\bigcup_{i=1}^t \xB_i\bigr)$. Here we allow the initial sample $S_0$ to start out nonempty. These samples should be biased towards recent items so as to enforce \eqref{eq:expratio} for $i\in\xB_{t'}$ and $j\in\xB_{t''}$ while keeping the sample size as close as possible to (and preferably never exceeding) a specified target~$n$.

Our assumption that batches arrive at integer time points can easily be dropped. In all of our algorithms, inclusion probabilities---and, as discussed later, closely related item ``weights''---are updated at a batch arrival time $t'$ with respect to their values at the previous time~$t=t'-1$ via multiplication by $e^{-\lambda}$. To extend our algorithms to handle arbitrary successive batch arrival times $t$ and $t'$, we simply multiply instead by $e^{-\lambda(t'-t)}$. Thus our results can be applied to arbitrary sequences of real-valued batch arrival times, and hence to an arbitrary sequences of item arrivals (since batches can comprise single items). 
%Also, as discussed subsequently, each of our algorithms can be extended to handle out-of-order arrivals by simply modifying the initial acceptance probabilities or weights.

%\begin{algorithm}[h]
%\caption{Bernoulli time-biased sampling (B-TBS)}\label{alg:bernsamp}
%{\footnotesize
%$\lambda$: decay factor ($\ge 0$)
%\BlankLine
%Initialize: $S\gets S_0$; $p\gets e^{-\lambda}$\Comment*[r]{$p=$ retention prob.}
%\For{$t\gets1,2,\ldots$}{
%$M \gets \textsc{Binomial}(|S|,p)$\Comment*[r]{simulate $|S|$ trials}\label{ln:binom}
%$S \gets \textsc{Sample}(S,M)$\Comment*[r]{retain $M$ random elements}
%$S\gets S\cup \xB_t$\;\label{ln:accept}
%output $S$}
%}
%\end{algorithm}

\textbf{Bernoulli Time-Biased Sampling (B-TBS):} In the simplest sampling scheme, at each time~$t$, we accept each incoming item $x\in \xB_t$ into the sample with probability~1. At each subsequent time $t'>t$, we flip a coin independently for each item currently in the sample: an item is retained in the sample with probability~$p=e^{-\lambda}$ and removed with probability~$1-p$. It is straightforward to adapt the algorithm to batch arrivals; see Appendix~\ref{sec:BTBS}\longPaper, where we show that $\prob{x\in S_{t'}}= e^{-\lambda(t'-t)}$ for $x\in\xB_t$, implying \eqref{eq:expratio}. This is essentially the algorithm used, e.g., in \cite{XieTSBH15} to implement time-biased edge sampling in dynamic graphs. The user, however, cannot independently control the expected sample size, which is completely determined by $\lambda$ and the sizes of the incoming batches. In particular, if the batch sizes systematically grow over time, then sample size will grow without bound. Arguments in \cite{XieTSBH15} show that if $\sup_t |\xB_t|<\infty$, then the sample size can be bounded, but only probabilistically. See Remark~\ref{rem:alg1} below for extensions and refinements of these results.

\textbf{Batched Reservoir Sampling (B-RS):} The classic reservoir sampling algorithm can be modified to handle batch arrivals; see Appendix~\ref{sec:BRS}\longPaper. Although B-RS guarantees an upper bound on the sample size, it does not support time biasing. The R-TBS algorithm (Section~\ref{sec:tbsamp}) maintains a bounded reservoir as in B-RS while simultaneously allowing time-biased sampling.

%% file: ttbs.tex
% !TEX root = tbs-mm.tex
% above command is for TeXShop

\section{Targeted-Size TBS}\label{sec:ttbs}

As a first step towards time-biased sampling with a controlled sample size, we describe the simple T-TBS scheme, which improves upon the simple Bernoulli sampling scheme B-TBS by ensuring the inclusion property in \eqref{eq:expratio} while providing probabilistic guarantees on the sample size. We require that the mean batch size equals a constant $b$ that is both known in advance and ``large enough'' in that $b\ge n(1-e^{-\lambda})$, where $n$ is the target sample size and $\lambda$ is the decay rate as before. The requirement on $b$ ensures that, at the target sample size, items arrive on average at least as fast as they decay.

\begin{algorithm}[ht]
\caption{Targeted-size TBS (T-TBS)}\label{alg:targsamp}
{\footnotesize
$\lambda$: decay factor ($\ge 0$)\;
$n$: target sample size\;
$b$: assumed mean batch size such that $b\ge n(1-e^{-\lambda})$\;
\BlankLine
Initialize: $S\gets S_0$; $p \gets e^{-\lambda}$; $q\gets n(1-e^{-\lambda})/b$\;
\For{$t\gets1,2,\ldots$}{
$m \gets \textsc{Binomial}(|S|,p)$\Comment*[r]{simulate $|S|$ trials}\label{ln:binoma}
$S \gets \textsc{Sample}(S,m)$\Comment*[r]{retain $m$ random elements}
$k \gets \textsc{Binomial}(|\xB_t|,q)$\;\label{ln:binomb}
$B'_t \gets \textsc{Sample}(\xB_t,k)$\Comment*[r]{down-sample new batch}
$S\gets S\cup B'_t$\;\label{ln:taccept}
output $S$}
}
\end{algorithm}

The pseudocode is given as Algorithm~\ref{alg:targsamp}. T-TBS is similar to B-TBS in that we downsample by performing a coin flip for each item with retention probability $p$. Unlike B-TBS, we downsample the incoming batches at rate $q=n(1-e^{-\lambda})/b$, which ensures that $n$ becomes the ``equilibrium'' sample size. Specifically, when the sample size equals $n$, the expected number $n(1-e^{-\lambda})$ of current items deleted at an update equals the expected number $qb$ of inserted new items, which causes the sample size to drift towards $n$. Arguing similarly to Appendix~\ref{sec:BTBS}\longPaper, we have for $t'\ge t\ge 1$ and $x\in \xB_t$ that $\prob{x\in S_{t'}} = qe^{-\lambda(t'-t)}$, so that the key relative appearance property in \eqref{eq:expratio} holds.

For efficiency, the algorithm exploits the fact that for $k$ independent trials, each having success probability~$r$, the total number of successes has a binomial distribution with parameters $k$ and $r$. Thus, in lines~\ref{ln:binoma} and \ref{ln:binomb}, the algorithm simulates the coin tosses by directly generating the number of successes $m$ or $k$---which can be done using standard algorithms \cite{binomgen88}---and then retaining $m$ or $k$ randomly chosen items. So the function $\textsc{Binomial}(j,r)$ returns a random sample from the binomial distribution with $j$ independent trials and success probability $r$ per trial, and the function $\textsc{Sample}(A,m)$ returns a uniform random sample, without replacement, containing $\min(m,|A|)$ elements of the set $A$; note that the function call $\textsc{Sample}(A,0)$ returns an empty sample for any empty or nonempty $A$.

%An out-of-order item $i\in\xB_t$ arriving at time $t'>t$ is handled by accepting it into the sample with probability $qe^{-\lambda(t'-t)}$.

\begin{figure*}[tbh]
 \centering
	\subfigure[Growing Batch Size]{
	   \label{fig:grow}\includegraphics[width=0.23\linewidth]{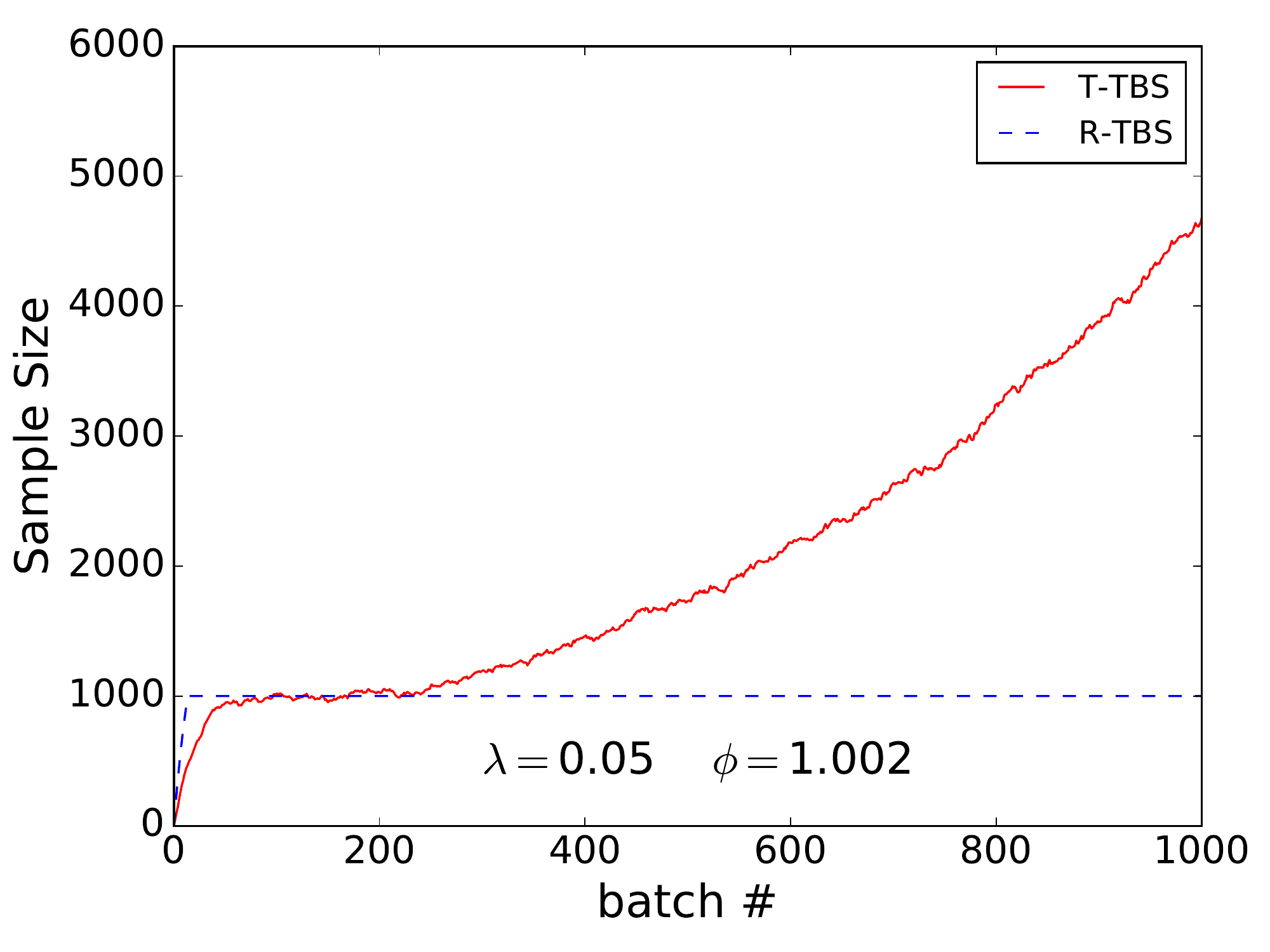}} 
	\subfigure[Stable Batch Size (Det.)]{
	   \label{fig:stableD}\includegraphics[width=0.23\linewidth]{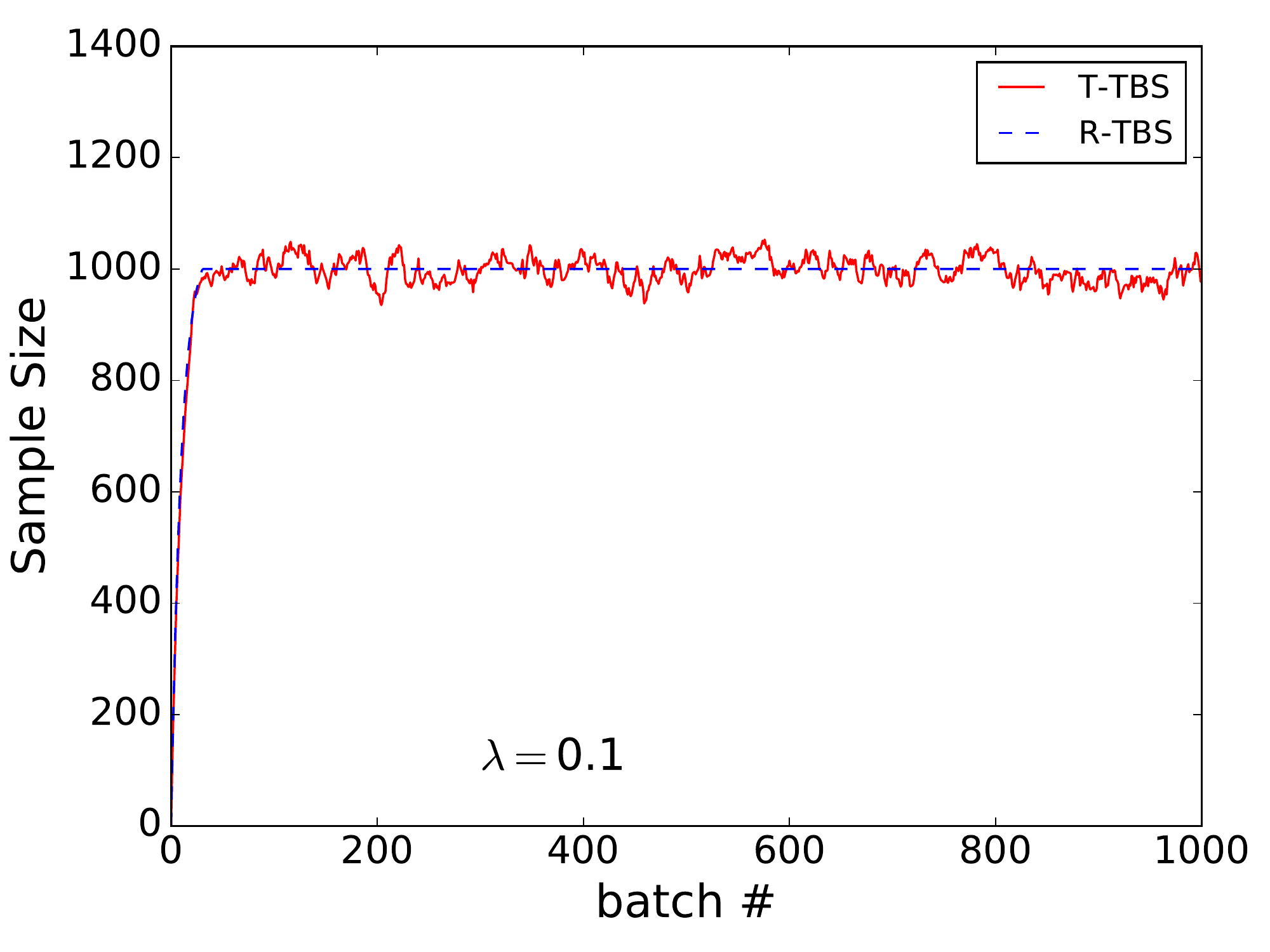}} 
	 \subfigure[Stable Batch Size (Unif.)]{
	   \label{fig:stableU}\includegraphics[width=0.23\linewidth]{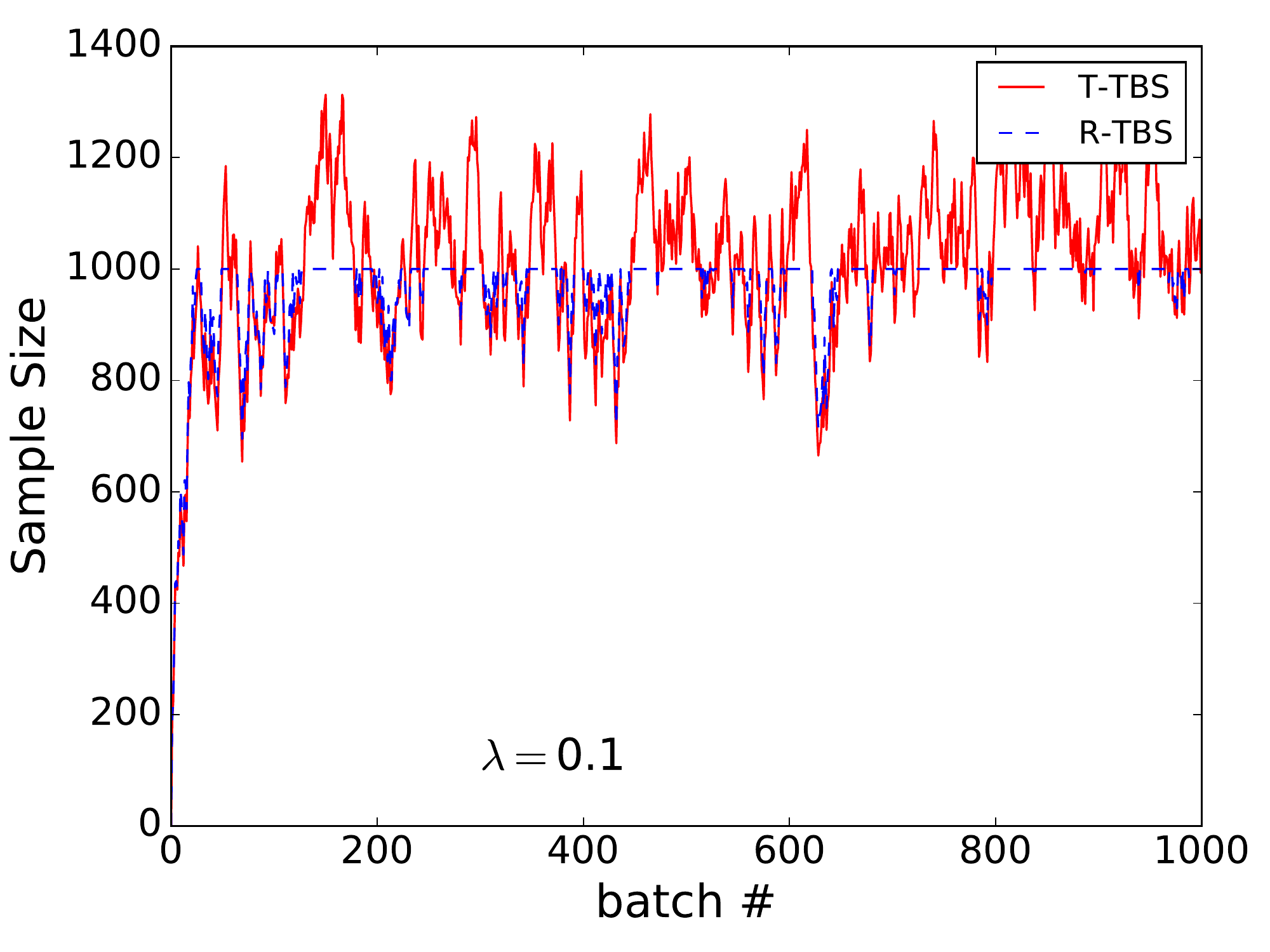}}
%	 \subfigure[Decaying Batch Size: 1]{
%	   \label{fig:decay1}\includegraphics[width=0.23\linewidth]{figs/decay1}}
	 \subfigure[Decaying Batch Size]{
	   \label{fig:decay2}\includegraphics[width=0.23\linewidth]{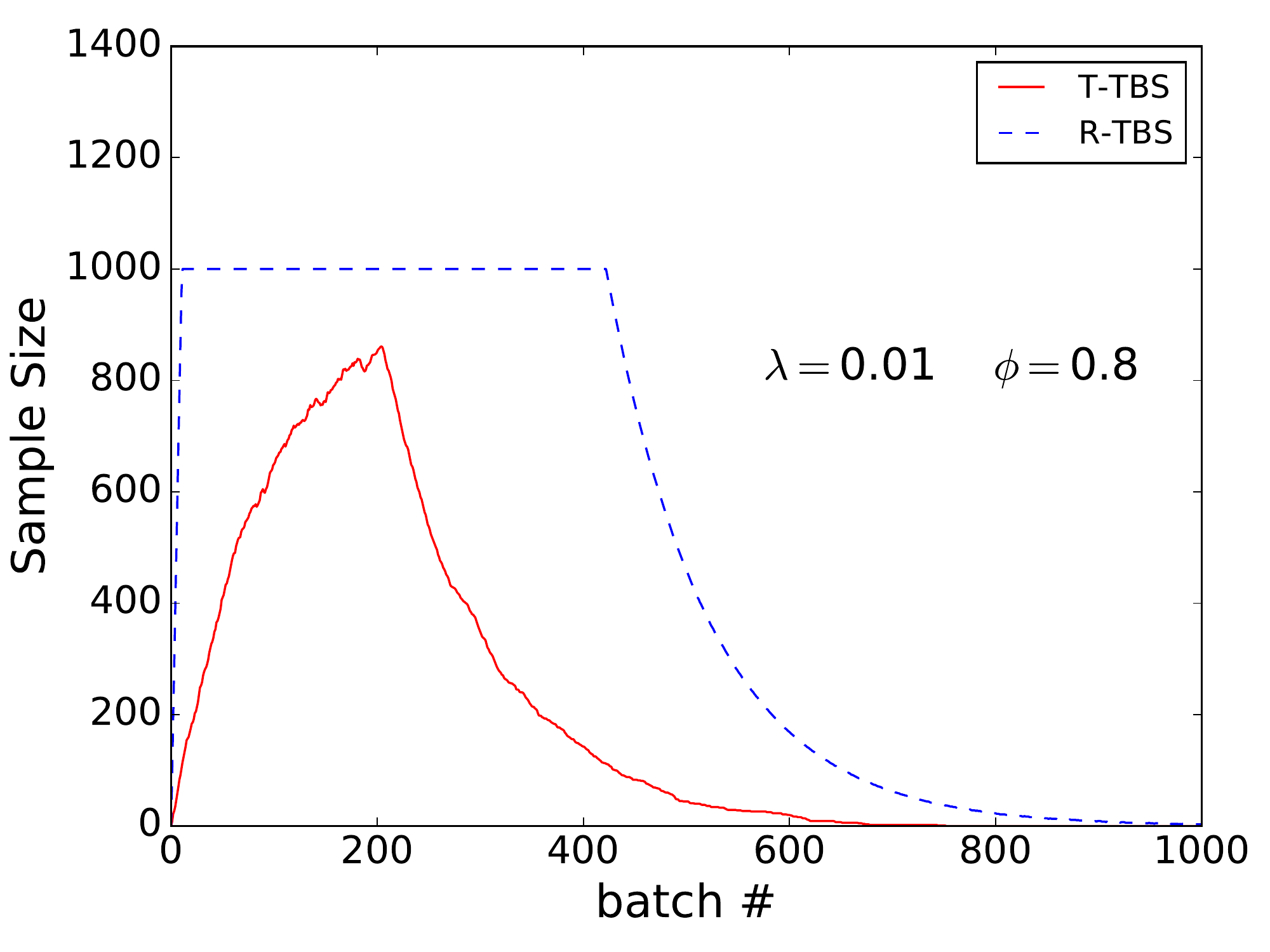}}
%	 \subfigure[Decaying Batch Size: 3]{
%	   \label{fig:decay3}\includegraphics[width=0.23\linewidth]{figs/decay3}}
%	   	 \subfigure[Decaying Batch Size: 4]{
%	   \label{fig:decay4}\includegraphics[width=0.23\linewidth]{figs/decay4}}
\BigCrunch
\caption{\label{fig:breakdown}Targeted TBS: Sample Size Behavior, $\lambda=$ decay rate and $\phi=$ batch size multiplier.}
\BigCrunch
\end{figure*}

Theorem~\ref{th:recurr} below precisely describes the behavior of the sample size; the proof---along with the proofs of most other results in the paper---is given in Appendix~\ref{sec:proofs}\longPaper. Denote by $B_t=|\xB_t|$ the (possibly random) size of $\xB_t$ for $t\ge 1$ and by $C_t=|S_t|$ the sample size at time~$t$ for $t\ge 0$; assume that $C_0$ is a finite deterministic constant. Define the \emph{upper-support ratio} for a random batch size $B$  as $r=b^*/b\ge 1$, where $b=\mean[B]$ and $b^*$ is the smallest positive number such that $P[B\le b^*]=1$; set $r=\infty$ if $B$ can be arbitrarily large. For $r\in[1,\infty)$, set
\[
\nu^+_{\eps,r}=(1+\eps)\ln\bigl((1+\eps)/r\bigr)-(1+\eps-r).
\]
for $\eps>0$ and
\[
\nu^-_{\eps,r} = (1-\eps)\ln\bigl((1-\eps)/r\bigr)-(1-\eps-r)
\]
for $\eps\in(0,1)$. Note that $\nu^+_{\eps,r}>0$ and is strictly increasing in $\eps$ for $\eps>r-1$, and that $\nu^-_{\eps,r}$ increases from $r-1-\ln r$ to $r$ as $\eps$ increases from 0 to 1. Write ``i.o.'' to denote that an event occurs ``infinitely often'', i.e., for infinitely many values of $t$, and write ``w.p.1'' for ``with probability~1''. 
\begin{theorem}\label{th:recurr}
Suppose that the batch sizes $\{B_t\}_{t\ge 1}$  are i.i.d\ with common mean $b\ge n(1-e^{-\lambda})$, finite variance, and upper support ratio $r$. Then, for any $p=e^{-\lambda}<1$,
\begin{enumerate}\parskip=0pt
\item[(i)] for all $m\ge 0$, we have $\prob{C_t=m\text{ i.o.}}=1$;
\item[(ii)] $\mean[C_t]=n+p^t(C_0-n)$ for $t>0$;
\item[(iii)] $\lim_{t\to\infty}(1/t)\sum_{i=0}^t C_i=n$ w.p.1;
%\item[(iv)] $\var[C_t]=\alpha n+\sigma^2_Bq^2/(1-p^2)+O(p^t)$ for $t>0$;
\item[(iv)] if $C_0=n$ and $r<\infty$, then
\begin{enumerate}
\item[(a)] $\prob{C_t\ge (1+\eps)n}\le e^{-n\nu^+_{\eps,r}}\bigl(1+O(n\eps p^t)\bigr)$%$e^{g^+_t(n,\eps)}
and
\item[(b)] $\prob{C_t \le (1-\eps)n}\le e^{-n\nu^-_{\eps,r}}\Bigl(1+O\bigl(n(1-\eps )p^t\bigr)\Bigr)$%e^{g^-_t(n,\eps)}
\end{enumerate}
for (a) $\eps,t>0$ and (b) $\eps\in(0,1)$ and $t\ge \ln\eps/\ln p$.
\end{enumerate}
\end{theorem}
%In the deterministic case where $\sigma^2=0$, so that $B_t\equiv b$ for $t\ge 1$, the assertion in (iv)(b) together with \cite[Lemma~1]{HagerupR90} imply that 
%\[
%\prob{C_t \le (1-\eps)n}\le e^{-n\eps^2/2}\Bigl(1+O\bigl(n(1-\eps )p^t\bigr)\Bigr).
%\]
In Appendix~\ref{sec:proofs}\longPaper, we actually prove a stronger version of the theorem in which the assumption in (iv) that $r<\infty$ is dropped.

Thus, from (ii),  $\lim_{t\to\infty}\mean[C_t]=n$ so that the expected sample size converges to the target size $n$  as $t$ becomes large; indeed, if $C_0=n$ then the expected sample size equals $n$ for all $t>0$. By (iii), an even stronger property holds in that, w.p.1, the average sample size---averaged over the first~$t$ batch-arrival times---converges to~$n$ as $t$~becomes large. For typical batch-size distributions, the assertions in (iv) imply that, at any given time $t$, the probability that the sample size deviates from $n$ by more than $100\eps\%$ decreases exponentially with $n$ and---in the case of a positive deviation as in (iv)(a)---super-exponentially in $\eps$. However, the assertion in (i) implies that any sample size~$m$, no matter how large, will be exceeded infinitely often w.p.1; indeed, it follows from the proof that the mean times between successive exceedances are not only finite, but are uniformly bounded over time. In summary, the sample size is generally stable and close to~$n$ on average, but is subject to infrequent, but unboundedly large spikes in the sample size, so that sample-size control is incomplete.

Indeed, when batch sizes fluctuate in a non-predicable way, as often happens in practice, T-TBS can break down; see Figure~\ref{fig:breakdown}, in which we plot sample sizes for T-TBS and, for comparison, R-TBS. The problem is that the value of the mean batch size $b$ must be specified in advance, so that the algorithm cannot handle dynamic changes in $b$ without losing control of either the decay rate or the sample size.

In Figure~\ref{fig:grow}, for example, the (deterministic) batch size is initially fixed and the algorithm is tuned to a target sample size of 1000, with a decay rate of $\lambda=0.05$. At $t=200$, the batch size starts to increase (with $B_{t+1}=\phi B_t$ where $\phi=1.002$), leading to an overflowing sample, whereas R-TBS maintains a constant sample size.

Even in a stable batch-size regime with constant batch sizes (or, more generally, small variations in batch size), R-TBS can maintain a constant sample size whereas the sample size under T-TBS fluctuates in accordance with Theorem~\ref{th:recurr}; see Figure~\ref{fig:stableD} for the case of  a constant batch size $B_t\equiv100$ with $\lambda=0.1$.

Large variations in the batch size lead to large fluctuations in the sample size for T-TBS; in this case the sample size for R-TBS is bounded above by design, but large drops in the batch size can cause drops in the sample size for both algorithms; see Figure~\ref{fig:stableU} for the case of $\lambda=0.1$ and i.i.d.\ uniformly distributed batch sizes on $[0,200]$ so that $\mean[B_t]\equiv 100$. Similarly, as shown in Figure~\ref{fig:decay2}, systematically decreasing batch sizes will cause the sample size to shrink for both T-TBS and R-TBS. Here, $\lambda=0.01$ and, as with Figure~\ref{fig:grow}, the batch size is initially fixed and then starts to change at time $t=200$, with $\phi=0.8$ in this case. This experiment---and others, not reported here, with varying values of $\lambda$ and $\phi$---indicate that R-TBS is more robust to sample underflows than T-TBS.

Overall, however, T-TBS is of interest because, when the mean batch size is known and constant over time, and when some sample overflows are tolerable, T-TBS is simple to implement and parallelize, and is very fast (see Section~\ref{sec:exp}). For example, if the data comes from periodic polling of a set of robust sensors, the data arrival rate will be known a priori and will be relatively constant, except for the occasional sensor failure, and hence T-TBS might be appropriate. On the other hand, if data is coming from, e.g., a social network, then batch sizes may be hard to predict.

\begin{remark}\label{rem:alg1}\normalfont
When $q=1$, Theorem~\ref{th:recurr} provides a description of sample-size behavior for  B-TBS. Under the conditions of the theorem, the expected sample size converges to $n=b/(1-e^{-\lambda})$, which illustrates that the sample size and decay rate cannot be controlled independently. The actual sample size fluctuates around this value, with large deviations above or below being exponentially or super-ex\-po\-nen\-tially rare. Thus Theorem~\ref{th:recurr} both complements and refines the analysis in \cite{XieTSBH15}.
\end{remark}

%% file: rtbs3.tex
% !TEX root = tbs-mm.tex
% above command is for TeXShop

\section{Reservoir-Based TBS}\label{sec:tbsamp}

Targeted time-biased sampling (T-TBS) controls the decay rate but only partially controls the sample size, whereas batched reservoir sampling (B-RS) bounds the sample size but does not allow time biasing. Our new reservoir-based time-biased sampling algorithm (R-TBS) combines the best features of both, controlling the decay rate while ensuring that the sample never overflows and has optimal sample size and stability properties. Importantly, unlike T-TBS, the R-TBS algorithm can handle any sequence of batch sizes.

%does not require any assumptions about the sequence of batch sizes.

\subsection{The R-TBS Algorithm}

\begin{algorithm}[t]
\caption{Reservoir-based TBS (R-TBS)}\label{alg:rtbs}
{\footnotesize
$\lambda$: decay factor ($\ge 0$)\;
$n$: maximum sample size\;
\BlankLine
Initialize: $A\gets A_0$; $W\gets C\gets |A_0|$;  $\pi\gets\emptyset$\Comment*[r]{$|A_0|\le n$}
\For{$t\gets1,2,\ldots$}{
\eIf(\Comment*[f]{has been unsaturated}){$W< n$\label{ln:ifunsat}}{
    $W\gets e^{-\lambda}W$\Comment*[r]{decay current items}\label{ln:decayW}
    \If{$W>0$}{$(A,\pi,C)\gets\textsc{Dsample}\bigl((A,\pi,C),W\bigr)$}\label{ln:dsample}
    $A\gets A\cup\xB_t$\Comment*[r]{accept all items in $\xB_t$}\label{ln:uUpdateA}
    $W\gets W+|\xB_t|$\Comment*[r]{update total weight}\label{ln:uUpdateW}
    \If(\Comment*[f]{sample is now saturated}){$W> n$}{
    \Comment*[l]{adjust for overshoot}
    $(A,\pi,C)\gets\textsc{Dsample}\bigl((A,\pi,W),n\bigr)$}\label{ln:overshoot}
    }(\Comment*[f]{has been saturated}){
    $W\gets e^{-\lambda}W+|\xB_t|$\Comment*[r]{new total weight}\label{ln:newWeight}
    \eIf(\Comment*[f]{still saturated}){$W\ge n$}{
        $m\gets\textsc{StochRound}(|\xB_t|n/W)$\;\label{ln:bround}
        \Comment{replace $m$ $A$-items with $m$ $\xB_t$-items}
        $A\gets A\setminus\textsc{Sample}(A,m)\cup\textsc{Sample}(\xB_t,m)$\;\label{ln:sUpdateA}
        }(\Comment*[f]{now unsaturated}){
        \Comment*[l]{adjust for undershoot}
        $(A,\pi,C)\gets\textsc{Dsample}\bigl((A,\pi,n),W-|\xB_t|\bigr)$\;\label{ln:undershoot}
        $A\gets A\cup\xB_t$\Comment*[r]{all batch items are full}\label{ln:fillup}
        }
    }
$S\gets\textsc{getSample}(A,\pi,C)$\;
output $S$
}
}
\end{algorithm}

To maintain a bounded sample, R-TBS combines the use of a reservoir with the notion of item \emph{weights}. In R-TBS, the weight of an item initially equals~1 but then decays at rate $\lambda$, i.e., the weight of an item $i\in\xB_t$ at time $t'\ge t$ is $w_{t'}(i)=e^{-\lambda(t'-t)}$. All items arriving at the same time have the same weight, so that the \emph{total weight} of all items seen up through time~$t$ is $W_t=\sum_{j=1}^t B_je^{-\lambda(t-j)}$, where, as before, $B_j=|\xB_j|$ is the size of the $j$th batch.

\begin{figure}[tbh]
       \centering
       \includegraphics[width=0.4\linewidth]{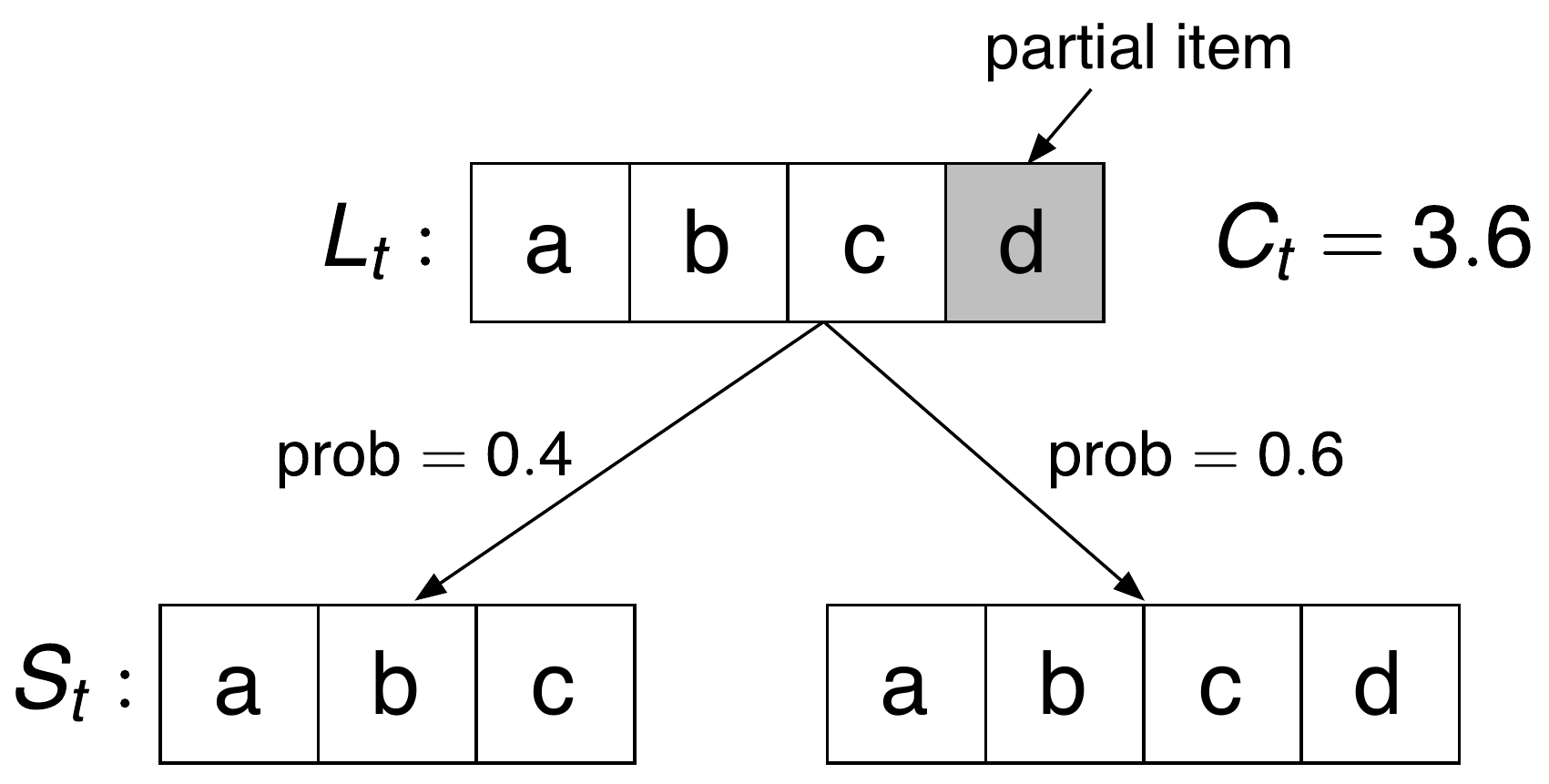}
       \SmallCrunch
       \caption{Latent sample $L_t$ (sample weight $C_t=3.6$) and possible realized samples.}
       \label{fig:latentsamp}
       \BigCrunch
\end{figure}

R-TBS generates a sequence of latent ``fractional samples'' $\{L_t\}_{t\ge 0}$ such that (i) the ``size'' of each $L_t$ equals the \emph{sample weight} $C_t$, defined as $C_t=\min(n,W_t)$, and (ii) $L_t$ contains $\floor{C_t}$ ``full'' items and at most one ``partial'' item. For example, a latent sample of size $C_t=3.6$ contains three ``full'' items that belong to the actual sample $S_t$ with probability~1 and one partial item that belongs to $S_t$ with probability~0.6. Thus $S_t$ is obtained by including each full item and then including the partial item according to its associated probability, so that $C_t$ represents the expected size of $S_t$. E.g., in our example, the sample $S_t$ will contain either three or four items with respective probabilities 0.4 and 0.6, so that the expected sample size is 3.6; see Figure~\ref{fig:latentsamp}. Note that if $C_t=k$ for some $k\in\{0,1,\ldots,n\}$, then with probability~1 the sample contains precisely $k$ items, and $C_t$ is the actual size of $S_t$, rather than just the expected size. Since each $C_t$ by definition never exceeds $n$, no sample $S_t$ ever contains more than $n$ items.

More precisely, given a set $U$ of items, a \emph{latent sample} of $U$ with sample weight $C$ is a triple $L=(A,\pi,C)$, where $A\subseteq U$ is a set of $\floor{C}$ \emph{full} items and $\pi\subseteq U$ is a (possibly empty) set containing at most one \emph{partial} item.
%With a slight abuse of notation, we write ``$i\in L$'' for an item $i$ to denote ``$i\in A\cup\pi$'' and ``$|L|$'' to denote ``$|A\cup\pi|$''.
At each time~$t$, we randomly generate $S_t$ from $L_t=(A_t,\pi_t,C_t)$ by sampling such that
\begin{equation}\label{eq:getsample}
S_t=
\begin{cases}
A_t\cup\pi&\text{with probability $\frc(C_t)$};\\
A_t&\text{with probability $1-\frc(C_t)$},
\end{cases}
\end{equation}
where $\frc(x)=x-\floor{x}$. That is, each full item is included with probability~1 and the partial item is included with probability $\frc(C_t)$. Thus
\begin{equation}\label{eq:meansize}
\begin{split}
&\mean[|S_t|]=\ceil{C_t}\frc(C_t)+\floor{C_t}\bigl(1-\frc(C_t)\bigr)\\
&\quad=(\ceil{C_t}-\floor{C_t})\frc(C_t)+\floor{C_t}\\
&\quad=\frc(C_t)+\floor{C_t}=C_t
\end{split}
\end{equation}
as previously asserted. By allowing at most one partial item, we minimize the latent sample's footprint: $|A_t\cup\pi_t|\le \floor{C_t}+1$.

\begin{figure*}[tbh]
        \centering
	\subfigure[Unsat. $\rightarrow$ Unsat.]{
	   \label{fig:rtbsScenario1}\includegraphics[width=0.22\linewidth]{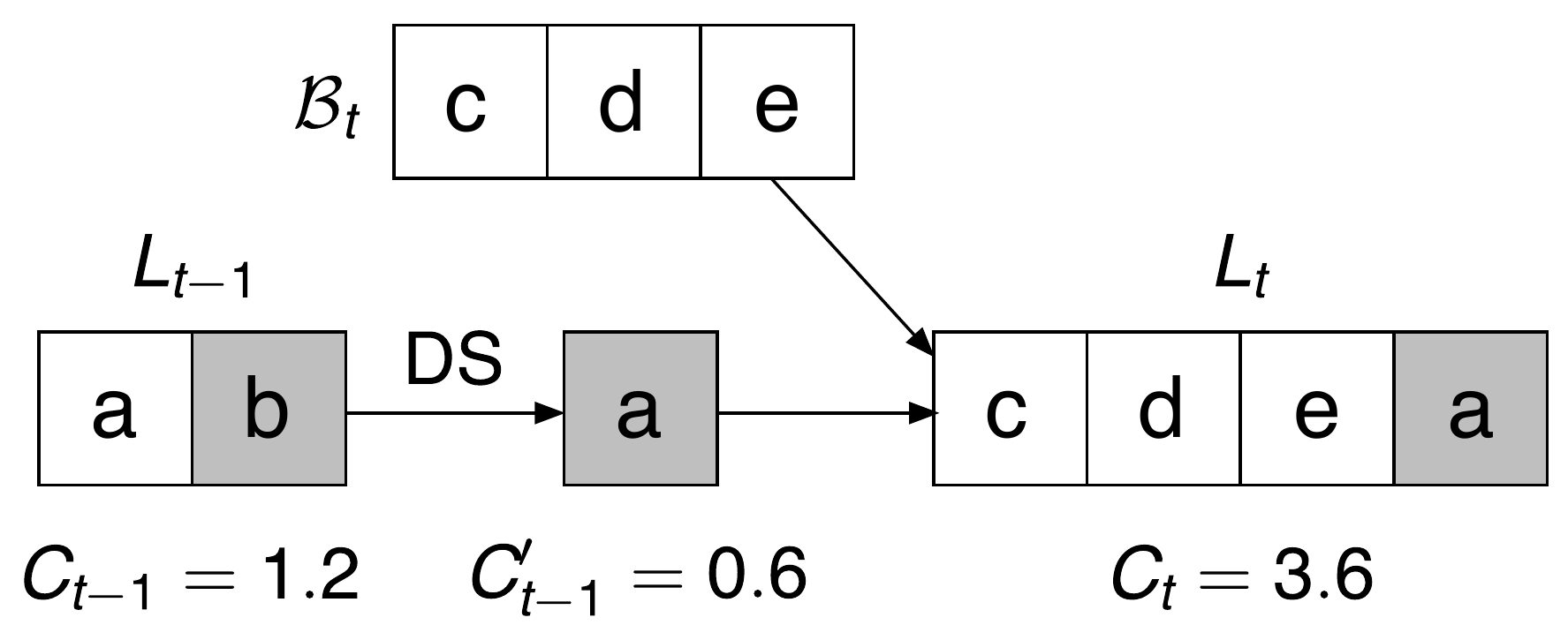}} \hspace{0.00\textwidth}
	\subfigure[Unsat. $\rightarrow$ Sat.]{
	   \label{fig:rtbsScenario2}\includegraphics[width=0.26\linewidth]{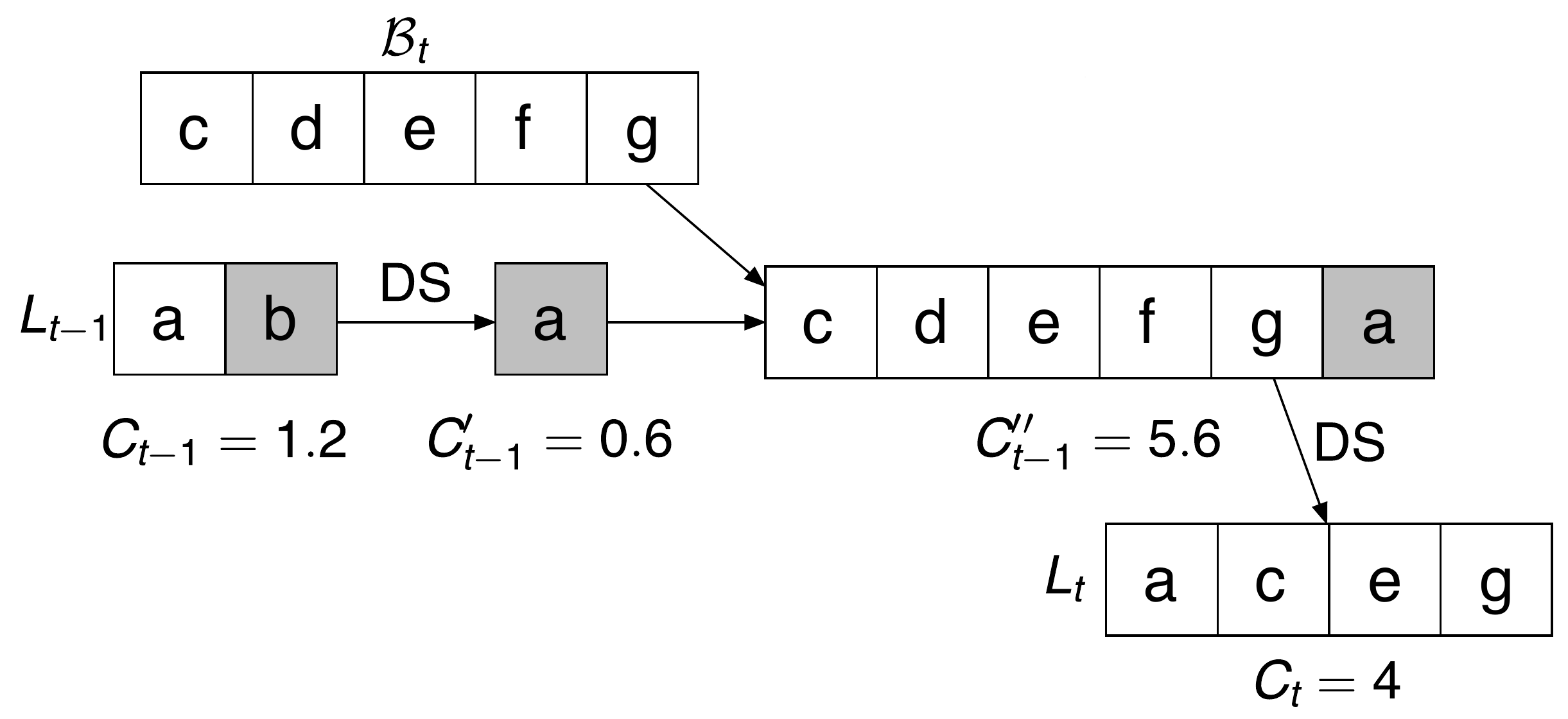}} \hspace{0.00\textwidth}
	\subfigure[Sat. $\rightarrow$ Unsat.]{
	   \label{fig:rtbsScenario3}\includegraphics[width=0.26\linewidth]{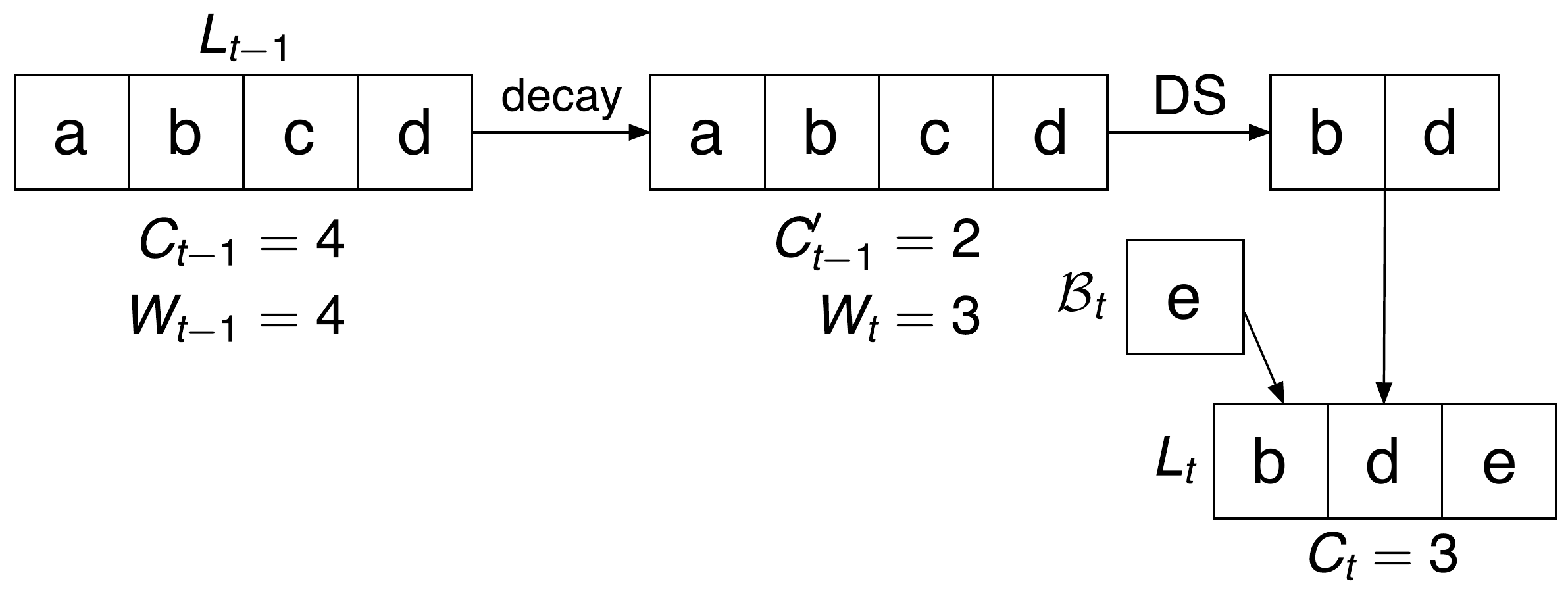}} \hspace{0.00\textwidth}
	\subfigure[Sat. $\rightarrow$ Sat.]{
	   \label{fig:rtbsScenario4}\includegraphics[width=0.21\linewidth]{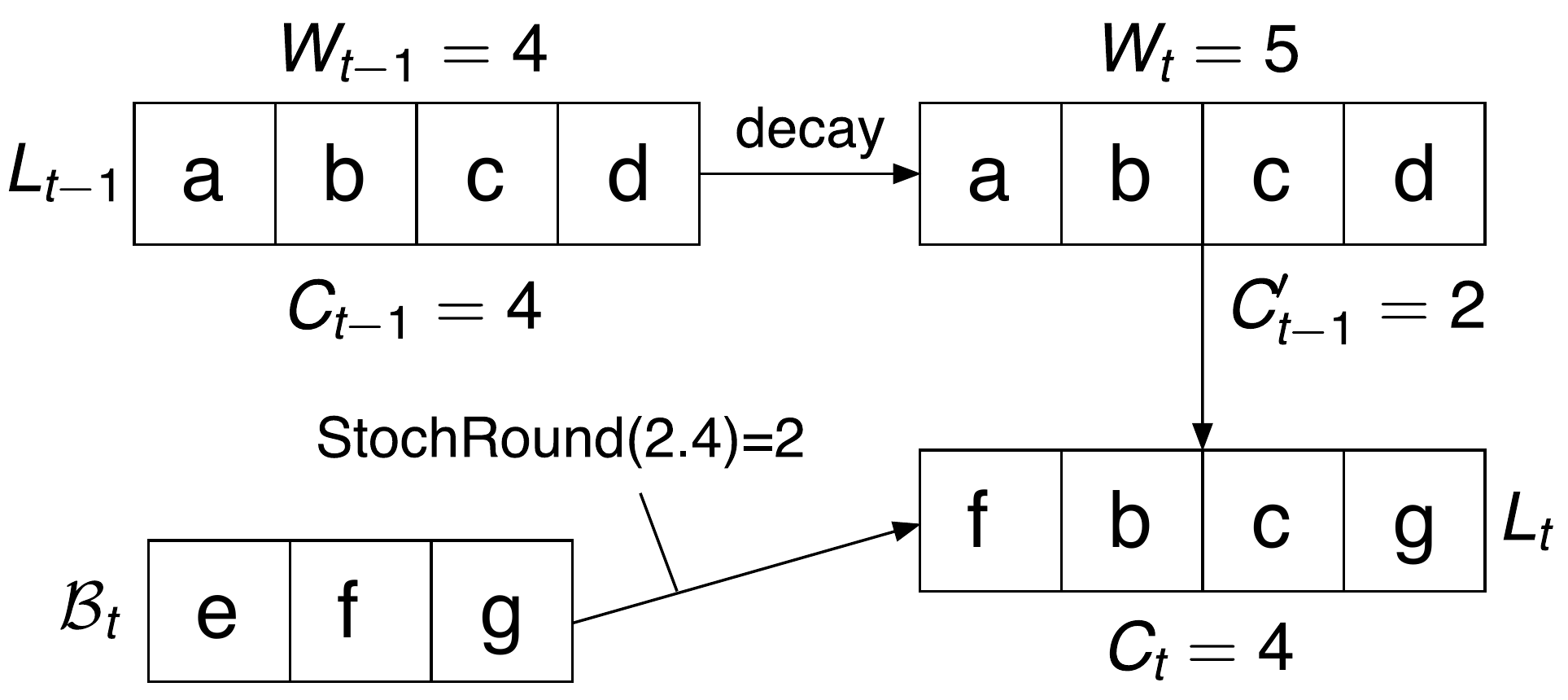}} 
\BigCrunch
\caption{\label{fig:rtbsScenarios}R-TBS scenarios for $n=4$ and $e^{-\lambda}=0.5$. For simplicity, we take $W_{t-1}=C_{t-1}$. ``DS'' denotes downsampling.}
\BigCrunch
\end{figure*}

The key goal of R-TBS is to maintain the invariant
\begin{equation}\label{eq:inclRTBS}
\prob{i\in S_t}=\bigl(C_t/W_t\bigr)w_{t}(i)
\end{equation}
for each $t\ge 0$ and each item $i\in  U_t$, where, as before, $U_t$ denotes the set of all items that arrive up through time $t$, so that the appearance probability for an item~$i$ at time~$t$ is proportional to its weight $w_{t}(i)$. This immediately implies the desired relative-inclusion property \eqref{eq:expratio}. Since $w_t(i)=1$ for an arriving item $i\in \xB_t$, the equality in \eqref{eq:inclRTBS} implies that the initial acceptance probability for this item is
\begin{equation}\label{eq:inclInitRTBS}
\prob{i\in S_t}=C_t/W_t.
\end{equation}

The pseudocode for R-TBS is given as Algorithm~\ref{alg:rtbs}. Suppose the sample is \emph{unsaturated} at time~$t-1$ in that $W_{t-1}<n$ and hence $C_{t-1}=W_{t-1}$ (line~\ref{ln:ifunsat}). The decay process first reduces the total weight (and hence the sample weight) to $W'_{t-1}=C'_{t-1}=e^{-\lambda}W_{t-1}$ (line~\ref{ln:decayW}). R-TBS then \emph{downsamples} $L_{t-1}$ (line~\ref{ln:dsample}) to reflect this decay and maintain a minimal sample footprint; the downsampling method, described in Section~\ref{sec:latent}, is designed to maintain the invariant in \eqref{eq:inclRTBS}. 
If the weight of the arriving batch does not cause the sample to overflow, i.e., $C'_{t-1}+|\xB_t|<n$, then $C_t=C'_{t-1}+|\xB_t|=W'_{t-1}+|\xB_t|=W_t$. The relation in \eqref{eq:inclInitRTBS} then implies that all newly arrived items are accepted into the sample with probability~1 (line~\ref{ln:uUpdateA}); see Figure~\ref{fig:rtbsScenario1} for an example of this scenario. The situation is more complicated if the weight of the arriving batch would cause the sample to overflow. It turns out that the simplest way to deal with this scenario is to initially accept all incoming items as in line~\ref{ln:uUpdateA}, and then run an additional round of downsampling to reduce the sample weight to $n$ (line~\ref{ln:overshoot}), so that the sample is now saturated; see Figure~\ref{fig:rtbsScenario2}. Note that these two steps can be executed without ever causing the sample footprint to exceed $n$.

Now suppose that the sample is \emph{saturated} at time~$t-1$, so that $W_{t-1}\ge n$ and hence $C_{t-1}=|S_{t-1}|=n$. The new total weight is $W_t=W'_{t-1}+|\xB_t|$ as before (line~\ref{ln:newWeight}). If $W_t\ge n$, then the weight of the arriving batch exceeds the weight loss due to decay, and the sample remains saturated. Then \eqref{eq:inclInitRTBS} implies that each item in $\xB_t$ is accepted into the sample with probability $p=n/W_t$. Letting $I_j=1$ if item $j\in\xB$ is accepted and $I_j=0$ otherwise, we see that the expected number of accepted items is
\[
m=\mean\Bigl[\sum_{j\in\xB_t}I_j\Bigr]=\sum_{j\in\xB_t}\mean[I_j]=\sum_{j\in\xB_t}\prob{I_j=1}=B_t n/W_t.
\]
There are a number of possible ways to carry out  this acceptance operation, e.g., via independent coin flips. To minimize the variability of the sample size (and hence the likelihood of severely small samples), R-TBS uses \emph{stochastic rounding} in line~\ref{ln:bround} and accepts a random number of items $M$ such that $M=\floor{m}$ with probability~$\ceil{m}-m$ and $M=\ceil{m}$ with probability~$m-\floor{m}$, so that $\mean[M]=m$ by an argument essentially the same as in \eqref{eq:meansize}. To maintain the bound on the sample size, the $M$ accepted items replace $M$ randomly selected ``victims'' in the current sample (line~\ref{ln:sUpdateA}). If $W_t<n$, then the sample weight decays to $W'_{t-1}$ and the weight of the arriving batch is not enough to fill the sample back up. Moreover, \eqref{eq:inclInitRTBS} implies that all arriving items are accepted with probability~1. Thus we downsample to the decayed weight of $W'_{t-1}=W_t-|\xB_t|$ in line~\ref{ln:undershoot} and then insert the arriving items in line~\ref{ln:fillup}.

\subsection{Downsampling}\label{sec:latent}

Before describing Algorithm~\ref{alg:downsamp}, the downsampling algorithm, we intuitively motivate a key property that any such procedure must have. For any item $i\in L$, the relation in \eqref{eq:inclRTBS} implies that we must have $\prob{i\in S} = (C/W)w_i$ and $\prob{i\in S'} = (C'/W')w'_i$, where $W$ and $w_i$ represent the total and item weight before decay and downsampling, and $W'$ and $w'_i$ represent the weights afterwards. Since decay affects all items equally, we have $w/W=w'/W'$, and it follows that 
\begin{equation}\label{eq:downsample}
\prob{i\in S'}=(C'/C)\prob{i\in S}.
\end{equation}
That is, the inclusion probabilities for all items must be scaled down by the same fraction, namely $C'/C$. Theorem~\ref{th:downsamp} (later in this section) asserts that Algorithm~\ref{alg:downsamp} satisfies this property. 

\begin{algorithm}[t]
\caption{Downsampling}\label{alg:downsamp}
{\footnotesize
$L=(A,\pi,C)$: input latent sample\;
$C'$: input target weight with $0<C'<C$\;
$L'=(A',\pi',C')$: output latent sample\;
\BlankLine
$U\gets \textsc{Uniform}()$\;
\uIf(\Comment*[f]{no full items retained}){$\floor{C'}=0$\label{ln:no_full}}{
    \If{$U>\frc(C)/C$}{$(A',\pi')\gets\textsc{Swap1}(A,\pi)$}\label{ln:swap0}
    $A'\gets\emptyset$\;\label{ln:killA}
}
\uElseIf(\Comment*[f]{no items deleted}){$0<\floor{C'}=\floor{C}$\label{ln:no_delete}}{
    \If{$U>\bigl(1-(C'/C)\frc(C)\bigr)/\bigl(1-\frc(C')\bigr)$\label{ln:no_swap}}{
        $(A',\pi') \gets\textsc{Swap1}(A,\pi)$\;\label{ln:convert}
        }
}
\Else(\Comment*[f]{items deleted: $0<\floor{C'}<\floor{C}$}){
    \eIf{$U\le (C'/C)\frc(C)$\label{ln:normal}}{
        $A' \gets \textsc{Sample}(A,\floor{C'})$\;
        $(A',\pi')\gets\textsc{Swap1}(A',\pi)$\;\label{ln:swap}
        }{
        $A' \gets \textsc{Sample}(A,\floor{C'}+1)$\;\label{ln:smpl}
        $(A',\pi')\gets\textsc{Move1}(A',\pi)$\;\label{ln:move}
        }      
}
\If(\Comment*[f]{no fractional item}){$C'=\floor{C'}$}{
    $\pi'\gets\emptyset$\;
    }
}
\end{algorithm}

In the pseudocode for Algorithm~\ref{alg:downsamp},  the function $\textsc{Uniform}()$ generates a random number uniformly distributed on $[0,1]$. The subroutine $\textsc{Swap1}(A,\pi)$ moves a randomly selected item from $A$ to $\pi$ and moves the current item in $\pi$ (if any) to $A$. Similarly, $\textsc{Move1}(A,\pi)$ moves a randomly selected item from $A$ to $\pi$, replacing the current item in $\pi$ (if any). More precisely, $\textsc{Swap1}(A,\pi)$ executes the operations $I\gets \textsc{Sample}(A,1)$, $A\gets (A\setminus I)\cup\pi$, and $\pi\gets I$,
and $\textsc{Move1}(A,\pi)$ executes the operations $I\gets \textsc{Sample}(A,1)$, $A\gets A\setminus I$, and $\pi\gets I$.

\begin{figure*}[tbh]
        \centering
	\subfigure[From $C_t=3$ to $C'_t=1.5$.]{
	   \label{fig:downsamp1}\includegraphics[width=0.25\linewidth]{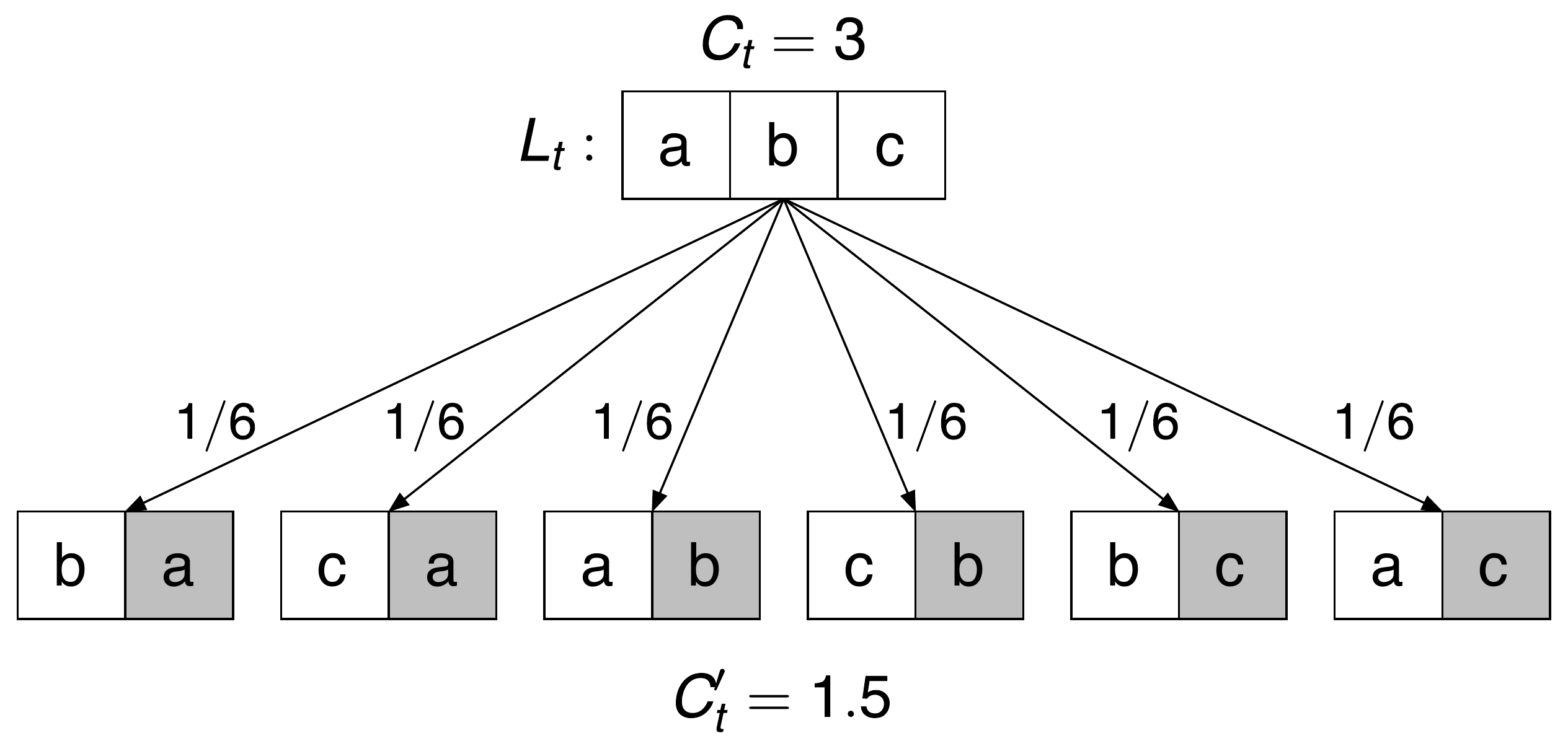}} \hspace{0.02\textwidth}
	\subfigure[From $C_t=3.2$ to $C'_t=1.6$.]{
	   \label{fig:downsamp2}\includegraphics[width=0.25\linewidth]{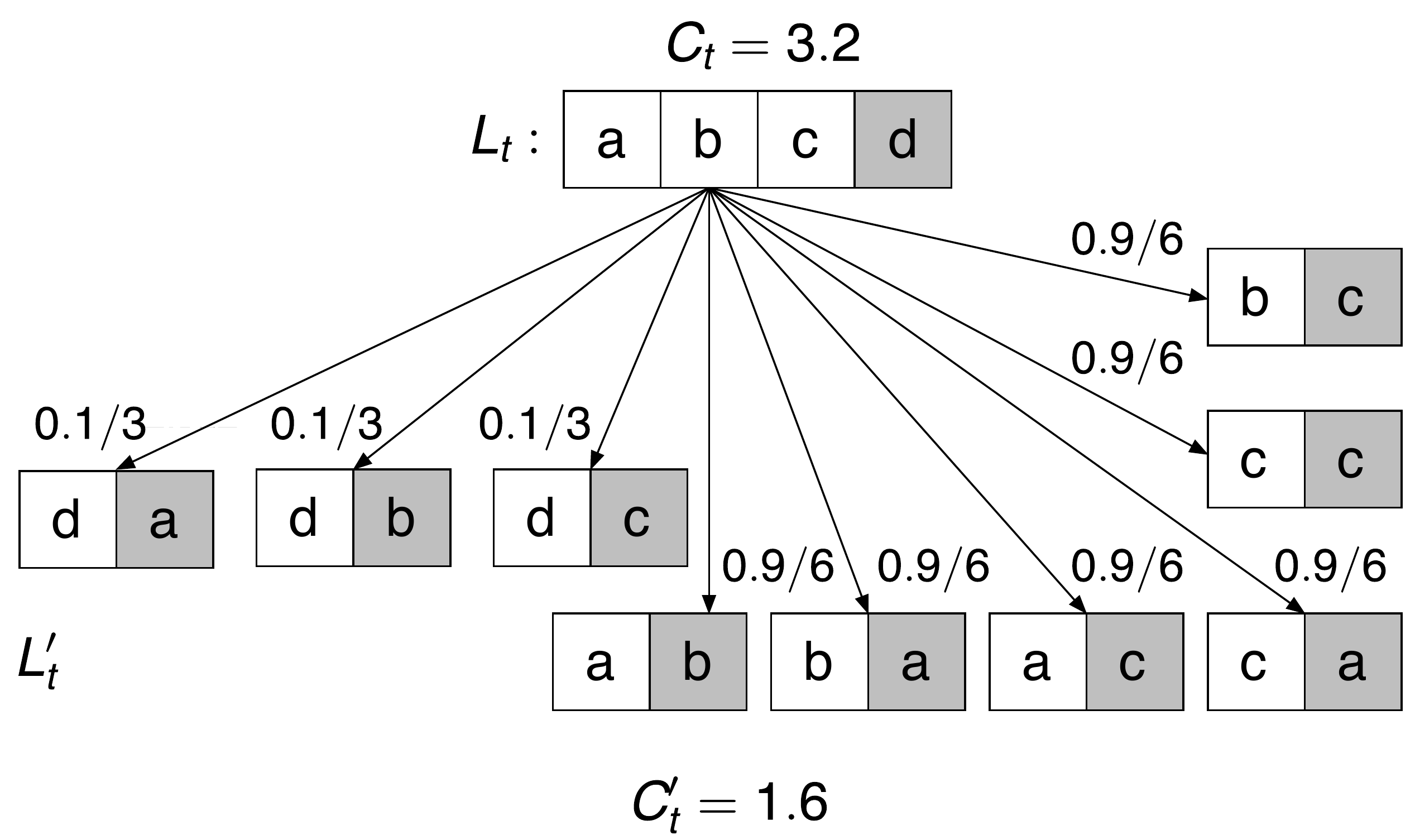}} \hspace{0.02\textwidth}
	\subfigure[From $C_t=2.4$ to $C'_t=0.4$.]{
	   \label{fig:downsamp3}\includegraphics[width=0.15\linewidth]{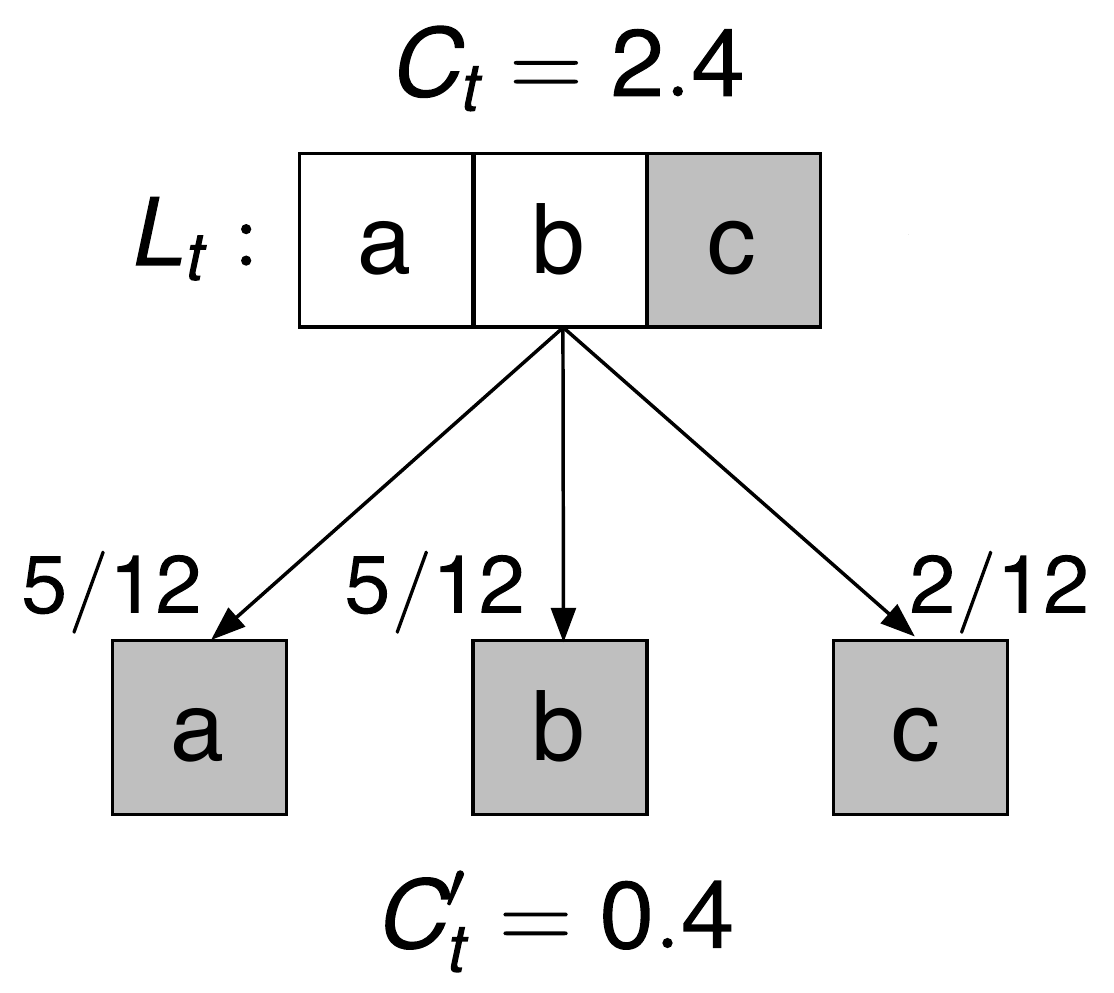}} \hspace{0.02\textwidth}
	\subfigure[From $C_t=2.4$ to $C'_t=2.1$.]{
	   \label{fig:downsamp4}\includegraphics[width=0.22\linewidth]{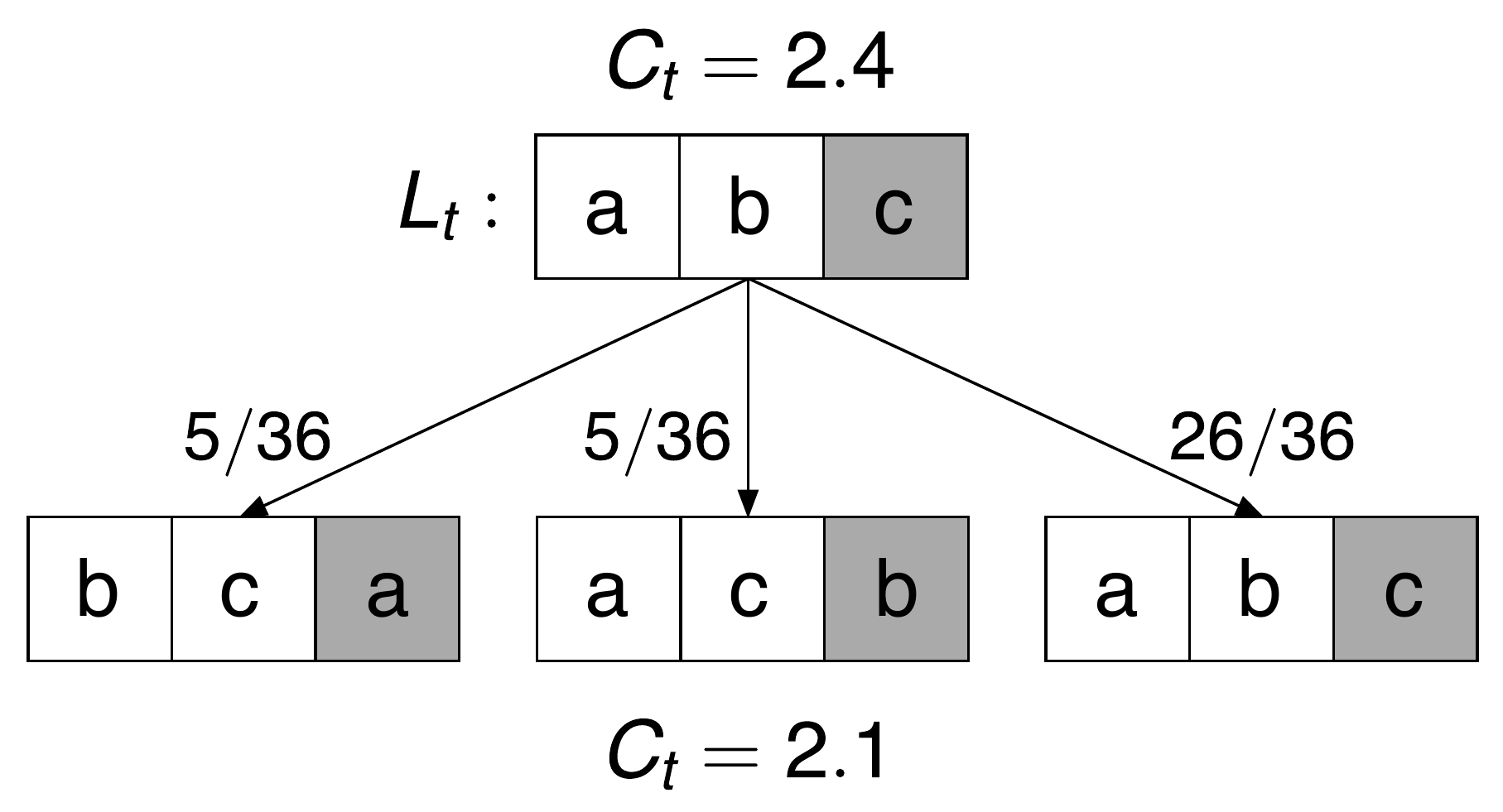}} 
\BigCrunch
\caption{\label{fig:downsampling}Downsampling examples ($t=0$).}
\BigCrunch
\end{figure*}

%\begin{figure}[tbh]
%       \centering
%       \includegraphics[width=0.7\linewidth]{figs/downSamp1}
%       \SmallCrunch
%       \caption{Downsampling from $C_t=3$ to $C'_t=1.7$.}
%       \label{fig:downsamp1}
%       \SmallCrunch
%\end{figure}
%
%\begin{figure}[tbh]
%       \centering
%       \includegraphics[width=0.6\linewidth]{figs/downSamp2}
%       \SmallCrunch
%       \caption{Downsampling from $C_t=3.2$ to $C'_t=1.6$.}
%       \label{fig:downsamp2}
%       \SmallCrunch
%\end{figure}
%
%\begin{figure}[tbh]
%       \centering
%       \includegraphics[width=0.25\linewidth]{figs/downSamp3}
%       \SmallCrunch
%       \caption{Downsampling from $C_t=2.4$ to $C'_t=0.4$.}
%       \label{fig:downsamp3}
%       \SmallCrunch
%\end{figure}
%
%\begin{figure}[tbh]
%       \centering
%       \includegraphics[width=0.5\linewidth]{figs/downSamp4}
%       \SmallCrunch
%       \caption{Downsampling from $C_t=2.4$ to $C'_t=2.1$.}
%       \label{fig:downsamp4}
%       \SmallCrunch
%\end{figure}

To gain some intuition for why the algorithm works, consider a simple special case, where the goal is to form a fractional sample $L' = (A',\pi', C')$ from a fractional sample $L=(A,\pi,C)$ of integral size $C > C'$; that is, $L$ comprises exactly $C$ full items. Assume that $C'$ is non-integral, so that $L'$ contains a partial item. In this case, we simply select an item at random (from $A$) to be the partial item in $L'$ and then select $\floor{C'}$ of the remaining $C -1$ items at random to be the full items in $L'$; see Figure~\ref{fig:downsamp1}. By symmetry, each item $i\in L$ is equally likely to be included in $S'$, so that the inclusion probabilities for the items in $L$ are all scaled down by the same fraction, as required for \eqref{eq:downsample}. For example, taking $t=0$ in Figure~\ref{fig:downsamp1}, item $a$ appears in $S_t$ with probability~1 since it is a full item. In $S'_t$, where the weights have been reduced by 50\%, item $a$ (either as a full or partial item, depending on the random outcome) appears with probability $2\cdot(1/6)+2\cdot(1/6)\cdot 0.5=0.5$, as expected.
%; as shown in the correctness proof, this fraction is precisely $C'/C$.
This scenario corresponds to lines~\ref{ln:smpl} and \ref{ln:move} in the algorithm, where we carry out the above selections by randomly sampling $\floor{C'}+1$ items from $A$ to form $A'$ and then choosing a random item in $A'$ as the partial item by moving it to $\pi$.

In the case where $L$ contains a partial item~$i^*$ that appears in $S$ with probability $\frc(C)$, it follows from \eqref{eq:downsample} that $i^*$ should appear in $S'$ with probability $p=(C'/C)P[i^*\in S]=(C'/C)\frc(C)$. Thus, with probability~$p$, lines~\ref{ln:normal}--\ref{ln:swap} retain $i^*$ and convert it to a full item so that it appears in $S'$. Otherwise, in lines~\ref{ln:smpl} and \ref{ln:move}, $i^*$ is removed from the sample when it is overwritten by a random item from $A'$; see Figure~\ref{fig:downsamp2}. Again, a new partial item is chosen from $A$ in a random manner to uniformly scale down the inclusion probabilities. For instance, in Figure~\ref{fig:downsamp2}, item $d$ appears in $S_t$ with probability 0.2 (because it is a partial item) and in $S'_t$, appears with probability $3\cdot (0.1/3)=0.1$. Similarly, item $a$ appears in $S_t$ with probability~1 and in $S'_t$ with probability $(1.8)/6+0.6\cdot (1.8/6)+0.6\cdot(0.1/3)=0.5$.

The if-statement in line~\ref{ln:no_full} corresponds to the corner case in which $L'$ does not contain a full item. The partial item $i^*\in L$ either becomes full or is swapped into $A'$ and then immediately ejected; see Figure~\ref{fig:downsamp3}.

The if-statement in line~\ref{ln:no_delete} corresponds to the case in which no items are deleted from the latent sample, e.g., when $C=4.7$ and $C'=4.2$. In this case, $i^*$ either becomes full by being swapped into $A'$ or remains as the partial item for $L'$. Denoting by $\rho$ the probability of \emph{not} swapping, we have $P[i^*\in S'] = \rho\cdot\frc(C') + (1-\rho)\cdot 1$. On the other hand, \eqref{eq:downsample} implies that $P[i^*\in S']=(C'/C)\frc(C)$. Equating these expression shows that $\rho$ must equal the expression on the right side of the inequality on line~\ref{ln:no_swap}; see Figure~\ref{fig:downsamp4}.

Formally, we have the following result.
\begin{theorem}\label{th:downsamp}
For $0<C'<C$, let $L'=(A',\pi',C')$ be the latent sample produced from a latent sample~$L=(A,\pi,C)$ via Algorithm~\ref{alg:downsamp}, and let $S'$ and $S$ be samples produced from $L'$ and $L$ via \eqref{eq:getsample}. Then $\prob{i\in S'}=(C'/C)\prob{i\in S}$ for all $i\in L$.
\end{theorem}

\subsection{Properties of R-TBS}\label{sec:RTBSalg}

Theorem~\ref{th:rtbsIncl} below asserts that R-TBS satisfies \eqref{eq:inclRTBS} and hence \eqref{eq:expratio}, thereby maintaining the correct inclusion probabilities; see Appendix~\ref{sec:proofs}\longPaper for the proof. Theorems~\ref{th:maxMean} and \ref{th:minVar} assert that, among all sampling algorithms with exponential time biasing, R-TBS both maximizes the expected sample size in unsaturated scenarios and minimizes sample-size variability. Thus R-TBS tends to yield more accurate results (from more training data) and greater stability in both result quality and retraining costs. 

\begin{theorem}\label{th:rtbsIncl}
The relation $\prob{i\in S_t}=(C_t/W_t)w_t(i)$ holds for all $t\ge 1$ and $i\in U_t$.
\end{theorem}

\begin{theorem}\label{th:maxMean}
Let $H$ be any sampling algorithm that satisfies~\eqref{eq:expratio} and denote by $S_t$ and $S^{H}_t$ the samples produced at time~$t$ by R-TBS and H. If the total weight at some time $t\ge 1$ satisfies $W_t<n$, then $\mean[|S^H_t|]\le\mean[|S_t|]$.
\end{theorem}

\begin{proof}
Since $H$ satisfies \eqref{eq:expratio}, it follows that, for each time $j\le t$ and $i\in\xB_j$, the inclusion probability $\prob{i\in S^H_t}$ must be of the form $r_te^{-\lambda(t-j)}$ for some function $r_t$ independent of $j$. Taking $j=t$, we see that $r_t\le 1$. For R-TBS in an unsaturated state, \eqref{eq:inclRTBS} implies that $r_t=C_t/W_t=1$, so that $\prob{i\in S^H_t}\le\prob{i\in S_t}$ , and the desired result follows directly.
\end{proof}

\begin{theorem}\label{th:minVar}
Let $H$ be any sampling algorithm that satisfies~\eqref{eq:expratio} and has maximal expected sample size $C_t$ and denote by $S_t$ and $S^{H}_t$ the samples produced at time~$t$ by R-TBS and H. Then $\var[|S^H_t|]\ge\var[|S_t|]$ for any time $t\ge 1$.
\end{theorem}

\begin{proof}
Considering all possible distributions over the sample size having a mean value equal to $C_t$, it is straightforward to show that variance is minimized by concentrating all of the probability mass onto $\floor{C_t}$ and $\ceil{C_t}$. There is precisely one such distribution, namely the stochastic-rounding distribution, and this is precisely the sample-size distribution attained by R-TBS.
\end{proof}

%R-TBS can easily handle out-of-order arrivals. Specifically, if an item with timestamp $t$  arrives at time $t'>t$, then in line~\ref{ln:uUpdateW} or \ref{ln:newWeight} of Algorithm~\ref{alg:rtbs} we simply increment the running total weight $W$ by $e^{-\lambda(t'-t)}$ rather than by 1. Indeed, the algorithm can be modified in this manner to handle non-unitary weights in general, allowing additional factors besides arrival time to determine the probability that an item appears in the current sample.

%% file: implmt_ver3.tex
% !TEX root = tbs-mm.tex
% above command is for TeXShop

\section{Distributed TBS Algorithms}\label{sec:imp}

In this section, we describe how to implement distributed versions of T-TBS and R-TBS to handle large volumes of data. 

%Spark is a general-purpose distributed processing framework based on a functional programming paradigm. It has built-in support for batch processing, streaming, machine learning, SQL queries, and graph processing.

%A Resilient Distributed Dataset (RDD) is the core distributed memory abstraction in Spark to support fault-tolerant computation across a cluster of machines. An RDD is divided into partitions that are then distributed across the cluster for parallel processing. RDDs can either reside in the aggregate main memory of the cluster, or in efficiently serialized disk blocks. An RDD is immutable and cannot be modified, but a new RDD can be constructed by transforming an existing RDD. Spark utilizes both lineage tracking and checkpointing of RDDs for fault tolerance.

%A Spark program consists of a single driver and many executors. The driver of a Spark program orchestrates the control flow of an application, tracks the application state, and can request distributed computation over RDDs from the Spark framework. When the driver requests a distributed computation, executors across the cluster perform operations on the RDDs, creating new RDDs. 
%Because distributed computation is limited to transformations over RDDs, the Spark Driver is the only part of a Spark application that can contain mutable variables. %For fault tolerance, Spark uses a combination of lineage and checkpointing. 

\subsection{Overview of Distributed Algorithms}

The distributed T-TBS and R-TBS algorithms, denoted as D-T-TBS and D-R-TBS respectively, 
%follow similar steps as in their corresponding centralized versions shown in Algorithm~\ref{alg:targsamp} and Algorithm~\ref{alg:rtbs} but 
need to distribute large data sets across the cluster and parallelize the computation on them. 
%Both distributed algorithms behave very similarly: at each time~$t$, they incrementally maintain the sample by deleting items in the existing sample and inserting items from incoming batch $\xB_t$ into the sample. Both algorithms also periodically checkpoint the sample as well as other system state variables to ensure the fault tolerance. The major difference between the two algorithms lies in how they choose items to delete and insert. D-R-TBS has to synchronize the workers to ensure that the overall sample remains bounded in the face of distributed inserts and deletes.  

\textbf{Overview of D-T-TBS:} The implementation of the D-T-TBS algorithm is very similar to the simple distributed Bernoulli time-biased sampling algorithm in~\cite{XieTSBH15}. It is embarrassingly parallel, requiring no coordination. At each time point~$t$, each worker in the cluster subsamples its partition of the sample with probability $p$, subsamples its partition of $\xB_t$ with probability $q$, and then takes a union of the resulting data sets. 

\textbf{Overview of D-R-TBS:} This algorithm, unlike D-T-TBS, maintains a bounded sample, and hence cannot be embarrassingly parallel. D-R-TBS first needs to aggregate local batch sizes to compute the incoming batch size $|\xB_t|$ to maintain the total weight $W$. Then, based on $|\xB_t|$ and the previous total weight $W$, D-R-TBS determines whether the reservoir was previously saturated and whether it will be saturated after processing $\xB_t$. For each possible situation, D-R-TBS chooses the items in the reservoir to delete through downsampling and the items in $\xB_t$ to insert into the reservoir. This process requires the master to coordinate among the workers. In Section~\ref{sec:updates}, we introduce two alternative approaches to determine the deleted and inserted items. Finally, the algorithm applies the deletes and inserts to form the new reservoir, and computes the new total weight $W$. 

Both D-T-TBS and D-R-TBS periodically checkpoint the sample as well as other system state variables to ensure fault tolerance. The implementation details for D-T-TBS are mostly subsumed by those for D-R-TBS, so we focus on the latter.
%Since the D-T-TBS algorithm is much simpler, and its implementation details are mostly subsumed by that of the D-R-TBS algorithm, we focus henceforth on providing the implementation details of D-R-TBS.

\subsection{Distributed Data Structures}\label{sec:reservoir}

There are two important data structures in the D-R-TBS algorithm: the incoming batch and the reservoir. Conceptually, we view an incoming batch $\xB_t$ as an array of slots numbered from 1 through $|\xB_t|$, and the reservoir as an array of slots numbered from 1 through $\floor{C}$ containing full items plus a special slot for the partial item. For both data structures, data items need to be distributed into partitions due to the large data volumes. Therefore, the slot number of an item maps to a specific partition ID and a position inside the partition.

The incoming batch usually comes from a distributed streaming system, such as Spark Streaming; the actual data structure is specific to the streaming system (e.g. an incoming batch is stored as an RDD in Spark Streaming). As a result, the partitioning strategy of the incoming batch is opaque to the D-R-TBS algorithm. Unlike the incoming batch, which is read-only and discarded at the end of each time period, the reservoir data structure must be continually updated. An effective strategy for storing and operating on the reservoir is thus crucial for good performance. We now explore alternative approaches to implementing the reservoir.

%A naive approach that stores the reservoir at master is not viable because (i) it imposes an unacceptable additional burden on the master, (ii) it prevents parallelization of the operations on the reservoir since inserts and deletes must be synchronized at the master, and (iii) it can cause violations of the master memory limit  when the reservoir size is large. We now explore alternative approaches to implementing the reservoir in distributed ways.

\begin{figure}[t]
	\centering
	\includegraphics[width=2.6 in]{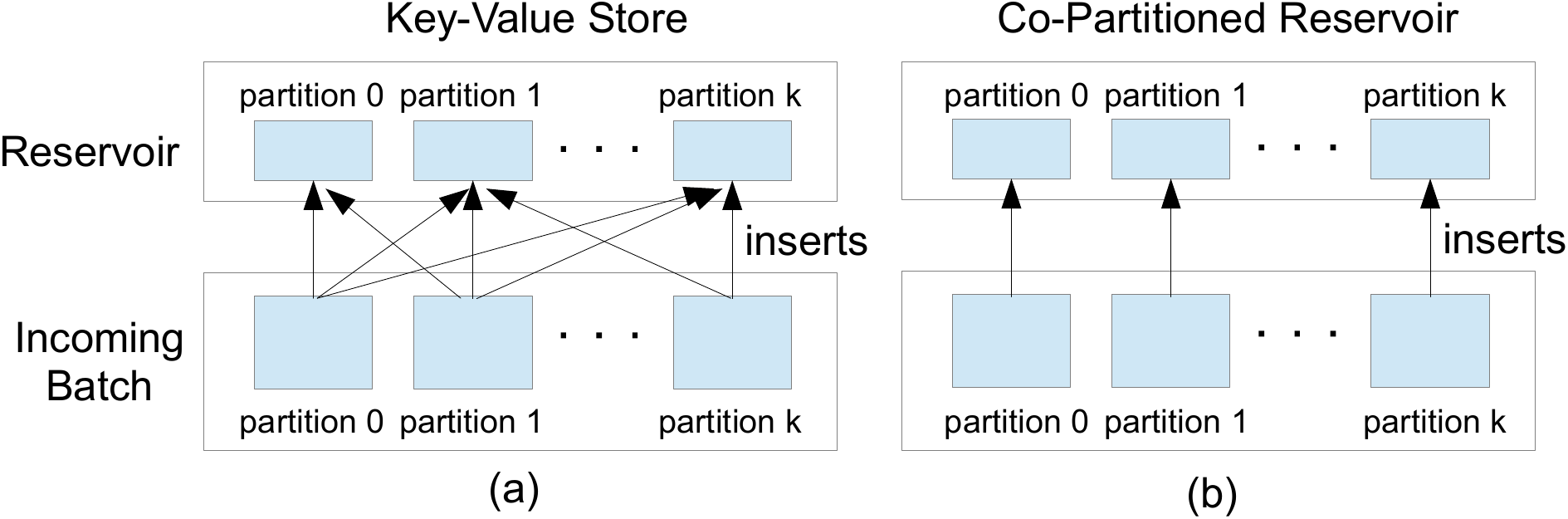}
	\SmallCrunch
	\caption{Design choices for implementing the reservoir}
	\label{fig:reservoir}
	%\BigCrunch
	\SmallCrunch
\end{figure}
\begin{figure}[t]
	\centering
	\includegraphics[width=2.6 in]{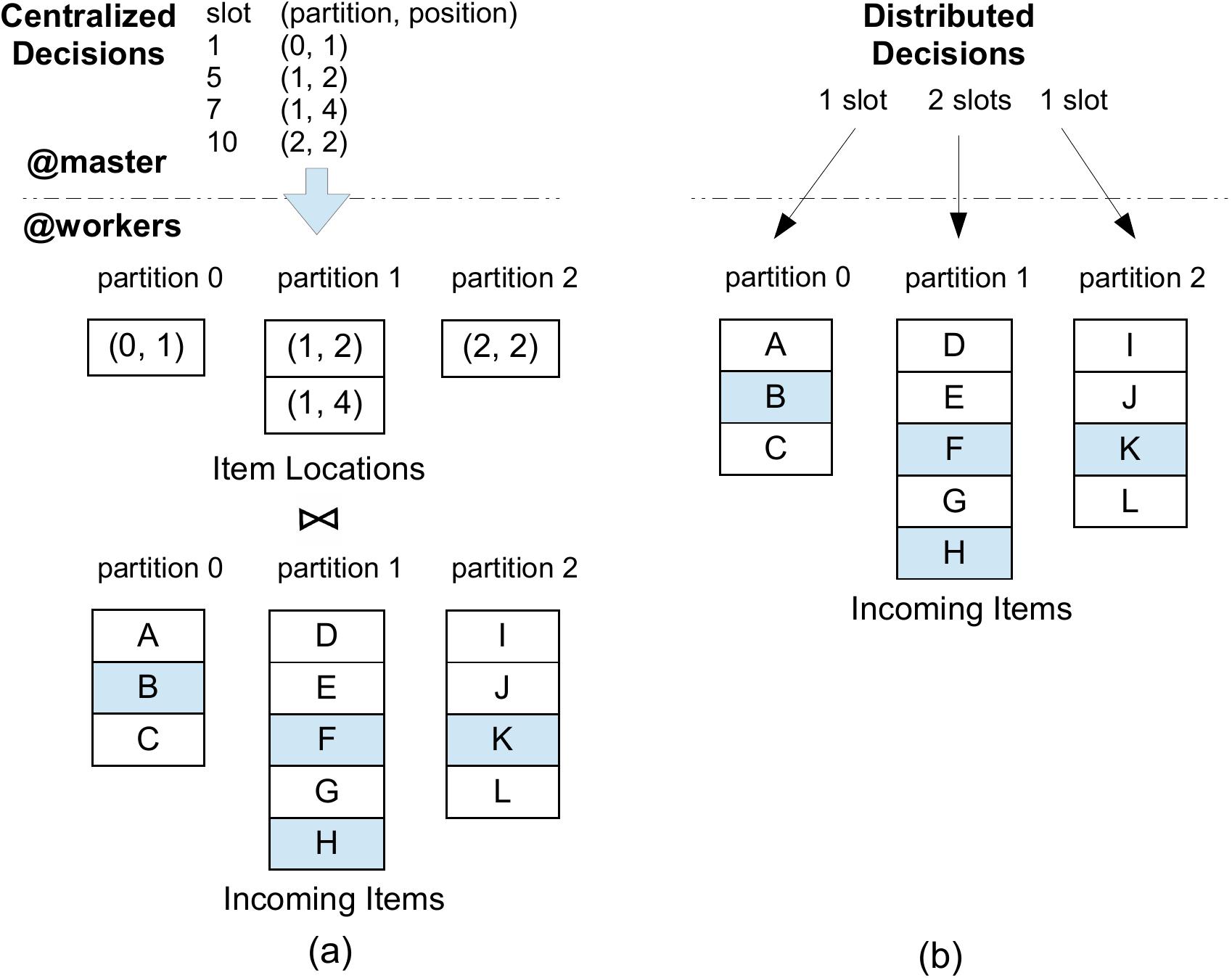}
	\SmallCrunch
	\caption{Retrieving insert items}
	\label{fig:insert}
	\BigCrunch
	\SmallCrunch
\end{figure}

\textbf{Distributed in-memory key-value store:}
One quite natural approach implements the reservoir using an off-the-shelf distributed in-memory key-value store, such as Redis~\cite{redis} or Memcached~\cite{memcached}. In this scheme, each item in the reservoir is stored as a key-value pair, with the slot number as the key and the item as the value. Inserts and deletes to the reservoir naturally translate into put and delete operations to the key-value store.

There are two major limitations to this approach. Firstly, the hash-based or range-based data-partitioning scheme used by a distributed key-value store yields reservoir partitions that do not correlate with the partitions of incoming batch. As illustrated in Figure~\ref{fig:reservoir}(a), when items from a given partition of an incoming batch are inserted into the reservoir, the inserts touch many (if not all) partitions of the reservoir, incurring heavy network I/O. Secondly, key-value stores incur needless concurrency-control overhead. For each batch, D-R-TBS already carefully coordinates the deletes and inserts so that no two delete or insert operations access the same slots in the reservoir and there is no danger of write-write or read-write conflicts. 

%its update performance is inferior to real key-value stores like Redis or MemcacheD and hence is dominated by them.

\textbf{Co-partitioned reservoir:}
In the alternative approach, we implement a distributed in-memory data structure for the reservoir so as to ensure that the reservoir partitions coincide with the partitions from incoming batches, as shown in Figure~\ref{fig:reservoir}(b). This can be achieved in spite of the unknown partitioning scheme of the streaming system. Specifically, the reservoir is initially empty, and all items in the reservoir are from the incoming batches. Therefore, if an item from a given partition of an incoming batch is always inserted into the corresponding ``local'' reservoir partition and deletes are also handled locally, then the co-partitioning and co-location of the reservoir and incoming batch partitions is automatic. For our experiments, we implemented the co-partitioned reservoir in Spark using the in-place updating technique for RDDs in~\cite{XieTSBH15}; see Appendix~\ref{sec:spark-impl}\longPaper.

Note that, at any point in time, a given slot number in the reservoir maps to a specific partition ID and a position inside the partition. Thus the slot number for a given full item may change over time due to reservoir insertions and deletions. This does not cause any statistical issues, because the functioning of the set-based R-TBS algorithm is oblivious to specific slot numbers.  

%\textbf{IndexedRDD:}
%Another possibility is to represent the reservoir as an IndexedRDD~\cite{indexedrdd}, which was designed to allow updates to an RDD in a manner similar to those in a key-value store. We ruled out IndexedRDD for the several reasons. First, like Redis or Memcached, the partitions for the incoming batch and reservoir do not correlate. In addition, it is a third-party package for Spark that hasn't been actively developed on since Sep 2015, and no performance study has shown that it is better than highly optimized in-memory key-value stores like Redis or MemcacheD.

%\SmallCrunch
\subsection{Choosing Items to Delete and Insert}\label{sec:updates}

In order to bound the reservoir size, D-R-TBS requires careful coordination when choosing the set of items to delete from, and insert into, the reservoir. At the same time, D-R-TBS must ensure the statistical correctness of random number generation and random permutation operations in the distributed environment. We consider two possible approaches.

\textbf{Centralized decisions:}
In the most straightforward approach, the master makes centralized decisions about which items to delete and insert. For deletes, the driver generates slot numbers of the items in the reservoir to be deleted, which are then mapped to the actual data locations in a manner that depends on the representation of the reservoir (key-value store or co-partitioned reservoir). For inserts, the driver generates the slot numbers of the incoming items $\xB_t$ at time~$t$ that need to be inserted into the reservoir. Suppose that $\xB_t$ comprises $k\ge 1$ partitions. Each generated slot number $i\in\{1,2,\ldots,|\xB_t|\}$ is mapped to a partition $p_i$ of the $\xB_t$ (where $0\le p_i\le k-1$) and a position $r_i$ inside partition~$p_i$. Denote by $\mathcal{Q}$ the set of ``item locations'', i.e., the set of $(p_i,r_i)$ pairs. In order to perform the inserts, we need to first retrieve the actual items based on the item locations. This can be achieved with a join-like operation between $\mathcal{Q}$ and $\xB_t$, with the $(p_i,r_i)$ pair matching the actual location of an item inside $\xB_t$. To optimize this operation, we make $\mathcal{Q}$ a distributed data structure and use a customized partitioner to ensure that all pairs $(p_i,r_i)$ with $p_i=j$ are co-located with partition~$j$ of $\xB_t$ for $j=0,1,\ldots,k-1$. Then a co-partitioned and co-located join can be carried out between $\mathcal{Q}$ and $\xB_t$, as illustrated in Figure~\ref{fig:insert}(a) for $k=3$. The resulting set of retrieved insert items, denoted as $\mathcal{S}$, is also co-partitioned with $\xB_t$ as a by-product. After that, the actual deletes and inserts are then carried out depending on how reservoir is stored, as discussed below.

When the reservoir is implemented as a key-value store, the deletes can be directly applied based on the slot numbers. For inserts, the master takes each generated slot number of an item in $\xB_t$ and chooses a companion destination slot number in the reservoir into which the $\xB_t$ item will be inserted.  This destination reservoir slot might currently be empty due to an earlier deletion, or might contain an item that will now be replaced by the newly inserted batch item. After the actual items to insert are retrieved as described previously, the destination slot numbers are used to put the items into the right locations in the key-value store. 

When the co-partitioned reservoir is used, the delete slot numbers in the reservoir are mapped to $(p_i,r_i)$ pairs of partitions of the reservoir and positions inside the partitions. As with inserts, we again use a customized partitioner for the set of pairs $\mathcal{R}$ such that deletes are co-located with the corresponding reservoir partitions. Then a join-like operation on $\mathcal{R}$ and the reservoir performs the actual delete operations on the reservoir. For inserts, we simply use another join-like operation on the set of retrieved insert items $\mathcal{S}$ and the reservoir to add the corresponding insert items to the co-located partition of the reservoir. In this approach, we don't need the master to generate destination reservoir slot numbers for these insert items, because we view the reservoir as a set when using co-partitioned reservoir data structure.%; the slot number of an item can change from time to time, but the statistical correctness of the sampling process is not affected.

\textbf{Distributed decisions:}
The above approach requires generating a large number of slot numbers inside the master, so we now explore an alternative approach that offloads the slot number generation to the workers while still ensuring the statistical correctness of the computation. This approach has the master choose only the number of deletes and inserts per worker according to appropriate multivariate hypergeometric distributions. For deletes, each worker chooses random victims from its local partition of the reservoir based on the number of deletes given by the master. For inserts, the worker randomly and uniformly selects items from its local partition of the incoming batch $\xB_t$ given the number of inserts. Figure~\ref{fig:insert}(b) depicts how the insert items are retrieved under this decentralized approach. We use the technique in~\cite{Haramoto} for parallel pseudo-random number generation.

Note that this distributed decision making approach works only when the co-partitioned reservoir data structure is used. This is because the key-value store representation of the reservoir requires a target reservoir slot number for each insert item from the incoming batch, and the target slot numbers have to be generated in such a way as to ensure that, after the deletes and inserts, all of the slot numbers are still unique and contiguous in the new reservoir. This requires a lot of coordination among the workers, which inhibits truly distributed decision making.

%% file: expmt_copy.tex
% !TEX root = tbs-mm.tex
% above command is for TeXShop

%\begin{figure*}[bth]
%	\centerline{
%		\epsfxsize=2.2in \epsffile{figs/runtime.eps}
%		\hfill
%		\epsfxsize=1.8in \epsffile{figs/scaleout.eps}
%		\hfill
%		\epsfxsize=1.8in \epsffile{figs/scaleup.eps}
%	}
%	\SmallCrunch
%	%    \SingleSpace
%	\centerline{
%		\parbox{2.4in}{\caption{Distributed Implementations}
%			\label{fig:avg_times}}
%		\hfill
%		\parbox{1.8in}{\caption{Scale Out of D-R-TBS}
%			\label{fig:scale_out}}
%		\hfill
%		\parbox{1.8in}{\caption{Scale Up of D-R-TBS}
%			\label{fig:scale_up}}
%	}
%	\BigCrunch
%	%    \DoubleSpace
%	%\Skip
%\end{figure*}

%\begin{figure}[tbh]
%	\centering
%	\subfigure[Scale out]{
%		\label{fig:scale_out}\includegraphics[width=0.485\linewidth]{figs/scaleout.eps}} 
%	\subfigure[Scale up]{
%		\label{fig:scale_up}\includegraphics[width=0.485\linewidth]{figs/scaleup.eps}} 
%	\BigCrunch
%	\caption{Scalability of D-R-TBS}
%	\BigCrunch
%\end{figure}

\section{Experiments}\label{sec:exp}

In this section, we study the empirical performance of D-R-TBS and D-T-TBS, and demonstrate the potential benefit of using them for model retraining in online model management. We implemented D-R-TBS and D-T-TBS on Spark (refer to Appendix~\ref{sec:spark-impl}\longPaper for implementation details).

\textbf{Experimental Setup:}
All performance experiments were conducted on a cluster of 13 IBM System x iDataPlex dx340 servers. Each has two quad-core Intel Xeon E5540 2.8GHz processors and 32GB of RAM. Servers are interconnected using a 1Gbit Ethernet and each server runs Ubuntu Linux, Java 1.7 and Spark 1.6. One server is dedicated to run the Spark coordinator and, each of the remaining 12 servers runs Spark workers. There is one worker per processor on each machine, and each worker is given all 4 cores to use, along with 8~GB of dedicated memory. All other Spark parameters are set to their default values. We used Memcached 1.4.33 as the key-value store in our experiments.

For all experiments, data was streamed in from HDFS using Spark Streaming's microbatches. We report run time per round as the average over 100 rounds, discarding the first round from this average because of Spark startup costs. Unless otherwise stated, each batch contains 10 million items, the target reservoir size is 20 million elements, and the decay parameter is $\lambda = 0.07$. %Batches are run consecutively and the next batch starts as soon as the current batch completes.

\begin{figure}[tbh]
	\centering
	\includegraphics[width=0.7\linewidth]{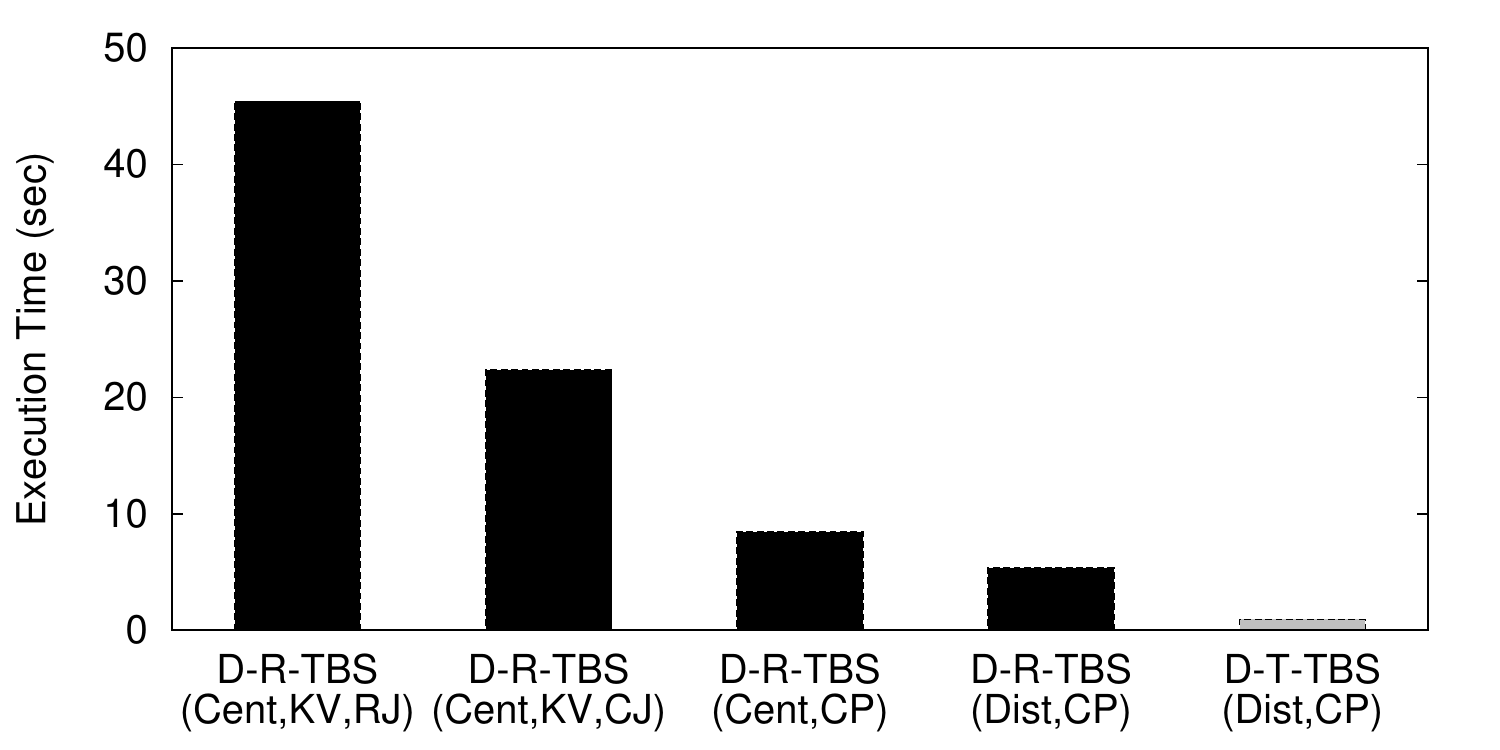}
	\BigCrunch
	\caption{Per-batch distributed runtime comparison}\label{fig:avg_times}
	\BigCrunch
	%\SmallCrunch
\end{figure}
\begin{figure}[bth]
	\centerline{
		\epsfxsize=1.8in \epsffile{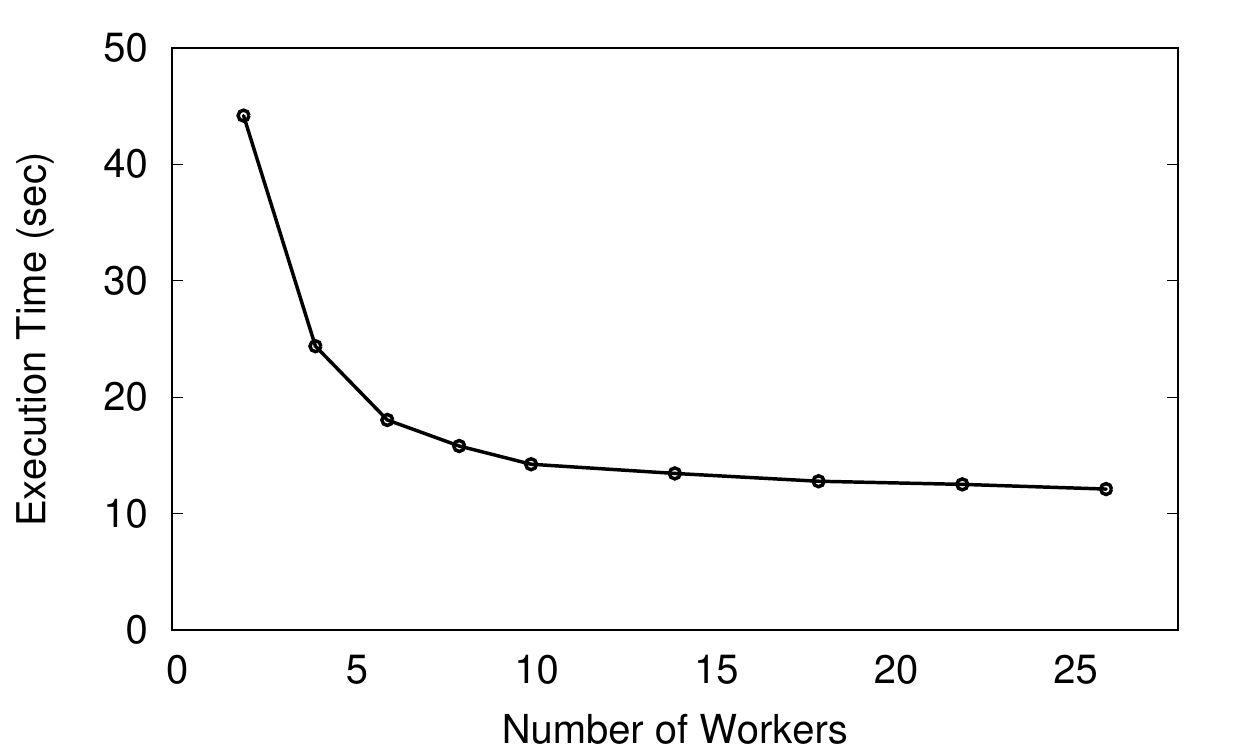}
		\hfill
		\epsfxsize=1.8in \epsffile{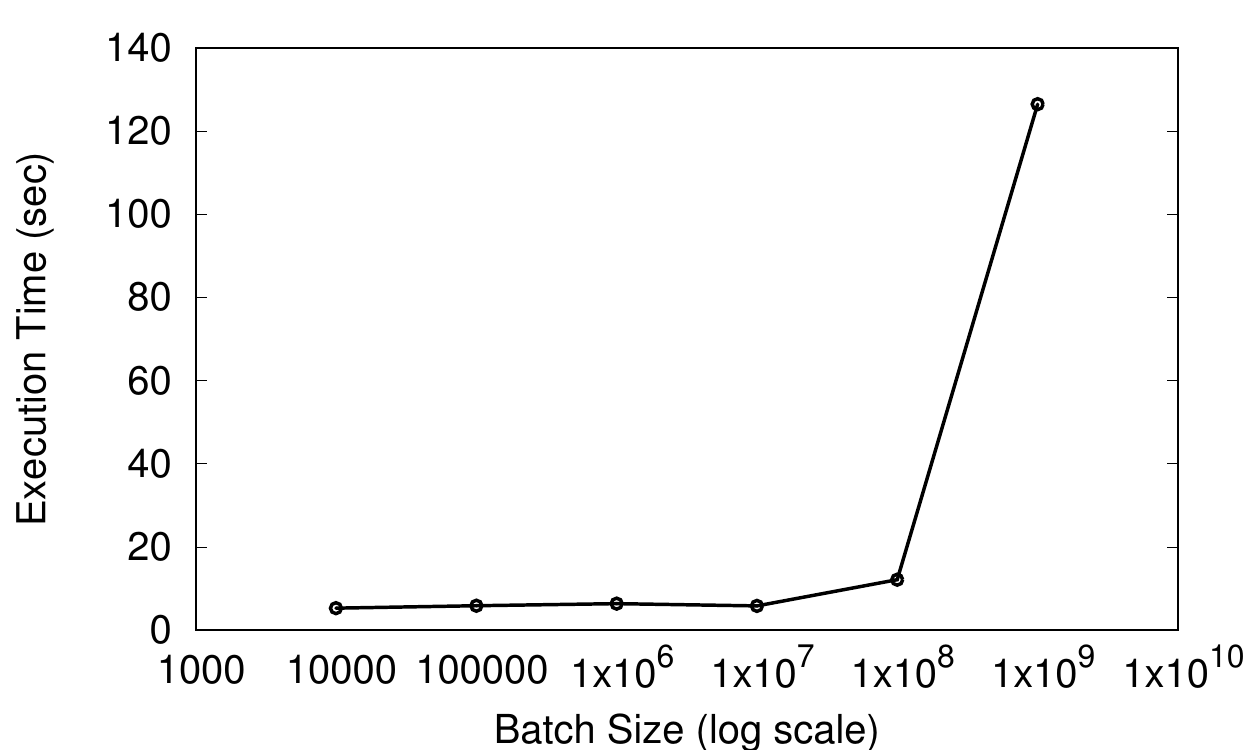}
	}
	\BigCrunch
	%    \SingleSpace
	\centerline{
		\parbox{1.8in}{\caption{Scale out of D-R-TBS}
			\label{fig:scale_out}}
		\hfill
		\parbox{1.8in}{\caption{Scale up of D-R-TBS}
			\label{fig:scale_up}}
	}
	\BigCrunch
	%\SmallCrunch
	%    \DoubleSpace
	%\Skip
\end{figure}

\subsection{Runtime Performance}

%\begin{figure*}[t]
%\centering
%\begin{minipage}{.23\linewidth}
%  \centering
%  \includegraphics[width=\linewidth]{figs/avg_times}
%  \caption{\small{Avg. Runtime per Round}}
%  \label{fig:test1}
%\end{minipage}%
%\begin{minipage}{.23\linewidth}
%  \centering
%  \includegraphics[width=\linewidth]{figs/dist_vs_cent_rg}
%  \caption{\small{Centralized vs. Distributed Reservoir Slot Generation}}
%  \label{fig:test2}
%\end{minipage}
%\begin{minipage}{.23\linewidth}
%  \centering
%  \includegraphics[width=\linewidth]{figs/scale_up}
%  \caption{\small{Throughput as a function of Workers}}
%  \label{fig:test1}
%\end{minipage}%
%\begin{minipage}{.23\linewidth}
%  \centering
%  \includegraphics[width=\linewidth]{figs/wip}
%  \caption{Another figure}
%  \label{fig:test2}
%\end{minipage}
%\end{figure*}

%\begin{figure}[thb]
%	\centering
%	\includegraphics[width=2.2 in]{figs/runtime.eps}
%	\SmallCrunch
%	\caption{Per-round Time of Distributed Implementations}
%	\label{fig:avg_times}
%	\BigCrunch
%\end{figure}

\textbf{Comparison of TBS Implementations:}
Figure~\ref{fig:avg_times} shows the average runtime per batch for five different implementations of distributed TBS algorithms. The first four (colored black) are D-R-TBS implementations with different design choices: whether to use centralized or distributed decisions (abbreviated as "Cent" and "Dist", respectively) for choosing items to delete and insert, and whether to use key-value store for storing reservoir or co-partitioned reservoir (abbreviated as "KV" and "CP", respectively). The first two implementations both use the key-value store representation for reservoir together with the centralized decision strategy for determining inserts and deletes. They only differ in how the insert items are actually retrieved when subsampling the incoming batch. The first uses the standard repartition join (abbreviated as "RJ"), whereas the second uses the customized partitioner and co-located join (abbreviated as "CJ") as described in Section~\ref{sec:updates} and depicted in Figure~\ref{fig:insert}(a). This optimization effectively cuts the network cost in half, but the KV representation of reservoir still requires the insert items to be written across the network to their corresponding reservoir location. The third implementation employs the co-partitioned reservoir instead, resulting in an significant speedup of over 2.6x. The fourth implementation further employs the distributed decision for choosing items to delete and insert. This yields a further 1.6x speedup. We use this D-R-TBS implementation in the remaining experiments.

The fifth implementation (colored grey) in Figure~\ref{fig:avg_times} is D-T-TBS using co-partitioned reservoir and the distributed strategy for choosing delete and insert items. Since, D-T-TBS is embarrassingly parallelizable, it's much faster than the best D-R-TBS implementation. But, as we discussed in Section~\ref{sec:ttbs}, T-TBS only works under a very strong restriction on the data arrival rate, and can suffer from occasional memory overflows; see Figure~\ref{fig:breakdown}. % for experiments on sample size). 
%In a distributed setting, the memory overflow problem of D-T-TBS is more severe, since a memory overflow on any single worker
%(due to imbalanced allocation of ingested data in the streaming system)
%can result in the failure of the entire job, even though the aggregate memory size is more than the sample size. 
In contrast, D-R-TBS is much more robust and works in realistic scenarios where it is hard to predict the data arrival rate.
%\begin{figure}[tbh]
%	\centering
%	%	\subfigure[\small{Avg. Runtime per Round}]{
%	%	   \label{fig:avg_times}\includegraphics[width=0.23\linewidth]{figs/avg_times}} 
%	%	\subfigure[\small{Cent. vs. Dist. Decision}]{
%	%	   \label{fig:rg}\includegraphics[width=0.23\linewidth]{figs/dist_vs_cent_rg}} 
%	\subfigure[\small{Scale-out}]{
%		\label{fig:scale_out}\includegraphics[width=0.48\linewidth]{figs/scaleout.eps}}
%	\subfigure[\small{Scale-up}]{
%		\label{fig:scale_up}\includegraphics[width=0.48\linewidth]{figs/scaleup.eps}}
%	\BigCrunch
%	\caption{\label{fig:accuracy} Scalability of D-R-TBS}
%	\BigCrunch
%\end{figure}

\textbf{Scalability of D-R-TBS:}
Figure \ref{fig:scale_out} shows how D-R-TBS scales with the number of workers. We increased the batch size to 100 million items for this experiment. Initially, D-R-TBS scales out very nicely with the increasing number of workers. However, beyond 10 workers, the marginal benefit from additional workers is small, because the coordination and communication overheads, as well as the inherent Spark overhead, become prominent. For the same reasons, in the scale-up experiment in Figure~\ref{fig:scale_up}, the runtime stays roughly constant until the batch size reaches 10 million items and increases sharply at 100 million items. This is because processing the streaming input and maintaining the sample start to dominate the coordination and communication overhead. With 10 workers, R-TBS can handle a data flow comprising 100 million items arriving approximately every 14 seconds.

\subsection{Application: Classification using kNN}

%\begin{figure*}[tbh]
%	\centering
%	\subfigure[Single Event]{
%		\label{fig:single}\includegraphics[width=0.23\linewidth]{figs/CAexp1}} 
%	\subfigure[Periodic (10, 10)]{
%		\label{fig:per1010}\includegraphics[width=0.23\linewidth]{figs/CAexp2}} 
%	\subfigure[Periodic (20, 10)]{
%		\label{fig:per2010}\includegraphics[width=0.23\linewidth]{figs/CAexp3}}
%	\subfigure[Periodic (30, 10)]{
%		\label{fig:per3010}\includegraphics[width=0.23\linewidth]{figs/CAexp4}}
%	\BigCrunch
%	\caption{\label{fig:accuracy}Accuracy of class prediction for R-TBS, SW, and Unif schemes.}
%	\BigCrunch
%\end{figure*}

We now demonstrate the potential benefits of the R-TBS sampling scheme for periodically retraining  representative ML models in the presence of evolving data. For each model and data set, we compare the quality of models retrained on the samples generated by R-TBS, a simple sliding window (SW), and uniform reservoir sampling (Unif). Due to limited space, we do not give quality results for T-TBS; we found that whenever it applies---i.e. when the mean batch size is known and constant---the quality is very similar to R-TBS, since they both use time-biased sampling. 
%For example, the prediction errors of the two for the kNN classifier experiment in Figure~\ref{fig:per1010} below are within only 0.6\% of each other, with R-TBS having the slightest edge.

Our first model is a kNN classifier, where a class is predicted for each item in an incoming batch by taking a majority vote of the classes of the $k$ nearest neighbors in the current sample, based on Euclidean distance; the sample is then updated using the batch. To generate training data, we first generate 100 class centroids uniformly in a $[0,80]\times [0,80]$ rectangle. Each data item is then generated from a Gaussian mixture model and falls into one of the 100 classes. Over time, the data generation process operates in one of two ``modes". In the ``normal" mode, the frequency of items from any of the first 50 classes is five times higher than that of items in any of the second 50 classes. In the ``abnormal" mode, the frequencies are five times lower. Thus the frequent and infrequent classes switch roles at a mode change. We generate each data point by randomly choosing a ground-truth class $c_i$ with centroid $(x_i,y_i)$ according to relative frequencies that depend upon the current mode, and then generating the data point's $(x,y)$ coordinates independently as samples from $N(x_i,1)$ and $N(y_i,1)$. Here $N(\mu,\sigma)$ denotes the normal distribution with mean $\mu$ and standard deviation $\sigma$.

In this experiment, the batch sizes are deterministic with $b=100$ items, and $k=7$ neighbors for the kNN classifier. The reservoir size for both R-TBS and Unif is 1000, and SW contains the last 1000 items; thus all methods use the same amount of data for retraining.
(We choose this value because it achieves near maximal classification accuracies for all techniques. In general, we choose sampling and ML parameters to achieve good learning performance while ensuring fair comparisons.)
In each run, the sample is warmed up by processing $100$ normal-mode batches before the classification task begins. Our experiments focus on two types of temporal patterns in the data, as described below.

\textbf{Single change:} Here we model the occurrence of a singular event. The data is generated in normal mode up to $t=10$ (time is measured here in units after warm-up), then switches to abnormal mode, and finally at $t=20$ switches back to normal (Figure~\ref{fig:single}). As can be seen, the misclassification rate (percentage of incorrect classifications) with R-TBS, SW and Unif all increase from around 18\% to roughly 50\% when the distribution becomes abnormal. Both R-TBS and SW adapt to the change, recovering to around 16\% misclassification rate after $t=16$, with SW adapting slightly better. In comparison, Unif does not adapt at all. But, when the distribution snaps back to normal, the error rate of SW rises sharply to 40\% before gradually recovering, whereas R-TBS error rate stays low around 15\% throughout. These results prove that R-TBS is indeed more robust: although slightly more sluggish than SW in adapting to changes, R-TBS avoids wild fluctuations in classification error as with SW.

\begin{figure}[bht]
	\centering
	\subfigure[Single Event]{
		\label{fig:single}\includegraphics[width=0.45\linewidth]{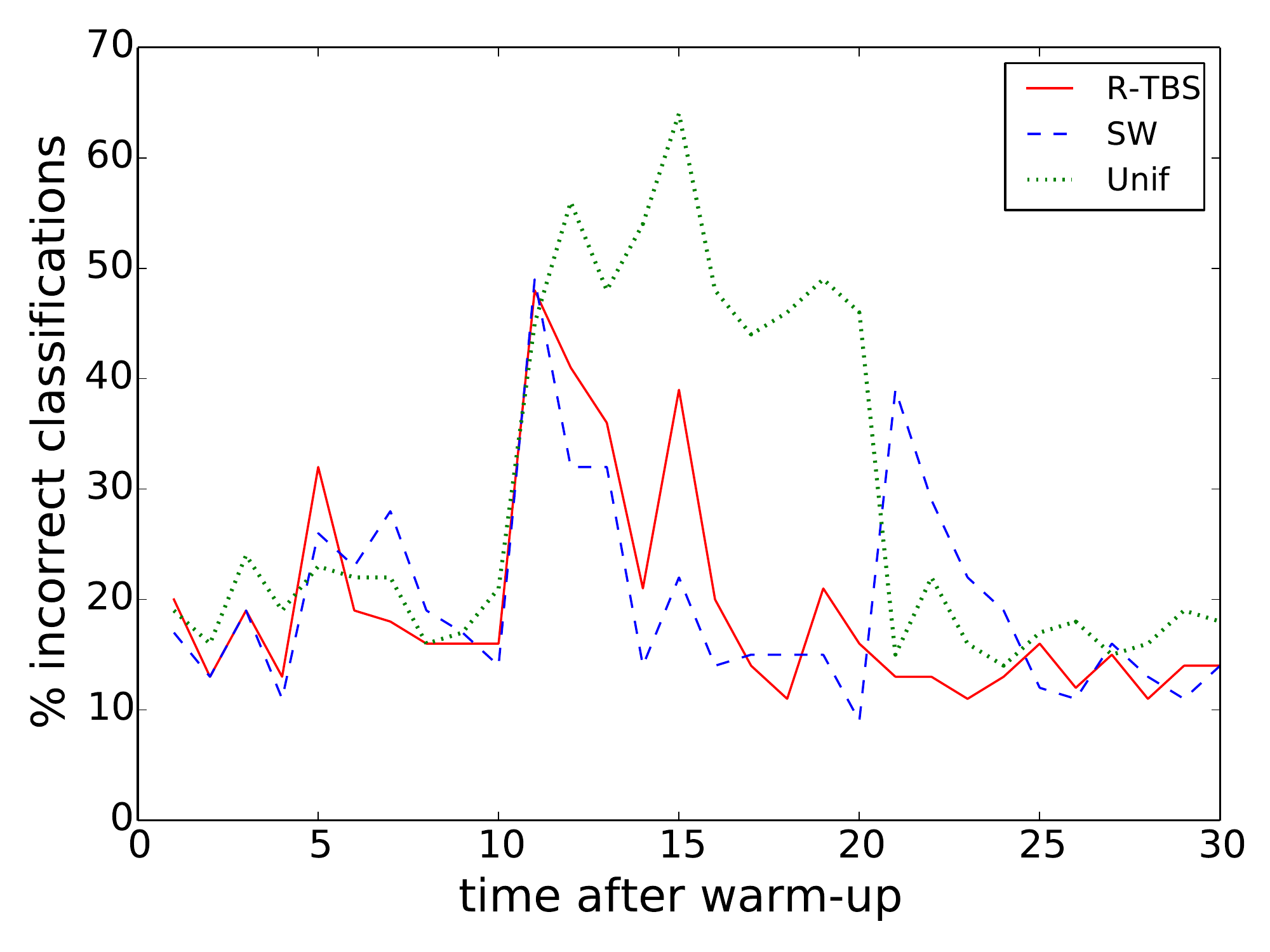}} 
	\subfigure[Periodic (10, 10)]{
		\label{fig:per1010}\includegraphics[width=0.45\linewidth]{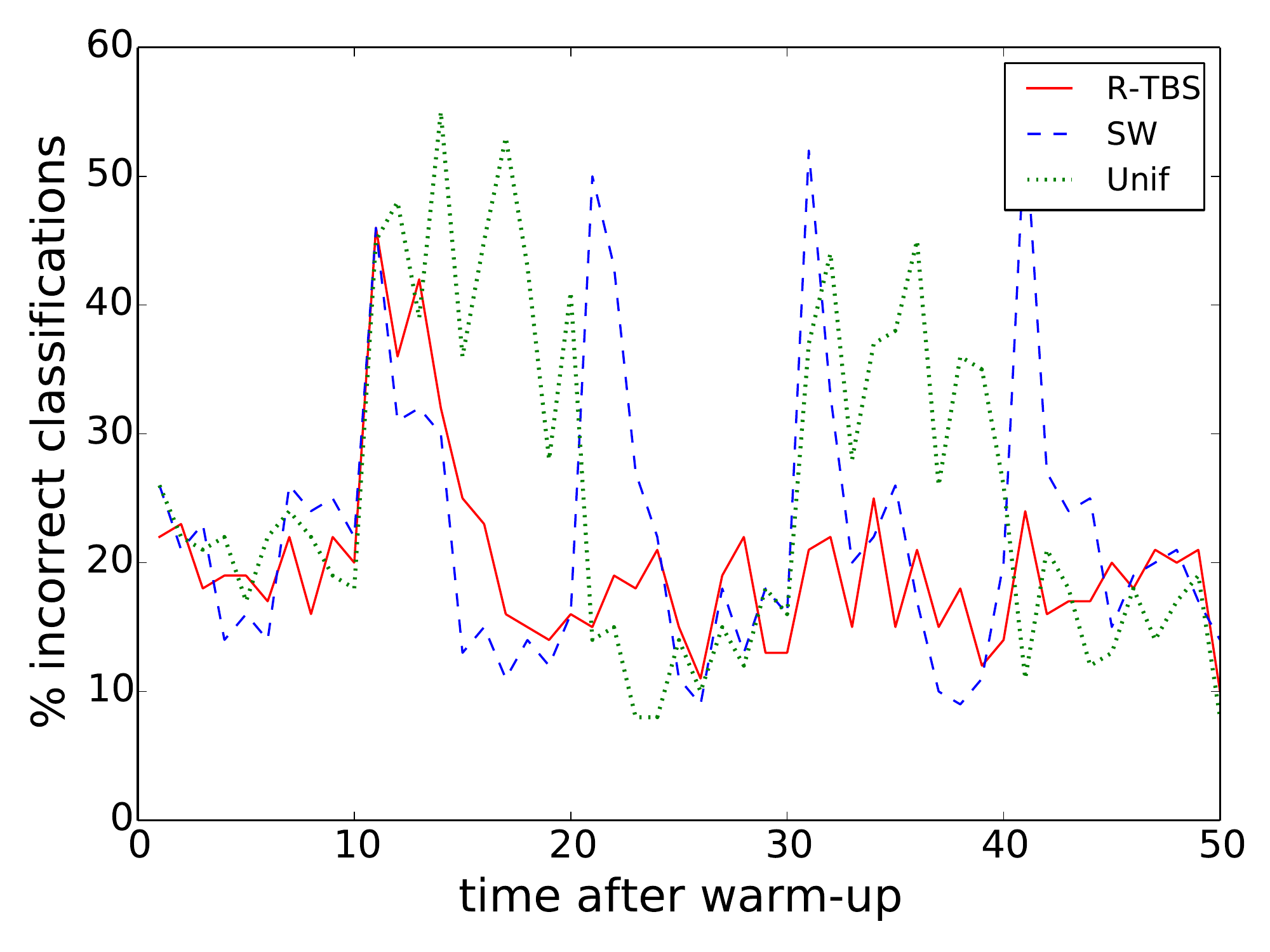}} 
	\BigCrunch
	\caption{\label{fig:accuracy}Misclassification rate (percent) for kNN}
	%\BigCrunch
	\SmallCrunch
\end{figure}

\textbf{Periodic change:} For this temporal pattern, the changes from normal to abnormal mode are periodic, with $\delta$ normal batches alternating with $\eta$ abnormal batches, denoted as $\text{Periodic}(\delta, \eta)$, or $P(\delta, \eta)$ for short. Figures~\ref{fig:per1010} shows the misclassification rate for $\text{Periodic}(10, 10)$. Experiments on other periodic patterns (in Appendix~\ref{sec:more-expmts}\longPaper) demonstrate similar results. The robust behavior of R-TBS described above manifests itself even more clearly in the periodic setting. Note, for example, how R-TBS reacts significantly better to the renewed appearances of the abnormal mode. Observe that the first 30 batches of $\text{Periodic}(10, 10)$ display the same behavior as in the single event experiment in Figure~\ref{fig:single}. We therefore focus primarily on the $\text{Periodic}(10, 10)$ temporal pattern for the remaining experiments.

%\textbf{Periodic change:} In Figures~\ref{fig:per1010}, the changes from normal to abnormal mode are periodic with 10, 20, or 30 normal batches alternating with 10 abnormal batches. The robust behavior described above manifests itself even more clearly in the periodic setting. Note, for example, how R-TBS reacts significantly better to the renewed appearances of the abnormal mode. 

\textbf{Robustness and Effect of Decay Parameter:} In the context of online model management, we need a sampling scheme that delivers high overall prediction accuracy and, perhaps even more importantly, robust prediction performance over time. Large fluctuations in the accuracy can pose significant risks in applications, e.g., in critical IoT applications in the medical domain such as monitoring glucose levels for predicting hyperglycemia events. To assess the robustness of the performance results across different sampling schemes, we use a standard risk measure called \emph{expected shortfall (ES)} \cite[p.~70]{es}. ES measures downside risk, focusing on worst-case scenarios. Specifically, the $z$\% ES is the average value of the worst $z$\% of cases. 

For each of 30 runs and for each sampling scheme, we compute the 10\% ES of the misclassification rate (expressed as a percentage) starting from $t=20$, since all three sampling schemes perform poorly (as would be expected) during the first mode change, which finishes at $t=20$. Table~\ref{tab:lambda} lists both the \emph{accuracy}, measured in terms of the average misclassification rate, and the \emph{robustness}, measured as the average 10\% ES, of the kNN classifier over 30 runs across different temporal patterns. To demonstrate the effect of the decay parameter $\lambda$ on model performance, we also include numbers for different $\lambda$ values in Table~\ref{tab:lambda}. 

In terms of accuracy, Unif is always the worst by a large margin. R-TBS and SW have similar accuracies, with R-TBS having a slight edge in most cases. On the other hand, for robustness, SW is almost always the worst, with ES ranging from 1.4x to 2.7x the maximum ES (over different $\lambda$ values) of R-TBS. Mostly, Unif is also significantly worse than R-TBS, with ES ratios ranging from 1.4x to 1.7x. The only exception is the single-event pattern: since the data remains in normal mode after the abnormal period, time biasing becomes unimportant and Unif performs well. In general, R-TBS provides both better accuracy and robustness in almost all cases. The relative performance of the sampling schemes in terms of accuracy and robustness tend to be consistent across temporal patterns. Table~\ref{tab:lambda} also shows that different $\lambda$ values affect the accuracy and robustness, however, R-TBS provides superior results over a fairly wide range of $\lambda$ values. 

%As discussed in Section~\ref{sec:background}, the choice of the decay parameter $\lambda$ is application specific. Typically, application domains with faster pace of changes need greater values of $\lambda$ for faster adaption. To measure the effect of different $\lambda$ values on model quality, we rerun the previous experiments with various decay parameters. For each experiment, we assess both the overall accuracy and the robustness of the classifier. The accuracy is measured using the percentage of incorrect classifications aggregated across all batches of a experimental run, whereas for robustness we use the standard risk measure called \emph{10\% expected shortfall} (see Page 70 of \cite{es}). {\color{red}Expected Shortfall should receive an explanation here. As well, we need to touch on that this is the average of 30 iterations, and that we are starting the measurement of ES after t=20.} 
 
\textbf{Varying batch size:} We now examine model quality when the batch sizes are no longer constant. Overall, the results look similar to those for constant batch size. For example, Figure~\ref{fig:unif} shows results for a Uniform(0,200) batch-size distribution, and Figure~\ref{fig:batchgrow} shows results for a deterministic batch size that grows at a rate of 2\% after warm-up. In both experiments, $\lambda=0.07$ and the data pattern is $\text{Periodic}(10, 10)$. These figures demonstrate the robust performance of R-TBS in the presence of varying data arrival rates. Similarly, the average accuracy and robustness over 30 runs resembles the results in Table~\ref{tab:lambda}. For example, pick $\lambda=0.07$ and a $\text{Periodic}(10, 10)$ pattern. Then, the misclassification rate under uniform/growing batch sizes is 1.16x/1.14x that of R-TBS for SW, and 1.47x/1.40x for Unif. In addition, the ES is 1.82x/1.98x that of R-TBS for SW, and 1.76x/1.78x for Unif.

\begin{table}[tbh]
\caption{Accuracy and robustness of kNN performance}\label{tab:lambda}
\BigCrunch
	\scriptsize{
		\begin{tabular}[0.3\linewidth]{|c|c|c|c|c|c|c|c|c|c|c|}
			\hline
			& \multicolumn{2}{c|}{
			\textbf{Single Event}} & \multicolumn{2}{c|}{\textbf{P(10,10)}} & \multicolumn{2}{c|}{\textbf{P(20,10)}} & \multicolumn{2}{c|}{\textbf{P(30,10)}} \\ \cline{2-9}
			$\lambda$  & Miss\% & ES  & Miss\% & ES &	Miss\% & ES & Miss\% & ES \\ \hline
			0.05 & 19.8 & \textbf{17.7}  & 18.2 & 24.2 &	17.9 & 28.2 & 15.5 & 31.6 \\ \hline
			0.07 & 19.1 & 18.7 &	17.4 & \textbf{23.2} &	17.2 & \textbf{28.1} & \textbf{14.9} & \textbf{31.0} \\ \hline
			0.10 & \textbf{18.0} & 20.0 &	\textbf{16.6} & 24.1 &	\textbf{16.6} & 29.9 &	15.1 & \textbf{31.0} \\ \hline
			SW & 19.2 & 53.3 & 19.0 & 49.8 &	18.8 & 47.3 &	16.5 & 44.5 \\ \hline
			Unif & 25.6 & 19.3 &	25.4 & 42.3 &	25.0 & 43.2 &	21.0 & 47.6 \\ \hline
		\end{tabular}
	}
	\BigCrunch
	\BigCrunch
\end{table}

\begin{figure}[tbh]
	\centering
	\subfigure[Uniform Batch Size]{
		\label{fig:unif}\includegraphics[width=0.45\linewidth]{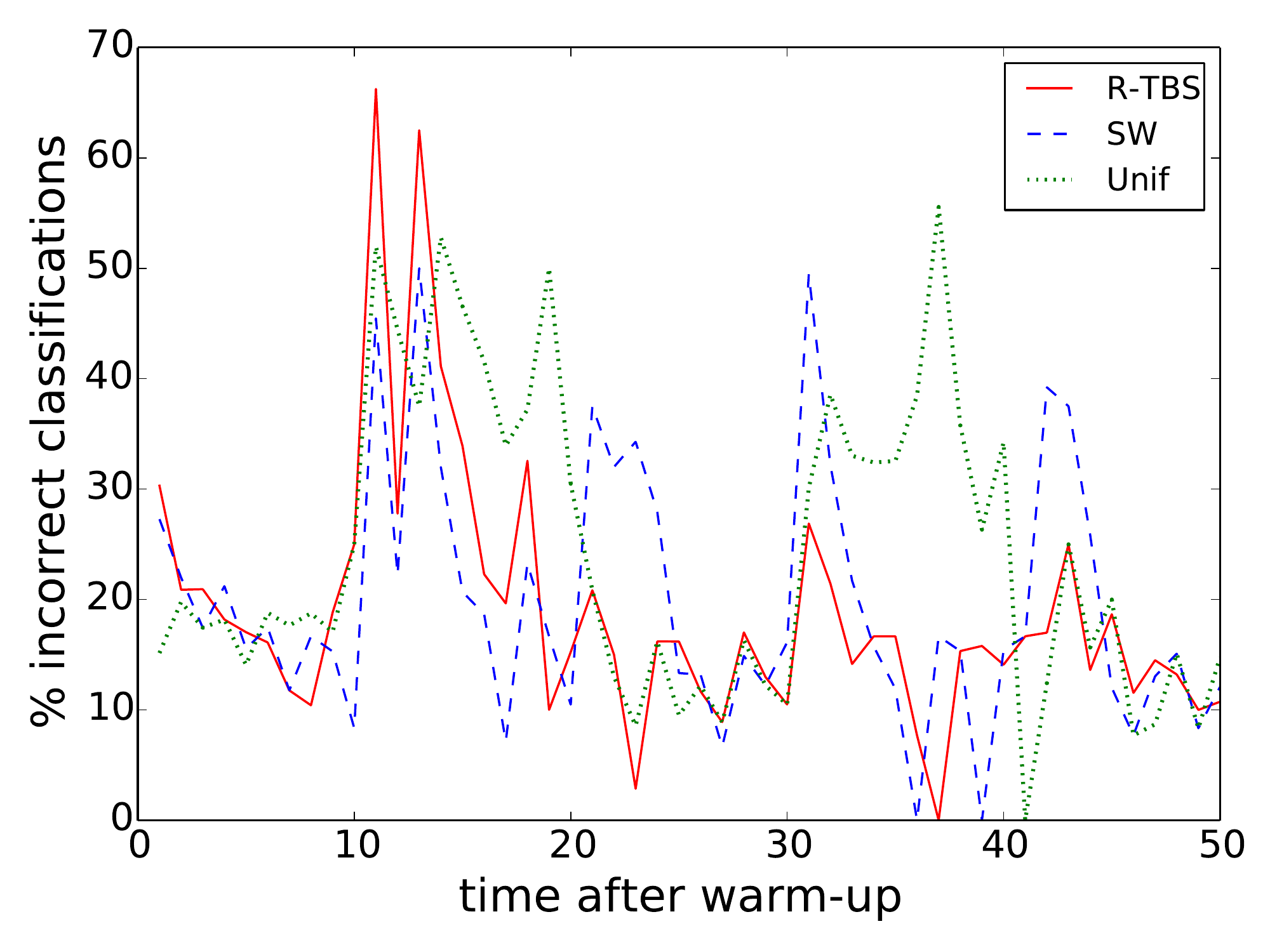}} 
	\subfigure[Growing Batch Size]{
		\label{fig:batchgrow}\includegraphics[width=0.45\linewidth]{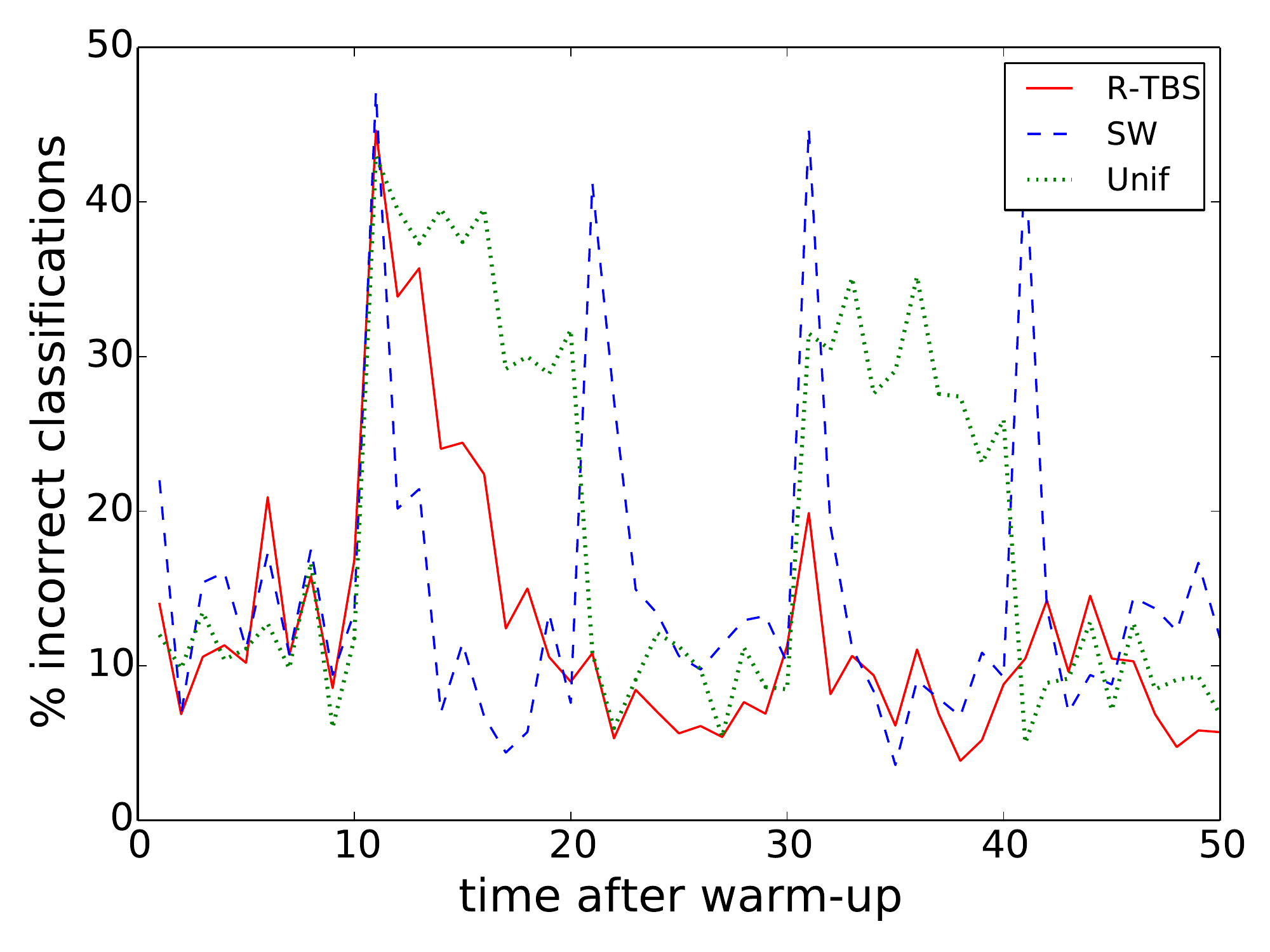}} 
	\BigCrunch
	\caption{\label{fig:rates} Varying batch sizes for kNN classifier}
	\BigCrunch
\end{figure}

%\begin{figure}[tbh]
%	\centering
%	\subfigure[Periodic (10, 10)]{
%		\label{fig:regress}\includegraphics[width=0.47\linewidth]{figs/Regress}}
%	\BigCrunch
%	\caption{\label{fig:regression} Mean square error for linear regression}
%	\BigCrunch
%\end{figure}

\subsection{Application: Linear Regression}

We now assess the effectiveness of R-TBS for retraining regression models. The experimental setup is similar to kNN, with data generated in ``normal'' and ``abnormal'' modes. In both modes, data items are generated from the standard linear regression model $y = b_1x_1 + b_2x_2 + \eps$, with the noise term $\eps$ distributed according to a $N(0,1)$ distribution. In normal mode, $(b_1,b_2)=(4.2,-0.4)$ and in abnormal mode, $(b_1,b_2)=(-3.6,3.8)$. In both modes, $x_1$ and $x_2$ are generated according to $\text{Uniform}(0,1)$ distribution. As before, the experiment starts with a warm-up of 100 ``normal'' mode batches and each batch contains 100 items. 

\textbf{Saturated samples:} Figure~\ref{fig:regression} shows the performance of R-TBS, SW, and Unif for the $\text{Periodic}(10, 10)$ pattern with a maximum sample size of 1000 for each technique. We note that, for this sample size and temporal pattern, the R-TBS sample is always saturated. (This is also true for all of the prior experiments.) The results echo that of the previous section, with R-TBS exhibiting slightly better prediction accuracy on average, and significantly better robustness, than the other methods. The mean square errors (MSEs) across all data points for R-TBS, Unif, and SW are 3.51, 4.43, 4.02 respectively, and their 10\% ES of the MSEs are 6.04, 10.05, 10.94 respectively.
%As a indicator of the overall effectiveness of each sampling technique, we give the mean square error (MSE) across all data points, as well as the 10\% ES of the MSEs. These numbers correspond respectively to 3.68 and 8.05 for R-TBS, 4.64 and 13.21 for Unif, and 4.24 and 14.83 for SW. 

\begin{figure*}[tbh]
	\centering
	\begin{minipage}[b]{.65\linewidth}
		\subfigure[n=1000, Periodic(10,10)]{
			\label{fig:regression}\includegraphics[width=1.4in]{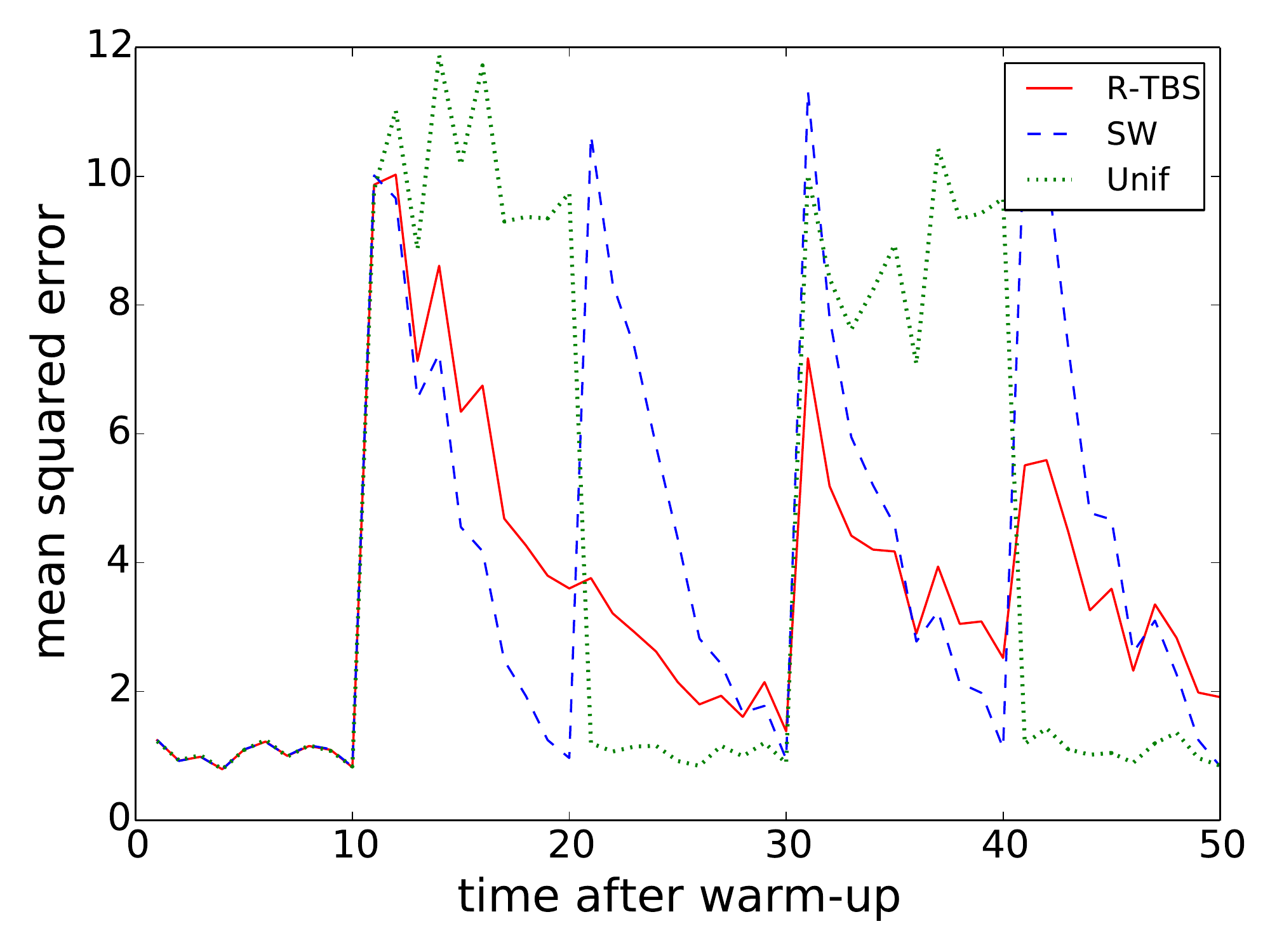}} 
		\subfigure[n=1600, Periodic(10,10)]{
			\label{fig:regression1500_10}\includegraphics[width=1.4in]{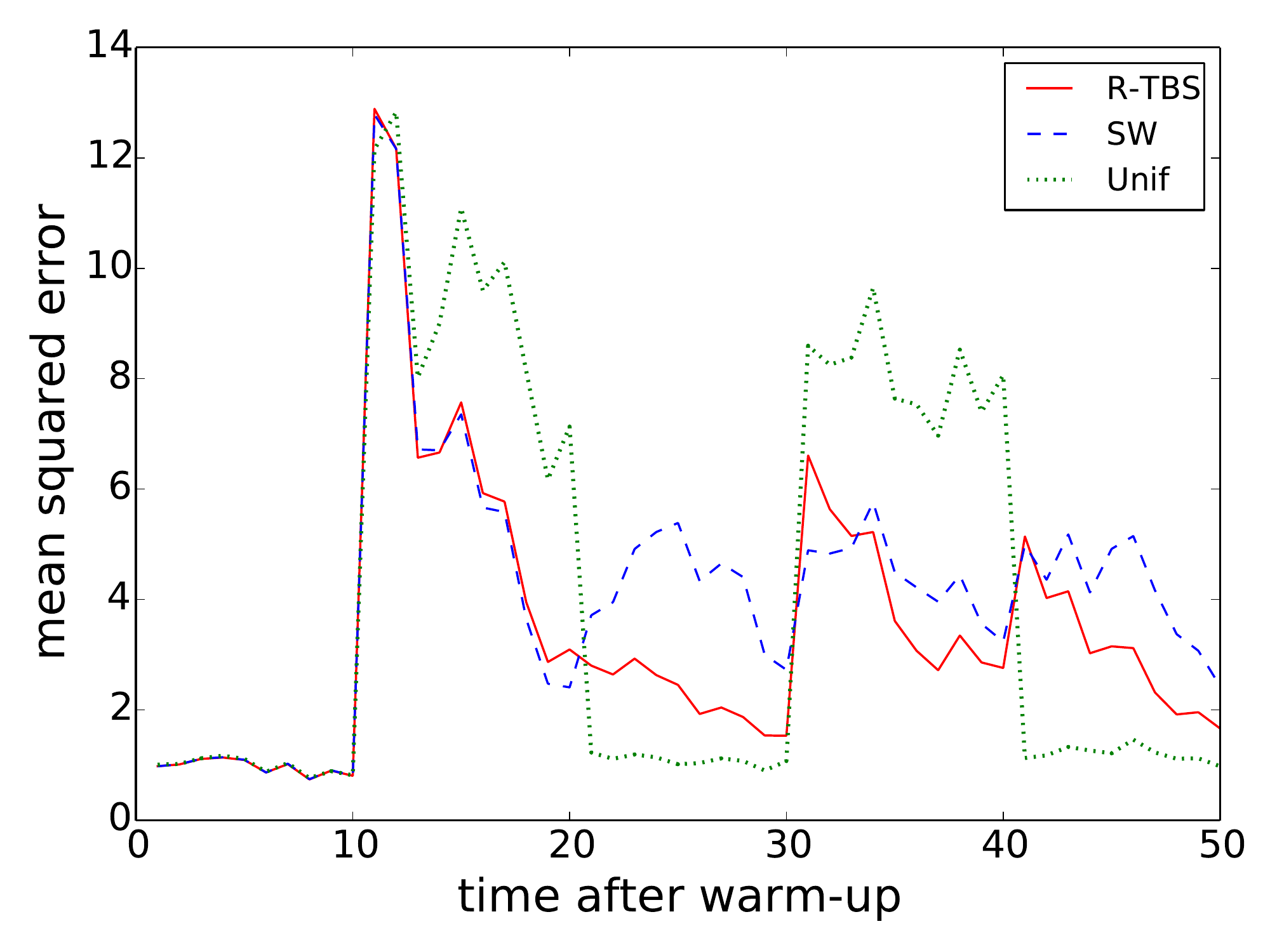}}
		\subfigure[n=1600, Periodic(16,16)]{
			\label{fig:regression1500_15}\includegraphics[width=1.4in]{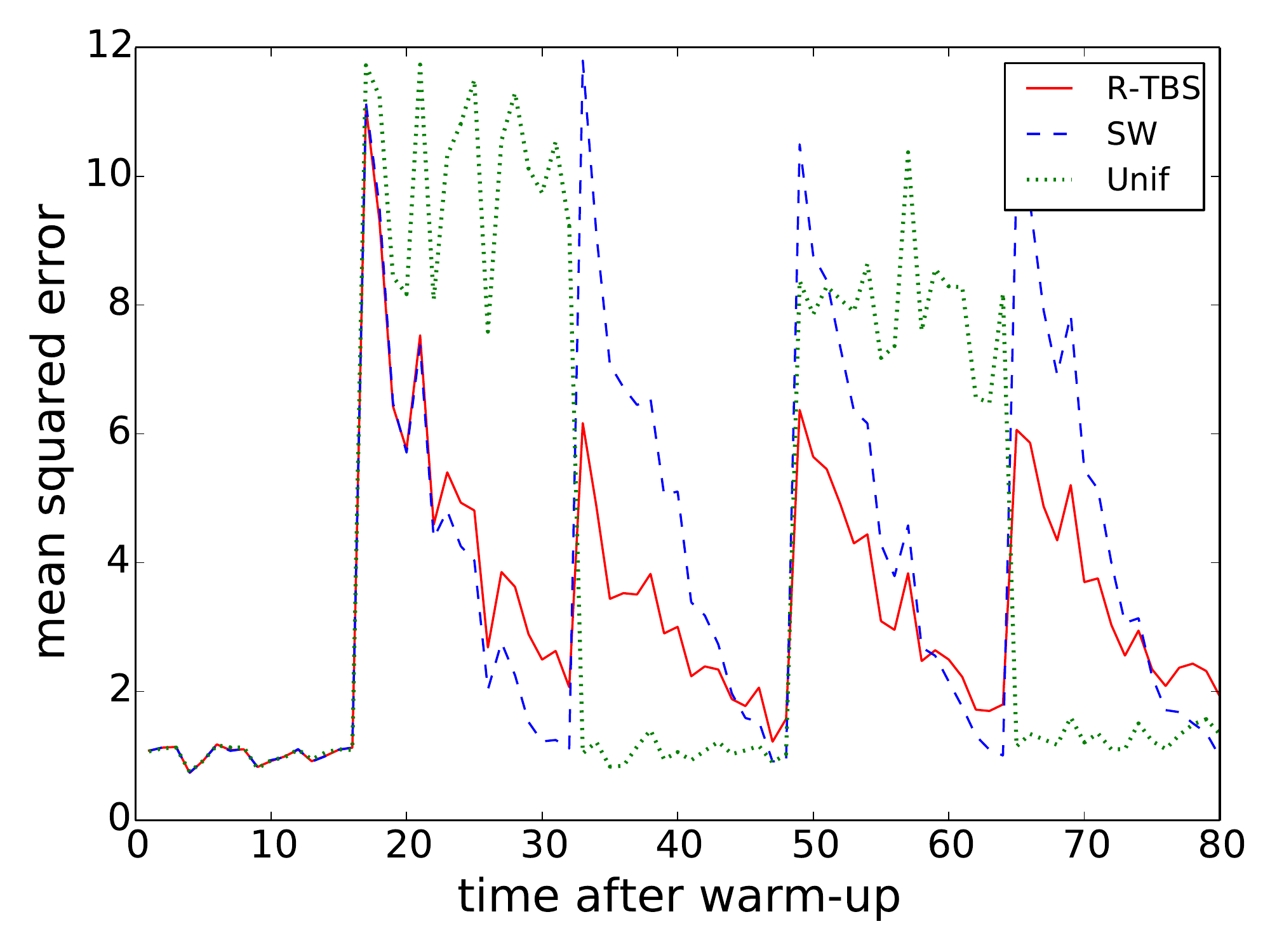}}  
		\BigCrunch
		\caption{Mean square error for linear regression}
	\end{minipage}\qquad
	\begin{minipage}[b]{.28\linewidth}
		\includegraphics[width=1.4in]{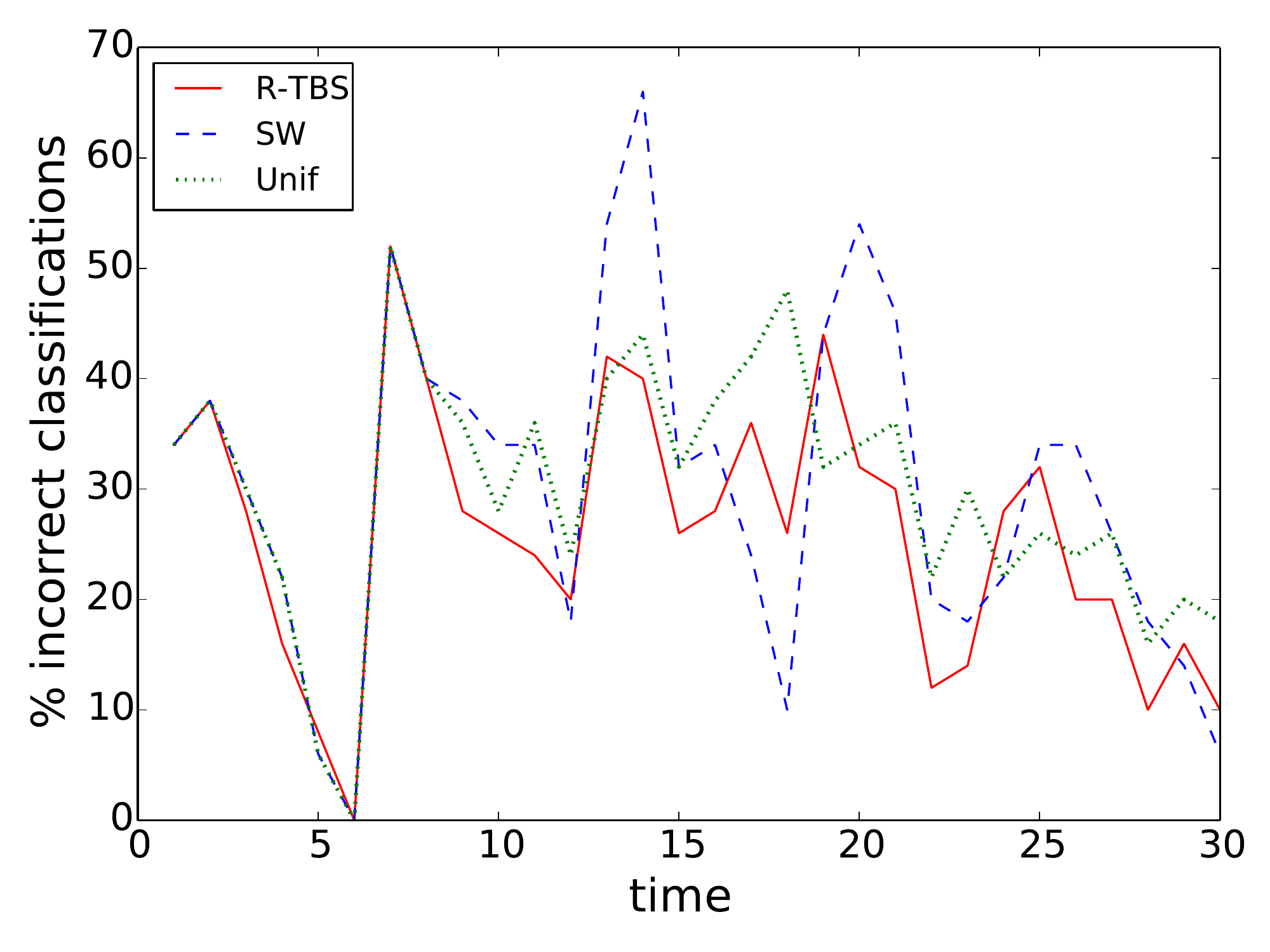} 
	\SmallCrunch
		\caption{Misclassification rate (percent) for Naive Bayes}\label{fig:NBaccuracy}
	\end{minipage}
	\BigCrunch
\end{figure*}

%\begin{figure*}[tbh]
%	\centering
%	\subfigure[n=1000, Periodic(10,10)]{
%		\label{fig:regression}\includegraphics[width=0.22\linewidth]{figs/Regress1000}} 
%	\subfigure[n=1600, Periodic(10,10)]{
%		\label{fig:regression1500_10}\includegraphics[width=0.22\linewidth]{figs/Regression1500_10period}}
%	\subfigure[n=1600, Periodic(16,16)]{
%		\label{fig:regression1500_15}\includegraphics[width=0.22\linewidth]{figs/Regression1500_15period}}  
%	\BigCrunch
%	\caption{Mean square error for linear regression}
%	%\BigCrunch
%	\SmallCrunch
%\end{figure*}

\input{last-expmt}

\subsection{Application: Naive Bayes}

In our final experiment, we evaluate the performance of R-TBS for retraining Naive Bayes models with the Usenet2 dataset (\url{mlkd.csd.auth.gr/concept_drift.html}), which was used in~\cite{Katakis2007} to study classifiers coping with recurring contexts in data streams. This dataset contains a stream of 1500 messages on different topics from the 20 News Groups Collections~\cite{Lichman2013}. They are sequentially presented to a simulated user who marks whether a message is interesting or not. The user's interest changes after every 300 messages. More details of the dataset can be found in~\cite{Katakis2007}.

%We next assess the performance of R-TBS on two combination real-synthetic datasets that exhibits Concept Drift. The two datasets both contain streams of messages and articles from the 20 News Groups Collection\cite{Lichman2013}. The messages are presented to the user sequentially and the user responds with whether they are interested or not. The goal is to predict the user's interest. While the messages in the dataset are real, the users interests in the articles are generated synthetically, and have been used in prior studies on Machine Learning under Concept Drift\cite{Katakis2007, Katakis2010}. In the labelling of interests, the user is given changing interests over time, with their interests changing by message number as shown in Table \ref{tab:interest}.

%\begin{table}[h]
%	\caption{Interest in topics for user over time}\label{tab:interest}
%	\BigCrunch
%	\scriptsize
%	\begin{tabular}{|c | c | c | c | c | c | } \hline
%		Topic & 1-300 & 301-600 & 601 - 900 & 901-1200 & 1201-1500 \\ \hline
%		\multicolumn{6}{|c|}{NewsGroup1} \\ \hline
%		Medicine & \checkmark & X & X & X & \checkmark \\ \hline
%		Baseball & X & \checkmark & X & \checkmark & X \\ \hline
%		Space & X & X & \checkmark & X & X \\ \hline
%		\multicolumn{6}{|c|}{NewsGroup2} \\ \hline
%		Medicine & \checkmark & X & \checkmark & X & \checkmark \\ \hline
%		Baseball & X & \checkmark & X & \checkmark & X \\ \hline
%		Space & X & \checkmark & X & \checkmark & X \\ \hline
%	\end{tabular}
%	\BigCrunch
%\end{table}
%\normalsize

Following~\cite{Katakis2007}, we use Naive Bayes with a bag of words model, and set the optimal parameters for SW with maximum sample size of 300 and batch size of 50. Since this dataset is rather small and contexts change frequently, we use the optimal value of $0.3$ for $\lambda$. We find through experiments that R-TBS displays higher prediction accuracy for all $\lambda$ in the range of $[0.1,0.5]$, so precise tuning of $\lambda$ is not critical. In addition, there is not enough data to warm up the models on different sampling schemes, so we report the model performance on all the 30 batches.
%(instead of reporting the numbers after warm-up as we did for the synthetic data).
Similarly, we report 20\% ES for this dataset, due to the limited number of batches.

The results are shown in Figure~\ref{fig:NBaccuracy}. The misprediction rate for R-TBS, SW, and Unif are 26.5\%, 30.0\%, and 29.5\%; and the 20\% ES values are 43.3\%, 52.7\%, and 42.7\%. Importantly, for this dataset the changes in the underlying data patterns are less pronounced than in the previous two experiments. Despite this, SW fluctuates wildly, yielding inferior accuracy and robustness. In contrast, Unif barely reacts to the context changes. As a result, Unif is very slightly better than R-TBS with respect to robustness, but at the price of lower overall accuracy. Thus, R-TBS is generally more accurate under mild fluctuations in data patterns, and its superior robustness properties manifest themselves as the changes become more pronounced.

%% file: last-expmt.tex
% !TEX root = tbs-mm.tex
\textbf{Unsaturated Samples:} We now investigate the case of unsaturated samples for R-TBS. We increase the target sample size to $n=1600$. With a constant batch size of 100, and a decay rate $\lambda=0.07$, the reservoir of R-TBS is never full, stabilizing at 1479 items, whereas Unif and SW both have a full sample of 1600 items. 
%(Note that we use a constant batch size to make sure that the reservoir is unsaturated at all times.)

For the Periodic$(10, 10)$ pattern, shown in Figure~\ref{fig:regression1500_10}, SW has a window size large enough to keep some data from older time periods (up to 16 batches ago), making SW's robustness comparable to R-TBS (ES of 5.86 for SW and 5.97 for R-TBS). However, this amalgamation of old data also hurts its overall accuracy, with MSE rising to 4.17, as opposed to 3.50 for R-TBS. In comparison, the shape of R-TBS remains almost unchanged from Figure~\ref{fig:regression}, and Unif behaves as poorly as before. When the pattern changes to Periodic$(16, 16)$ as shown in Figure~\ref{fig:regression1500_15}, SW doesn't contain enough old data, making its prediction performance suffer from huge fluctuations again, and the superiority of R-TBS is more prominent. In both cases, R-TBS provides the best overall performance, despite having a smaller sample size. This backs up our earlier claim that more data is not always better. A smaller but more balanced sample with good ratios of old and new data can provide better prediction performance than a large but unbalanced sample.

%% file: relwork.tex
\section{Related Work}\label{sec:relwork}

\textbf{Time-decay and sampling:} Work on sampling with unequal probabilities goes back to at least Lahiri's 1951 paper~\cite{Lahiri51}.
%see~\cite[Section~4]{OlkenR95} for some additional discussion of early work.
A growing interest in streaming scenarios with weighted and decaying items began in the mid-2000's, with most of that work focused on computing specific aggregates from such streams, such as heavy-hitters, subset sums, and quantiles; see, e.g., \cite{AlonDLT05a,CohenS06,CormodeKT08}. The first papers on time-biased reservoir sampling with exponential decay are due to Aggarwal~\cite{Aggarwal06} and Efraimidis and Spirakis~\cite{efraimidisS06}; batch arrivals are not considered in these works. As discussed in Section~\ref{sec:intro}, the sampling schemes in~\cite{Aggarwal06} are tied to item sequence numbers rather than the wall clock times on which we focus; the latter are more natural when dealing with time-varying data arrival rates.
%(Recall from Section~\ref{sec:background} that we state our algorithms for arrivals at integer-valued times for simplicity only.)

Cormode et al.~\cite{CormodeSSX09}
%make the simple but important observation that any recursively-defined weighted sampling scheme for a finite population can be converted to a streaming algorithm with decay by using a weight function that is non-decreasing over time; for exponential decay, this simply means using the weight function $w(t)=e^{\lambda t}$. A potential drawback to this strategy is the need for continual normalization computations to avoid numerical issues due to very large weights. In any case, there does not seem to be a recursive weighted scheme for which the strategy yields an algorithm satisfying our requirements. For example, the authors of~\cite{CormodeSSX09}
propose a \eat{specific} time biased reservoir sampling algorithm based on the A-Res weighted sampling scheme proposed in~\cite{efraimidisS06}. Rather than enforcing \eqref{eq:expratio}, the algorithm enforces the (different) A-Res biasing scheme.
%As discussed in \cite{Efraimidis15}, A-Res imposes a certain constraint on the initial item acceptance probabilities (the ``second intepretation'' of weights given in \cite{Efraimidis15}).
In more detail, if $s_{i}$ denotes the element at slot $i$ in the reservoir, then the algorithm in~\cite{efraimidisS06} implements a scheme where an item $x$ is chosen to be at slot $i + 1$ in the reservoir with probability $w_x/(\sum_{j= 1}^x w_j - \sum_{j=1}^{i} w_{s_j})$. From the form of this equation, it becomes clear that resulting sampling algorithm violates \eat{the relative inclusion relationship in}\eqref{eq:expratio}. Indeed, Efraimidis~\cite{Efraimidis15} gives some numerical examples illustrating this point (in his comparison of the A-Res and A-Chao algorithms). Again, we would argue that the constraint on appearance probabilities in \eqref{eq:expratio} is easier to understand in the setting of model management than the foregoing constraint on initial acceptance probabilities. 

The closest solution to ours adapts the weight\-ed sampling algorithm of Chao~\cite{Chao82} to batches and time decay; we call the resulting algorithm B-Chao and describe it in Appendix~\ref{sec:chao}\longPaper. Unfortunately, as discussed, the relation in \eqref{eq:expratio} is violated both during the initial fill-up phase and whenever the data arrival rate becomes slow relative to the decay rate, so that the sample contains ``overweight'' items. Including overweight items causes over-representation of older items, thus potentially degrading predictive accuracy.
%(As indicated in Appendix~\ref{sec:chao}\longPaper, the processing of overweight items incurs more complexity and cost than is generally discussed in the literature.) 
The root of the issue is that the sample size is nondecreasing over time. The R-TBS algorithm is the first algorithm to correctly (and optimally) deal with ``underflows'' by allowing the sample to shrink---thus handling data streams whose flow rates vary unrestrictedly over continuous time. The current paper also explicitly handles batch arrivals and explores parallel implementation issues. The VarOpt sampling algorithm of Cohen et al.~\cite{CohenDKLT11}---which was developed to solve the specific problem of estimating ``subset sums''---can also be modified to our setting. The resulting algorithm is more efficient than Chao, but as stated in \cite{CohenDKLT11}, it has the same statistical properties, and hence does not satisfy \eqref{eq:expratio}.

\textbf{Model management:} A key goal of our work is to support model management; see~\cite{GamaZBPB14} for a survey on methods for detecting changing data---also called ``concept drift'' in the setting of online learning---and for adapting models to deal with drift. As mentioned previously, one possibility is to re-engineer the learning algorithm. This has been done, for example, with support-vector machines (SVMs) by developing incremental versions of the basic SVM algorithm~\cite{CauwenberghsP00} and by adjusting the training data in an SVM-specific manner, such as by adjusting example weights as in Klinkenberg~\cite{Klinkenberg04}. Klinkenberg also considers using curated data selection to learn over concept drift, finding that weighted data selection also improves the performance of learners. Our approach of model retraining using time-biased samples follows this latter approach, and is appealing in that it is simple and applies to a large class of machine-learning models. The recently proposed Velox system for model management~\cite{CrankshawBGLZFG15} ties together online learning and statistical techniques for detecting concept drift. After detecting drift through poor model performance, Velox kicks off batch learning algorithms to retrain the model. Our approach to model management is complementary to the work in~\cite{CrankshawBGLZFG15} and could potentially be used in a system like Velox to help deployed models recover from poor performance more quickly. The developers of the recent MacroBase system~\cite{BailisGMNRS17} have incorporated a time-biased sampling approach to model retraining, for identifying and explaining outliers in fast data streams. MacroBase essentially uses Chao's algorithm, and so could potentially benefit from the R-TBS algorithm to enforce the inclusion criterion~\eqref{eq:expratio} in the presence of highly variable data arrival rates.

%% file: concl.tex
% !TEX root = tbs-mm.tex
% above command is for TeXShop

\section{Conclusion}\label{sec:concl}

Our experiments with classification and regression algorithms, together with the prior work on graph analytics in \cite{XieTSBH15}, indicate the potential usefulness of periodic retraining over time-biased samples to help ML algorithms deal with evolving data streams without requiring algorithmic re-engineering. To this end we have developed and analyzed several time-biased sampling algorithms that are of independent interest. In particular, the R-TBS algorithm allows simultaneous control of both the item-inclusion probabilities and the sample size, even when the data arrival rate is unknown and can vary arbitrarily. R-TBS also maximizes the expected sample size and minimizes sample-size variability over all possible bounded-size algorithms with exponential decay. Using techniques from \cite{CormodeSSX09}, we intend to generalize these properties of R-TBS to hold under arbitrary forms of temporal decay.  

We have also provided techniques for distributed implementation of R-TBS and T-TBS, and have shown that use of time-biased sampling together with periodic model retraining can improve model robustness in the face of abnormal events and periodic behavior in the data. In settings where (i) the mean data arrival rate is known and (roughly) constant, as with a fixed set of sensors, and (ii) occasional sample overflows can be easily dealt with by allocating extra memory, we recommend use of T-TBS to precisely control item-inclusion probabilities. In many applications, however,  we expect that either (i) or (ii) will violated, in which case we recommend the use of R-TBS. Our experiments showed that R-TBS is superior to sliding windows over a range of $\lambda$ values, and hence does not require highly precise parameter tuning; this may be because time-biased sampling avoids the all-or-nothing item inclusion mechanism inherent in sliding windows. %In future work, we intend to integrate our sampling technology in an end-to-end model management system.

%% file: th1prf.tex
% !TEX root = tbs-mm.tex
% above command is for TeXShop
\appendix

\section{Bernoulli time-biased sampling}\label{sec:BTBS}

\begin{algorithm}[h]
\caption{Bernoulli time-biased sampling (B-TBS)}\label{alg:bernsamp}
{\footnotesize
$\lambda$: decay factor ($\ge 0$)
\BlankLine
Initialize: $S\gets S_0$; $p\gets e^{-\lambda}$\Comment*[r]{$p=$ retention prob.}
\For{$t\gets1,2,\ldots$}{
$M \gets \textsc{Binomial}(|S|,p)$\Comment*[r]{simulate $|S|$ trials}\label{ln:binom}
$S \gets \textsc{Sample}(S,M)$\Comment*[r]{retain $M$ random elements}
$S\gets S\cup \xB_t$\;\label{ln:accept}
output $S$}
}
\end{algorithm}

Al\-go\-ri\-thm~\ref{alg:bernsamp} implements the time-biased Bernoulli sampling scheme discussed in Section~\ref{sec:background}, adapted to batch arrivals. As described previously, at each time~$t$ we accept each incoming item $x\in \xB_t$ into the sample with probability~1 (line~\ref{ln:accept}). At each subsequent time $t'>t$, we flip a coin independently for each item currently in the sample: an item is retained in the sample with probability~$p$ and removed with probability~$1-p$. As for the T-TBS algorithm,  we simulate the $|S|$ coin tosses by directly generating the number of successes according to a binomial distribution.

From the intuitive description of the algorithm, we see that, at time~$t'=t+k$ (where $k\ge 0$), we have for $x\in \xB_t$ that 
\begin{equation}\label{eq:bern1}
\begin{split}
\prob{x\in S_{t'}}
&=\prob{x\in S_t}\times\prod_{i=1}^k \prob{x\in S_{t+i}\mid x\in S_{t+i-1}}\\
&= 1\times p^k=e^{-\lambda k} = e^{-\lambda(t'-t)},
\end{split}
\end{equation}
where we have used the fact that the sequence of samples is a set-valued Markov process, so that
\[
\begin{split}
&\prob{x\in S_{t+i}\mid x\in S_{t+i-1}, x\in S_{t+i-2},\ldots,x\in S_t}\\
&\quad = \prob{x\in S_{t+i}\mid x\in S_{t+i-1}} 
\end{split}
\]
for $1\le i\le k$.  This is essentially the algorithm used, e.g., in \cite{XieTSBH15} to implement time-biased edge sampling in dynamic graphs.

Since \eqref{eq:expratio} follows immediately from \eqref{eq:bern1}, we see that Algorithm~\ref{alg:bernsamp} precisely controls the relative inclusion probabilities via the decay rate $\lambda$.
%The algorithm handles an out-of-order item that arrives at time $t'$ but has a timestamp of $t<t'$ by initially accepting it with probability $p^{t'-t}$, the probability that the item would appear in the sample at time $t'$ if it had been inserted at the correct time~$t$. Although there may be additional data items with timestamps less than $t'$ that have not yet arrived, the above procedure allows for analytics that are as correct as possible given the data at hand.

%The user, however, cannot independently control the expected sample size, which is completely determined by $\lambda$ and the sizes of the incoming batches. In particular, if the batch sizes systematically grow over time, then sample size will grow without bound. Arguments in \cite{XieTSBH15} show that if $\sup_t |\xB_t|<\infty$, then the sample size can be bounded, but only probabilistically. See Remark~\ref{rem:alg1} below for extensions and refinements of these results. 

\section{Batched Reservoir Sampling}\label{sec:BRS}

\begin{algorithm}[tbh]
\caption{Batched reservoir sampling (B-RS)}\label{alg:rs}
{\small
$n$: maximum sample size
\BlankLine
Initialize: $S\gets S_0$; $W\gets |S_0|$;\Comment*[r]{$|S_0|\le n$}
\For{$t=1,2,\ldots$}{
$C=\min(n,W+|\xB_t|)$\Comment*[r]{new sample size}\label{ln:bddsize}
$M \gets \textsc{HyperGeo}\bigl(C,|\xB_t|,W\bigr)$\;\label{ln:hyper}
\Comment{add $M$ elements to $S$,}
\Comment{overwrite $\max(|S|+M-n,0)$ elements}
$S\gets \textsc{Sample}\bigl(S,\min(n-M,|S|)\bigr)\cup \textsc{Sample}(\xB_t,M)$\;\label{ln:haccept}
$W\gets W+|\xB_t|$\;
output $S$}
}
\end{algorithm}

Algorithm~\ref{alg:rs} is the classical res\-er\-voir-sampling algorithm, modified to handle batch arrivals. (To our knowledge, this variant has not appeared previously in the literature.) Note that the sample size is bounded above by $n$ (line~\ref{ln:bddsize}). This algorithm, although bounding the sample size, does not allow time-biased sampling or, equivalently, only supports decay rate $\lambda= 0$. That is, at any given time, all items seen so far are equally likely to be in the sample; see \eqref{eq:expratio}. In the algorithm, \textsc{Sample} is defined as before and $\textsc{HyperGeo}(k,a,b)$ returns a sample from the hypergeometric$(k,a,b)$ distribution having probability mass function $p(n)=\binom{a}{n}\binom{b}{k-n}/\binom{a+b}{k}$ if $\max(0,k-b)\le n\le\min(a,k)$ and $p(n)=0$ otherwise; see \cite{KacS85} for a discussion of efficient implementations of \textsc{HyperGeo}.

As with the standard reservoir algorithm, Algorithm~\ref{alg:rs} enforces a stronger property at each time point $t$ than merely requiring that all marginal inclusion probabilities are equal: it is a \emph{uniform} sampling scheme. Specifically, if $U_t$ denotes, as before, the set of all items seen through time~$t$, then all possible samples of size $C_t=\min(n,W_t)$ are equally likely, where $W_t=|U_t|$  is the number of items seen through time~$t$.

To see this, observe that this property holds trivially for $S_0$. Assuming for induction that $S_{t-1}$ is a uniform sample from $U_{t-1}$, suppose that we execute $|\xB_t|$ steps of the standard reservoir algorithm. As shown, e.g., in \cite{HaasStr16}, the resulting sample $S'_t$ is a uniform sample of $U_t$. Thus the number of $\xB_t$ items in $S'_t$ has a hypergeometric$(|S_t|,|\xB_t|,W_{t-1})$ distribution, whereas the remaining items in the sample comprise a uniform sample from $U_{t-1}$. Algorithm~\ref{alg:rs} simulates this sequential execution of the standard reservoir algorithm by directly generating $S_t$ so that it has the same distribution as $S'_t$. In detail, the correct number of items $M=|S_t\cap \xB_t|$ to include from $\xB_t$ is selected according to the appropriate hypergeometric distribution as above, and then the $M$ items to insert are drawn uniformly from $\xB_t$. Next, the remaining $|S_t|-M$ uniformly sampled items from $U_{t-1}$ are generated by uniformly subsampling from $S_{t-1}$. This works because a uniform subsample of a uniform sample is itself uniform.
%Note that an out-of-order item can simply be treated as if it just arrived, and the resulting sample will be statistically the sample as the sample that would have resulted from an insertion at the correct time.

\section{Proofs}\label{sec:proofs}

\textbf{Proof of Theorem~\ref{th:recurr}} Let the random variable $B$ have the common distribution of $\{B_t\}_{t\ge 1}$ and write $p_B(k)=\prob{B=k}$ for $k\ge 0$. Observe that $\{C_t\}_{t\ge 1}$ is a time-homogeneous irreducible aperiodic Markov chain on the nonnegative integers~\cite{bremaud99} with state-transition probabilities given for $i,j\ge 0$ by
\[
p_{ij}=\prob{C_{t+1}=j\mid C_t=i}=\sum_{j\ge 0}p_B(k)\prob{\Delta_{i,k}=j-i}
\]
or, more compactly,  $p_{ij}=\prob{\Delta_{i,B}=j-i}$. Here each $\Delta_{i,k}$ is independent of $B$ and is distributed as the difference of independent binomial$(k,q)$ and binomial$\bigl(i,(1-p)\bigr)$ random variables. Thus the assertion in (i) is equivalent to the assertion that the chain $\{C_t\}_{t\ge 1}$ is positive recurrent; see~\cite[Thm.~2.7.3 and Chap.~3]{bremaud99}. To prove positive recurrence, define functions $V(i)=(i-n)^2$ and $h(i)=ip(1-p)-(i-n)^2(1-p^2)+n(1-p)(1-q)+\sigma_B^2q^2$ for $i\ge 0$, and fix a positive constant $\epsilon$. Here $\sigma^2_B$ denotes the common variance of the batch sizes; by assumption, $\sigma^2_B<\infty$. Then $\lim_{i\to\infty}h(i)=-\infty$, so that there exists an integer $i_\epsilon$ such that $h(i)<-\epsilon$ for all $i>i_\epsilon$. Next, define the finite set $F=\{\,i : i\le i_\epsilon\,\}$. Using basic properties of the binomial distribution together with the fact that $bq=n(1-p)$, we find that, for any $t\ge 1$,
\begin{equation}\label{eq:foster1}
\begin{split}
&\max_{i\in F}\mean[V(C_{t+1})\mid C_t=i]
= \max_{i\in F} \mean[(i+\Delta_{i,B_t}-n)^2]\\
&\quad= \max_{i\in F}\bigl((i-n)^2 + h(i)\bigr) <\infty
\end{split}
\end{equation}
and
\begin{equation}\label{eq:foster2}
\begin{split}
&\mean[V(C_{t+1})-V(C_t)\mid C_t=i]\\
&\quad =\mean[(i+\Delta_{i,B_t}-n)^2-(i-n)^2] = h(i) <-\epsilon
\end{split}
\end{equation}
for $i\not\in F$.
The assertion in~(i) now follows from \eqref{eq:foster1}, \eqref{eq:foster2}, and Foster's Theorem; see, e.g., \cite[Thm.~5.1.1]{bremaud99}.

To prove the assertion in~(ii), observe that
\[
\begin{split}
&\mean[C_t] = \mean\bigl[\mean[C_t\mid C_{t-1},B_t]\bigl]\\
&\quad=\mean[C_{t-1}+qB_t-(1-p)C_{t-1}) ]
%&\quad
=p\mean[C_{t-1}] + n(1-p).
\end{split}
\]
Iterating the above calculation yields assertion~(ii), since $\mean[C_0]=C_0$ by assumption. A similar computation, but for $\mean[C_t^2]$, shows that 
\begin{equation}\label{eq:varC}
\var[C_t]=\alpha n+\sigma^2_Bq^2/(1-p^2)+O(p^t)
\end{equation}
for $t>0$, where $\alpha=(1+p-q)/(1+p)$ and $\sigma^2_B$ is the common variance of the $B_t$'s.

To establish (iii), observe that, by (i), the chain $\{C_t\}_{t\ge 0}$ is ergodic, and so has a stationary distribution $\pi$ \cite[Thm.~3.3.1]{bremaud99}. This distribution is also a limiting distribution of the chain; in other words, $C_t\Rightarrow C_\infty$ for any fixed initial state, where $\Rightarrow$ denotes convergence in distribution and $C_\infty$  has distribution $\pi$ \cite[Thm.~4.2.1]{bremaud99}. Moreover, \eqref{eq:varC} implies that $\sup_t\mean[C_t^2]<\infty$, so that $\{C_t\}_{t\ge 0}$ is uniformly integrable and thus $\mean[C_\infty]=\lim_{t\to\infty}E[C_t]=n$~\cite[p.~338]{Billingsley95}. Finally, by the strong law of large numbers for Markov chains---see, e.g., \cite[Thm.~3.4.1]{bremaud99}, $\lim_{t\to\infty}(1/t)\sum_{i=0}^t C_i=E[C_\infty]=n$ w.p.1.

To prove (iv), we actually prove a stronger result. Recall that the \emph{cumulant generating function (cgf)} of a random variable $X$ is defined by $K_X(s)=\log\mean[e^{sX}]$ for all real~$s$ such that the right side is finite. Denote by $K_B$ the common cgf of the $B_t$'s and set
\[
g^+_t(n,\eps)=\min_{s>1}K_B\bigl(n(s-1)/b\bigr)+p^tn(s-1)-(1+\eps)n\ln s
\]
and
\[
g^-_t(n,\eps)=\min_{s\in(0,1)}K_B\bigl(n(s-1)(1-p^t)/b\bigr)-(1-\eps)n\ln s.
\]
We now show that
\begin{equation}\label{eq:ineqa}
\prob{C_t \le (1-\eps)n}\le e^{g^-_t(n,\eps)}
\end{equation}
for $\eps,t>0$ and
\begin{equation}\label{eq:ineqb} 
\prob{C_t \le (1-\eps)n}\le e^{g^-_t(n,\eps)}
\end{equation}
for $\eps\in [0,1]$ and $t\ge \ln\eps/\ln p$, whether or not the batch size distribution has finite support. Then (v) follows after noting that if $\prob{B_t\le \beta}=1$, then $K_B(s)\le \beta s$ for $s>0$. To prove \eqref{eq:ineqa} and  \eqref{eq:ineqb}, we require the following technical lemma.

\begin{lemma}\label{lem:ineq}
Let $X$ be a nonnegative random variable, $k$ a positive integer, and $a_1,a_2,\ldots,a_k$ a set of constants such that either (i) $a_j\in (0,1)$ for $1\le j\le k$ or (ii) $a_j>1$ for $1\le j\le k$, and set $G(s)=\mean[s^X]$ for $s> 0$. Then 
\[
\prod_{j=1}^k G(a_j)\le G\Bigl(\prod_{j=1}^k a_j\Bigl).
\]
\end{lemma}
\begin{proof}
For $1\le j\le k$, set $\alpha_j=\ln a_j$ so that $\ln G(a_j)=\ln \mean[e^{\alpha_jX}]=K_X(\alpha_j)$, where $K_X$ is the cgf of $X$. The function $K_X$, being a cgf, is convex~\cite[p.~106]{BoydV04}. Because $K_X(0)=0$, it follows that $K_X$ is superadditive---see~\cite[Th.~5]{Bruckner62} for case~(ii), with the same argument holding for case~(i).
%\footnote{The proof in  \cite{Bruckner62} applies to case (ii), but the same argument applies to case (i).}
Thus
\[
\sum_{j=1}^k \ln G(a_j)=\sum_{j=1}^k K_X(\alpha_j)\le K_X\Bigl(\sum_{j=1}^k \alpha_j\Bigr)
=\ln G\Bigl(\prod_{j=1}^k a_j\Bigr).
\]
Exponentiation now yields the desired result.
\end{proof}

If $X$ is a binomial$(m,r)$ random variable, then a standard calculation shows that $\mean[s^{\theta X}]=(rs^\theta+1-r)^m$ for $s>0$ and real $\theta$.
%we have, using the binomial theorem,
%\[
%\begin{split}
%\mean[s^{\theta X}] &=\sum_{j=0}^m s^{j\theta}\binom{m}{j}r^j(1-r)^{m-j}\\
%&= \sum_{j=0}^m \binom{m}{j}(rs^{\theta})^j(1-r)^{m-j} =(rs^\theta+1-r)^m.
%\end{split}
%\]
It follows that, w.p.1,
\[
\mean[s^{\Delta_{i,B}}\mid B]=(qs+1-q)^B\bigl((1-p)/s+p\bigr)^i
\]
for $s>0$, where $\Delta_{i,k}$ and $B$ are defined as before. Set $m_t(u)=\mean[u^{C_t}]$ and $G(u)=\mean[u^B]$ for $u\ge 1$. Now fix $t\ge 0$ and $s>1$, and observe that, since $B_t$ and $C_{t-1}$ are independent,
\[
\begin{split}
m_t(s)
&=\mean\bigl[\mean[s^{C_t}\mid C_{t-1},B_t]\bigr]\\
&=\mean\bigl[\mean[s^{C_{t-1}+\Delta_{C_{t-1},B_t}}\mid C_{t-1},B_t]\bigr]\\
&=\mean\bigl[s^{C_{t-1}}(qs+1-q)^{B_t}\bigl((1-p)/s+p\bigr)^{C_{t-1}}\bigr]\\
&= G(qs+1-q)m_{t-1}\,\bigl((1-p)+ps)\bigr).
\end{split}
\]
Iterating and using the fact that $m_0(s)=s^n$ by assumption, we have
\begin{equation}\label{eq:mts}
m_t(s)=\bigl(p^t(s-1)+1\bigr)^n  \theta_t(s),
\end{equation}
where
\begin{equation}\label{eq:thetabd}
\theta_t(s)=\prod_{j=0}^{t-1}G\bigl(qp^j(s-1)+1\bigr)\le G\bigl(z(s)\bigr),
\end{equation}
with $z(s)=\prod_{j=0}^{t-1}\bigl(qp^j(s-1)+1\bigr)$; the inequality follows from Lemma~\ref{lem:ineq}, case~(ii). Next, observe that
\begin{equation}\label{eq:lnz}
\begin{split}
&\ln z(s)\le \sum_{j=0}^\infty\ln\bigl(qp^j(s-1)+1\bigr)
\le \sum_{j=0}^\infty qp^j(s-1)\\
&\quad=q(s-1)/(1-p)
=n(s-1)/b,
\end{split}
\end{equation}
where we have used both the fact that $\ln(1+x)\le x$ for all $x\ge -1$ and the identity $bq=n(1-p)$. By \eqref{eq:mts}--\eqref{eq:lnz}, we have
\[
\begin{split}
&m_t(s)\le\exp\Bigl(K_B\bigl(n(s-1)/b\bigr)+n\ln\bigl(1+p^t(s-1)\bigr)\Bigr)\\
&\ \le\exp\Bigl(K_B\bigl(n(s-1)/b\bigr)+p^tn(s-1)\Bigr),
\end{split}
\]
and, using Markov's inequality~\cite[p.~80]{Billingsley95} , we have
\[
\begin{split}
&\prob{C_t\ge (1+\eps) n}=\prob{s^{C_t}\ge s^{(1+\eps) n}}
\le m_t(s)/s^{(1+\eps)n}\\
&\ \le \exp\Bigl(K_B\bigl(n(s-1)/b\bigr)+p^tn(s-1)-(1+\eps)n\ln s\Bigr).
\end{split}
\]
Minimizing the right side of the above inequality with respect to $s$ yields \eqref{eq:ineqa}. The proof for \eqref{eq:ineqb} is similar, and uses case~(i) of Lemma~\ref{lem:ineq}.

%The argument for (v)(b) is similar. Again fix $s>1$ and observe that, again using Markov's inequality,
%\[
%\begin{split}
%\prob{C_t\le (1-\eps) n}
%&= \prob{s^{-C_t}\ge s^{-(1-\eps) n}}\\
%&\le \mean[s^{-C_t}]/s^{-(1-\eps) n}
%=s^{(1-\eps)n}m_t(s^{-1}).
%\end{split}
%\]
%By our previous calculations, with $s$ replaced by $s^{-1}$ and using case~(i) of Lemma~\ref{lem:ineq},
%\[
%\ln m_t(s^{-1})\le \ln \theta_t(s^{-1})\le \ln G\bigl(z(s^{-1})\bigr).
%\]
%Observe that
%\[
%\ln z(s^{-1})=\sum_{j=0}^{t-1}\ln\bigl(qp^j(s^{-1}-1)+1\bigr)\le n(s-1)(1-p^t)/b,
%\]
%so that
%\begin{equation}\label{eq:lbd}
%\begin{split}
%&\prob{C_t\le (1-\eps) n}\\
%&\ \le \exp\Bigl(K_B\bigl(n(s^{-1}-1)(1-p^t)/b\bigr)-(1-\eps)\ln s^{-1}\Bigr).
%\end{split}
%\end{equation}
%Minimizing the right side of \eqref{eq:lbd} with respect to $s>0$ is equivalent to minimizing $g^-(n,\eps)$ with respect to $s\in(0,1)$, and (v)(b) follows.

%\section{Proof of Theorem~\ref{th:downsamp}}\label{sec:th2proof}
\vskip\baselineskip
\textbf{Proof of Theorem~\ref{th:downsamp}} We first assume that $\pi=\{i^*\}$, so that there exists a partial item in $L$, and prove the result for $i=i^*$ and then for $i\not= i^*$. We then prove the result when $\pi=\emptyset$. 

\textit{Proof for $i=i^*$}: Observe that when $\pi=\{i^*\}$, we have $\prob{i^*\in S}=\frc(C)$. First suppose that $\floor{C'}=0$, so that $\frc(C')=C'$. Either the partial item $i^*$ is swapped and ejected in lines~\ref{ln:swap0} and \ref{ln:killA} or is retained as a partial item: $\pi'=\{i^*\}$. Thus
\[
\begin{split}
&\prob{i^*\in S'}=\prob{i^*\in S'\mid i^*\in L'}\prob{i^*\in L'}\\
&\quad=\frc(C')\prob{\text{no swap}}=\frc(C')\bigl(\frc(C)/C\bigr)\\
&\quad=(C'/C)\frc(C)=(C'/C)\prob{i^*\in S}.
\end{split}
\]
Next suppose that $0<\floor{C'}=\floor{C}$. Then the partial item may or may not be converted to a full item via the swap in line~\ref{ln:convert}. Denoting by
$r=\bigl(1-(C'/C)\frc(C)\bigr)/\bigl(1-\frc(C')\bigr)$ the probability that this swap does not occur, we have
\[
\begin{split}
&\prob{i^*\in S'}=\prob{i^*\in S'\mid i^*\in\pi'}\prob{i^*\in\pi'}\\
&\hskip0.8in+\prob{i^*\in S'\mid i^*\not\in\pi'}\prob{i^*\not\in\pi'}\\
&\quad=\frc(C')\cdot\prob{\text{no swap}}+1\cdot \prob{\text{swap}}\\
&\quad=1-r\bigl(1-\frc(C')\bigr)\\
&\quad=(C'/C)\frc(C)=(C'/C)\prob{i^*\in S}.
\end{split}
\]
Finally, suppose that $\floor{C'}<\floor{C}$. Either the partial item $i^*$ is swapped into $A$ in line~\ref{ln:swap} or ejected in line~\ref{ln:move}. Thus
\[
\begin{split}
\prob{i^*\in S'}&=\prob{\text{swap}}\\
&=(C'/C)\frc(C)=(C'/C)\prob{i^*\in S},
\end{split}
\]
establishing the assertion of the lemma for $i=i^*$ when the partial item~$i^*$ exists.

\textit{Proof for $i\not=i^*$:} Still assuming the existence of $i^*$, set $Y_j=1$ if item~$j$ belongs to $S'$ and $Y_j=0$ otherwise. Also set $p_j=\prob{j\in S'}=\mean[Y_j]$. Since all full items in $S$ are treated identically, we have $p_j\equiv p$ for  $j\in A$, and
\[
\mean[|S'|]=\mean\Bigl[\sum_{j\in A}Y_j+Y_{i^*}\Bigr]=\sum_{j\in A}\mean[Y_j]+\mean[Y_{i^*}]=\floor{C}p+p_{i^*}
\]
so that, using \eqref{eq:meansize}, 
\[
\begin{split}
&\prob{j\in S'}\\
&\ =(\mean[|S'|]-p_{i^*})/\floor{C}=\bigl(C'-(C'/C)\frc(C)\bigr)/\floor{C}\\
&\ =(C'/C)\bigl(C-\frc(C)\bigr)/\floor{C}=C'/C=(C'/C)\prob{j\in S}
\end{split}
\]
for any full item $j\in A$.

\textit{Proof when $\pi=\emptyset$:} We conclude the proof by observing that, if $\pi=\emptyset$, then
$
C'=\mean[|S'|]=\sum_{j\in A}p_j=\floor{C}p=Cp
$
and again $\prob{j\in S'}=C'/C=(C'/C)\prob{j\in S}$.

%\section{Proof of Theorem~\ref{th:rtbsIncl}}\label{sec:th3proof}

\vskip\baselineskip
\textbf{Proof of Theorem~\ref{th:rtbsIncl}}
The proof of \eqref{eq:inclRTBS} is by induction on $t$, and the various steps are given below.

\textit{Base case:} We start with $t=1$. If $B_1>n$, then, in lines~\ref{ln:uUpdateW} and \ref{ln:uUpdateA}, R-TBS forms an initial latent sample $L_0=(\xB_1,\emptyset,B_1)$ and then runs \textsc{Dsample} to bring the sample weight down to $n$ (line~\ref{ln:overshoot}). If we apply \eqref{eq:getsample} to $L_0$ to obtain a sample~$S_0$, then $\prob{i\in L_0}=\prob{i\in S_0}=1$ for $i\in\xB_1$. It now follows from Theorem~\ref{th:downsamp} that $\prob{i\in S_1}=(n/B_1)\prob{i\in S_0}=n/B_1=C_1\bigl(w_1(i)/W_1\bigr)$ since $w_1(i)=1$. Similarly, if $B_1\le n$, then $\prob{i\in S_1}=1=C_1\bigl(w_1(i)/W_1\bigr)$, since $C_1=W_1$ and $w_1(i)=1$.

\textit{Induction, unsaturated case:} Now assume for induction that \eqref{eq:inclRTBS} holds for time $t-1$. Suppose that $W_{t-1}<n$ so that $C_{t-1}=W_{t-1}$. After decaying $W_{t-1}$ to $W'_t$ and downsampling $L_{t-1}$ to $L'_t=(A'_t,\pi'_t,W'_t)$ in lines~\ref{ln:decayW} and \ref{ln:dsample}, the elements of $\xB_t$ are included to form a latent sample $L''_t=(A'\cup\xB_t,\pi'_t,W'_t+B_t)$ in lines~\ref{ln:uUpdateW} and \ref{ln:uUpdateA}. Note that the sample weight of $L'_t$ is $C'_t=W'_t$. If there is no overshoot, then for $i\in\xB_t$ we have $\prob{i\in S_t}=1=C_t\bigl(w_t(i)/W_t\bigr)$ since $C_t=W_t=W'_t+B_t$ and $w_t(i)=1$. For $i\in U_{t-1}$, we have
\[
\begin{split}
\prob{i\in S_t}
&=(W'_t/W_{t-1})\prob{i\in S_{t-1}}\\
&=(W'_t/W_{t-1})C_{t-1}\bigl(w_{t-1}(i)/W_{t-1}\bigr)\\
&=(W'_t/W_{t-1})w_{t-1}(i)=w_t(i)=C_t\bigl(w_t(i)/W_t\bigr),
\end{split}
\]
where the first equality follows from Theorem~\ref{th:downsamp}, the second follows from the induction hypothesis, and the last follows because $C_t=W_t=W'_t+B_t$. If there is an overshoot, then $L''_t$ is downsampled in line~\ref{ln:overshoot}. For $i\in\xB_t$, we have $\prob{i\in S_t}=n/(W'_t+B_t)= C_t/W_t=C_t\bigl(w_t(i)/W_t\bigl)$ by Theorem~\ref{th:downsamp}. Moreover, for $i\in U_{t-1}$, we have $\prob{i\in S_t}=w_t(i)\bigr(n/(W'_t+B_t)\bigr)=C_t\bigl(w_t(i)/W_t\bigr)$, where the first equality follows from our prior calculations plus Theorem~\ref{th:downsamp}, and the second equality follows from the fact that $C_t=n$ and $W_t=W'_t+B_t$.

\textit{Induction, saturated case:} Suppose that $W_{t-1}\ge n$, so that $C_{t-1}=n$. Also suppose that $W_t\ge n$ after decaying the weight $W_{t-1}$ to $W'_t$ and updating the weight to $W_t=W'_t+B_t$ in line~\ref{ln:newWeight}. Then a random number of batch items are inserted into the sample, replacing existing items (lines~\ref{ln:bround} and \ref{ln:sUpdateA}). For $i\in\xB_t$, the mean-preserving property of stochastic rounding implies that
\[
\begin{split}
\prob{i\in S_t}&=\mean\bigl[\prob{i\in S_t\mid m}\bigr]=\mean[m/B_t]=\mean[m]/B_t\\
&=n/W_t=C_t\bigl(w_t(i)/W_t\bigr)
\end{split}
\]
since $C_t=n$ and $w_t(i)=1$. Similarly, for $i\in U_{t-1}$, we have
\[
\begin{split}
\prob{i\in S_t}&=\frac{n-\mean[m]}{n}\left(\frac{nw_{t-1}(i)}{W_{t-1}}\right)=\frac{W_t-B_t}{W_t}\left(\frac{nw_{t-1}(i)}{W_{t-1}}\right)\\
&=\frac{W'_t}{W_t}\left(\frac{nw_{t-1}(i)}{W_{t-1}}\right)=n\frac{w_t(i)}{W_t}=C_t\frac{w_t(i)}{W_t}.
\end{split}
\]
Now suppose that $W_t<n$ after $W_{t-1}$ is updated to $W_t=W'_t+B_t$, so that there is an undershoot. Then $L_{t-1}=(A_{t-1},\emptyset,n)$ is downsampled to form $L_t=(A_t,\pi_t,W'_t)$ in line~\ref{ln:undershoot} and all items in $\xB_t$ are then inserted as full items (line~\ref{ln:fillup}). For $i\in\xB_t$, we have $\prob{i\in S_t}=1=C_t\bigl(w_t(i)/W_t)$ since $C_t=W_t=W'_t+B_t$ and $w_t(i)=1$. For $i\in U_{t-1}$ we have
\[
\begin{split}
\prob{i\in S_t}&=\frac{W_t-B_t}{n}\left(\frac{nw_{t-1}(i)}{W_{t-1}}\right)=\frac{W'_t}{W_{t-1}}w_{t-1}(i)\\
&=w_t(i)=C_t\big(w_t(i)/W_t\bigr)
\end{split}
\]
since $C_t=W_t=W'_t+B_t$. The proof is now complete.

%\vskip\baselineskip
%\textbf{Proof of Theorem~\ref{th:maxMean}}
%By \eqref{eq:meansize}, $\mean[|S_t|]=C_t=W_t$, the latter equality holding since $W_t<n$. Since $H$ satisfies \eqref{eq:expratio}, it follows that, for each time $j\le t$ and $i\in\xB_j$, the inclusion probability $\prob{i\in S^H_t}$ must be of the form $r_te^{-\lambda(t-j)}$ for some function $r_t$ independent of $j$. Taking $j=t$, we see that $r_t\le 1$. It follows that
%\[
%\mean[|S^H_t|]=\sum_{j=1}^t|\xB_j|r_te^{-\lambda(t-j)}\le \sum_{j=1}^t|\xB_j|e^{-\lambda(t-j)}
%%=W_t
%=\mean[|S_t|],
%\]
%proving the result.

%\vskip\baselineskip
%\textbf{Proof of Theorem~\ref{th:minVar}}
%Considering all possible distributions over the sample size having a mean value equal to $C_t$, it is straightforward to show that variance is minimized by concentrating all of the probability mass onto $\floor{C_t}$ and $\ceil{C_t}$. There is precisely one such distribution, namely the stochastic-rounding distribution, and this is precisely the sample-size distribution attained by R-TBS.

\section{Chao's Algorithm}\label{sec:chao}

In this section, we provide pseudocode for a batch-oriented, time-decayed version of Chao's algorithm~\cite{Chao82} for maintaining a weighted reservoir sample of $n$ items, which we call B-Chao. In the algorithm, the function $\textsc{Get1}(x,A)$ randomly chooses an item $i$ in a set $A$, and then sets $x\gets i$ and  $A\gets A\setminus\{x\}$. We explain the function \textsc{Normalize} below.

\begin{algorithm}[t]
\caption{Batched version of Chao's scheme (B-Chao)}\label{alg:chao}
{\footnotesize
$\lambda$: decay factor ($\ge 0$)\;
$n$: reservoir size\;
\BlankLine
Initialize: $S\gets S_0$; $W\gets |S_0|$; $A\gets\emptyset$; $V\gets\emptyset$\Comment*[r]{$|S_0|\le n$}
\For{$t\gets1,2,\ldots$}{
  \Comment{update weights}
  $W\gets e^{-\lambda}W$\Comment*[r]{$W=$ agg. weight of non-overweight items}
  \lFor(\Comment*[f]{$V$ holds overwt items}){$(z,w_z)\in V$}{$w_z\gets e^{-\lambda}w_z$}
  \For{$j\gets 1,2,\ldots,|\xB_t|$}{
    $\textsc{Get1}(x,\xB_t)$ \Comment*[r]{get new item to process}
    \If(\Comment*[f]{reservoir not full yet}){$|S|<n$}{
      $S\gets S\cup \{x\}$; $W\gets W+1$\;
      }
    \Else(\Comment*[f]{reservoir is full}){
      $\textsc{Normalize}(x,V,A,W,\pi_x)$ \Comment*[r]{categorize items}
     \If{$\textsc{\emph{Uniform}}()\le \pi_x$}{
        \Comment{accept $x$ and choose victim to eject}
        $\alpha=0$; $y\gets$ {\bf null}; $U\gets\textsc{Uniform}()$\;
        \For(\Comment*[f]{attempt to choose from $A$}){$(z,w_z)\in A$}{
          $\alpha\gets\alpha + \bigl(1-\frac{(n-|V|)w_z}{W}\bigr)/\pi_x$\;
          \If{$U\le \alpha$}{$A\gets A\setminus\{(z,w_z)\}$; $y\gets z$; {\bf break}}
          }
        \lIf(\Comment*[f]{remove vic.$\in S$}){$y==$ {\bf null}}{$\textsc{Get1}(y,S)$}
        \lIf{$(x,1)\notin V$}{$S\gets S\cup\{x\}$}
        }
      \Comment{if no longer overweight, stop tracking}
      $S\gets S\cup \{z:(z,w_z)\in A\}$; $A\gets\emptyset$
      }
    }
  output $S\cup \{z:(z,w_z)\in V\}$
  }
}
\end{algorithm}

\begin{algorithm}[t]
\caption{Normalization of appearance probabilities}\label{alg:normp}
{\footnotesize
$x$: newly arrived item (has weight $=1$)\;
$V$: set of items that remain overweight (and their weights)\;
$A$: set of items that become non-overweight (and their weights)\;
$W$: aggregate weight of non-overweight items\;
$\pi_x$: inclusion probability for $x$\;
$n$: reservoir size\;
\BlankLine
$W\gets W+1+\sum_{(z,w_z)\in V}w_z$ \Comment*[r]{agg. wt. of new \& sample items}
\If(\Comment*[f]{$x$ is not overweight}){$n/W\le 1$}{
  $A\gets V$; $V\gets\emptyset$ \Comment*[r]{no item is now overweight}
  $\pi_x\gets n/W$
  }
\Else(\Comment*[f]{$x$ is overweight}){
  $\pi_x\gets 1$;
  $W\gets W-1$\;
  $D\gets \{(x,1)\}$ \Comment*[r]{$D=$ set of overweight items so far}
  \Repeat{$(n-|D|)w_z/W\le1$\Comment*[f]{first non-overweight item}}{
    $(z,w_z)\gets\textsc{GetMax}(V)$\;
    \If(\Comment*[f]{$z$ remains overweight}){$(n-|D|)w_z/W>1$}{
      $D\gets D\cup\{(z,w_z)\}$; $W\gets W-w_z$
      }
    \Else(\Comment*[f]{$z$ no longer overweight}){
      $A\gets A\cup \{(z,w_z)\}$}
      }
    $A\gets A\cup V$; $V\gets D$ \Comment*[r]{no more overweight items in $V$}
  }
}
\end{algorithm}

Note that the sample size increases to~$n$ and remains there, regardless of the decay rate. During the initial period in which the sample size is less than $n$, arriving items are included with probability~1; if more than one batch arrives before the sample fills up, then clearly the relative inclusion property in \eqref{eq:expratio} will be violated since all items will appear with the same probability even though the later items should be more likely to appear. Put another way, the weights on the first $n$~items are all forced to equal~1.

After the sample fills up, B-Chao encounters additional technical issues due to ``overweight'' items. In more detail, observe that $\mean[|S|]=\sum_{i\in S}\pi_i$, where $\pi_i=P[i\in S]$. At any given moment we require that $\mean[|S|]=\sum_{i\in S}\pi_i=n$. If we also require for each~$i$ that $\pi_i\propto w_i$, then we must have  $\pi_i=n w_i/W$, where as usual $W=\sum_{i\in S}w_i$. It is possible, however, that $w_i/W>1/n$, and hence $\pi_i>1$, for one or more items $i\in S$. Such items are called \emph{overweight}. As in \cite{Chao82}, B-Chao handles this by retaining the most overweight item, say $i$,  in the sample with probability~1. The algorithm then looks at the reduced sample of size $n-1$ and weight $W-w_i$, and identifies the item, say $j$, having the largest weight $w_j$. If item~$j$ is overweight in that the modified relative weight $w_j/(W-w_i)$ exceeds $1/(n-1)$, then it is included in the sample with probability~1 and the sample is again reduced. This process continues until there are no more overweight items, and can be viewed as a method for categorizing items as overweight or not, as well as normalizing the appearance probabilities to all be less than 1. The \textsc{Normalize} function in Algorithm~\ref{alg:chao} carries out this procedure; Algorithm~\ref{alg:normp} gives the pseudocode. In Algorithm~\ref{alg:normp}, the function $\textsc{GetMax}(V)$ returns the pair $(z,w_z)\in V$ having the maximum value of $w_z$ and also sets $V\gets V\setminus\{(z,w_z)\}$; ties are broken arbitrarily. An efficient implemmentation would represent $V$ as a priority queue.

When overweight items are present, it is impossible to both maintain a sample size equal to $n$ and to maintain the property in \eqref{eq:expratio}. Thus, as discussed in Section~2.1 of \cite{Chao82}, the algorithm only enforces the relationship in \eqref{eq:expratio} for items that are not overweight. When the decay rate $\lambda$ is high, newly arriving items are typically overweight, and transform into non-overweight items over time due to the arrival of subsequent items. In this setting, recently-arrived items are overrepresented. The R-TBS algorithm, by allowing the sample size to decrease, avoids the over\-weight-item problem, and thus the violation of the relative inclusion property \eqref{eq:expratio}, as well as the complexity arising from the need to track overweight items and their individual weights (as is done in the pseudocode via $V$). We note that prior published descriptions of Chao's algorithm tend to mask the complexity and cost incurred by the handling of overweight items; R-TBS is lightweight compared to B-Chao.

%% file: spark-implementation.tex
% !TEX root = tbs-mm.tex
% above command is for TeXShop

\section{Implementation of D-R-TBS on Spark}\label{sec:spark-impl}

Spark is a natural platform for such implementations because it supports streaming, machine learning, and efficient distributed data processing, and is widely used.  Efficient implementation is relatively straightforward for T-TBS but decidedly nontrivial for R-TBS because of both Spark's idiosyncrasies and the coordination needed between nodes. The key is to leverage the in-place updating technique for RDDs introduced in \cite{XieTSBH15}.
%We have chosen to implement our distributed versions of the R-TBS algorithm on Spark in order to leverage Spark's streaming support for handling large volumes of data and its powerful distributed processing mechanism.

\subsection{Spark Overview}
Spark is a general-purpose distributed processing framework based on a functional programming paradigm. Spark provides a distributed memory abstraction called a Resilient Distributed Dataset (RDD). An RDD is divided into partitions that are then distributed across the cluster for parallel processing. RDDs can either reside in the aggregate main memory of the cluster or in efficiently serialized disk blocks. An RDD is immutable and cannot be modified, but a new RDD can be constructed by transforming an existing RDD. Spark utilizes both lineage tracking and checkpointing of RDDs for fault tolerance. A Spark program consists of a single driver and many executors. The driver of a Spark program orchestrates the control flow of an application, while the executors perform operations on the RDDs, creating new RDDs.

\subsection{Distributed Data Structures}

Since we leverage Spark Streaming for data ingestion, the incoming batch $\xB_t$ at time $t$ is naturally stored as an RDD. For the reservoir, we have two alternatives: key-value store and co-partitioned reservoir, as discussed in Section~\ref{sec:reservoir}. The integration of a non-native key-value store with Spark usually incurs extra overhead. For example, to apply a Spark ML algorithm to the sample requires exporting all items in the key-value store into an RDD. Furthermore, when there is a failure, computation has to restart from the last checkpoint to ensure the consistency of the reservoir. We now describe how we implement the co-partitioned reservoir in Spark.

%This is done through Spark's accumulator variable: each worker collects the number of items in its local partition of $\xB_t$, and then the Spark driver accumulates the total number. 
%In both centralized and distributed decision making, a common subroutine used is to select $m$ random elements from a set of size $n$. An efficient way to do this is to perform a variation of the Fisher-Yates shuffle \cite{Durstenfeld64} to create a random permutation of $\min(m, n-m)$ numbers (when $n-m < m$, we generate the $n-m$ elements not chosen).  

%Since we leverage Spark Streaming for data ingestion, $\xB_t$ is naturally stored as an RDD, which is distributed into partitions. The slot number of an item $i\in\xB_t$ thus maps to a specific partition ID and a position inside the partition.

%We leverage Spark Streaming
%to ingest batches of arriving data, thereby supporting input sources such as HDFS, Kafka, Flume, and so on. Each batch of new data is stored in a Resilient Distributed Dataset (RDD) in Spark. Every Spark program consists of a single driver and many executors. An RDD is broken down into partitions and distributed among the executors. 

\subsubsection{Co-partitioned Reservoir Implementation in Spark}

We would like to utilize the distributed fault-tolerant data structure, RDD, in Spark for the co-partitioned reservoir. However, storing the co-partitioned reservoir as a vanilla RDD also is not ideal. Because RDDs are immutable, the large numbers of reservoir inserts and deletes at each time point would trigger the constant creation of new RDDs, quickly saturating memory. 

An alternative approach employs the RDD with the in-place update technique in~\cite{XieTSBH15} to ensure that partitions for incoming batches coincide with reservoir partitions. The key idea is to share objects across different RDDs. In particular, we store the reservoir as an RDD, each partition of which contains only one object, a (mutable) vector containing the items in the corresponding reservoir partition. A new RDD created from an old RDD via a batch of inserts and deletes references the same vector objects as the old RDD. We keep the lineage of RDDs intact by notifying Spark of changes to old RDDs (by calling the \textsc{Unpersist} function), so that in case of failure, old RDDs (with old samples) can be recovered from checkpoints, and Spark's recovery mechanism based on lineage will regenerate the sample at the point of failure.

\subsection{Choosing Items to Delete and Insert}

Section~\ref{sec:updates} has detailed the centralized and distributed decisions for choosing items to delete and insert. Here, we add some Spark-related details for the centralized decisions.

All of the transient large data structures are stored as RDDs in Spark; these include the set of item locations for the insert items $\mathcal{Q}$, the set of retrieved insert items $\mathcal{S}$, and the set of item locations for the delete items $\mathcal{R}$. To ensure the co-partitioning of these RDDs with the incoming batch RDD (and the reservoir RDD when the co-partitioned reservoir is used), we use a customized partitioner. For the join operations between RDDs, we use by default a standard repartition-based join in Spark. However, when RDDs are co-partitioned and co-located, we implement a customized join algorithm that performs only local joins on corresponding partitions.

%% file: appendix-expmts.tex
\section{Extra Experiments}\label{sec:more-expmts}

\begin{figure}[tbh]
 \centering
	 \subfigure[Periodic (20, 10)]{
	   \label{fig:per2010}\includegraphics[width=0.45\linewidth]{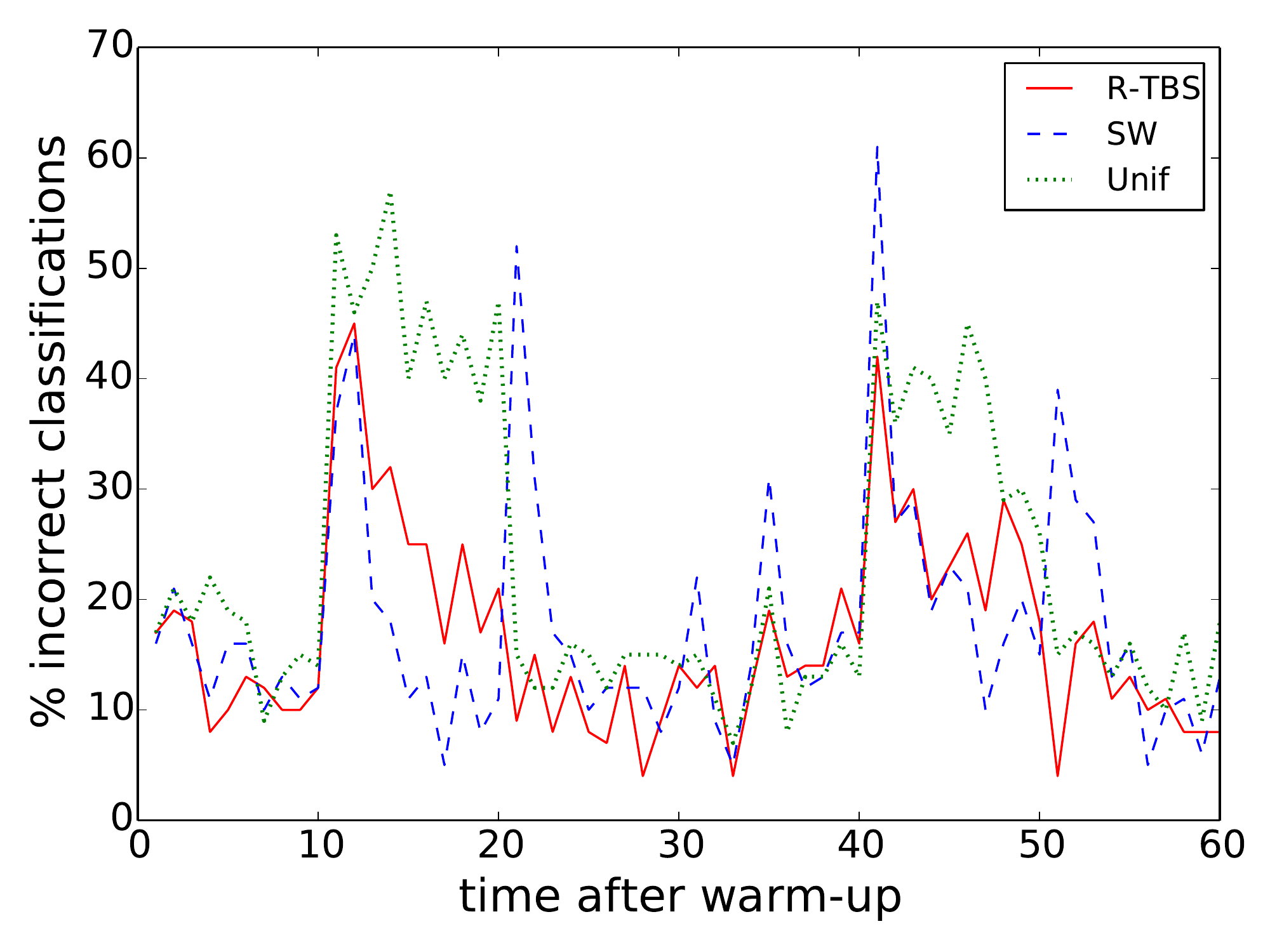}}
	 \subfigure[Periodic (30, 10)]{
	   \label{fig:per3010}\includegraphics[width=0.45\linewidth]{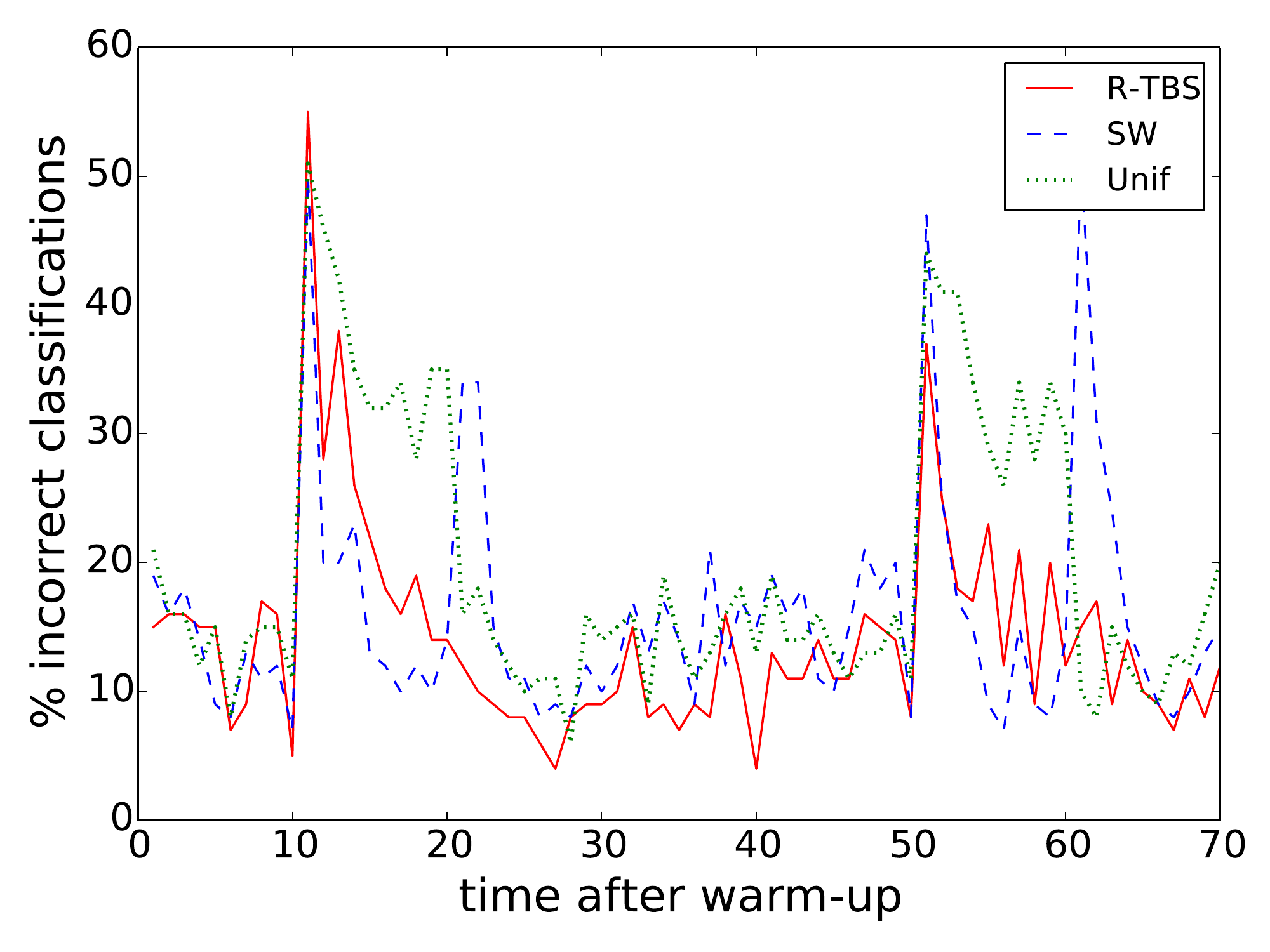}}
	\BigCrunch
  \caption{Misclassification rate (percent) for kNN}
  \BigCrunch
\end{figure}